\documentclass[reqno,12pt,letterpaper]{amsart}
\usepackage[proof]{newmzmacros2}
\usepackage[all,cmtip]{xy}

\usepackage[normalem]{ulem}

\newcommand{\neigh}{{\rm neigh}}

\newcommand{\CC}{{\mathbb C}}

\DeclareMathOperator{\ran}{ran}
\usepackage{soul}

\DeclareMathOperator{\Wr}{Wr}

\title[WKB structure in a scalar model of flat bands]{WKB structure in a scalar model of flat bands}

\author{Semyon Dyatlov}
\email{dyatlov@math.mit.edu}
\address{Department of Mathematics, Massachusetts Institute of Technology}

\author{Henry Zeng} 
\email{henryzeng@berkeley.edu}
\address{Department of Mathematics, University of California, Berkeley, CA 94720}

\author{Maciej Zworski} 
\email{zworski@berkeley.edu}
\address{Department of Mathematics, University of California, Berkeley, CA 94720}

\begin{document}

\begin{abstract}
  We consider a family of periodic scalar operators for which one can define {\em flat bands} in the sense of Floquet--Bloch
theory. One puzzling question originating in recent physics literature is a quantisation rule for the values of parameters at which these flat bands occur. We present a general theorem about the structure of solutions to the corresponding equation
and a heuristic argument explaining their WKB structure in a specific case. That structure also explains the quantisation condition -- both the WKB structure and that rule are confirmed by  numerical experiments. Finally, we consider a simplified model in which separation of variables allows the use of complex WKB methods. 
\end{abstract}

\maketitle

\section{Introduction}
\label{s:int}

Topologically nontrivial flat bands (in the sense of Bloch--Floquet theory) are of interest in condensed matter physics as
their presence has remarkable experimental consequences -- see \cite{BM11}, \cite{magic} for theoretical references and 
\cite{CFF18} for experimental results in the context of twisted bilayer graphene.

This paper is concerned with flat bands and the structure of {solutions to 
$ P u = 0 $ for a class of periodic operators, $ P $,} on $ \mathbb R^2 $.
They are of the form 
\begin{equation}
\label{eq:scalar0}   P ( \alpha , k ) =   ( 2 D_{\bar z } + k )^2 - \alpha^2 V ( z  ) ,   \ \  2 D_{\bar z } = \tfrac{1}i ( \partial_{x_1}  + i \partial_{x_2}) ,  \end{equation}
where $  x = ( x_1 , x_2 )$, $  \alpha, k \in \mathbb C $, and $ V $ is a real analytic, $ \Lambda$-periodic  potential 
satisfying additional symmetry assumptions -- see \eqref{eq:defV}. {The study of solutions to $ P ( \alpha, k ) u = 0 $ is equivalent to the study 
of eigenfunctions} of a system 
\begin{equation}
\label{eq:syst0} 
\begin{gathered}   P ( \alpha, k ) u_0 = 0 \ \Longleftrightarrow \  ( D_{\rm{S}} ( \alpha ) + k ) u = 0 , 
\ \ \ \alpha \neq 0 , \\  D_{\rm{S}}  (\alpha ) := 
\begin{pmatrix} 2 D_{\bar z } & \alpha V ( z   ) \\
  \alpha & 2 D_{\bar z } \end{pmatrix} , \ \ \ u = \begin{pmatrix}  \alpha^{-1} ( 2 D_{\bar z } + k )
  u_0 \\ \ \ -   u_0 \end{pmatrix} . 
\end{gathered} \end{equation} 
Flat bands (see \eqref{eq:defEjk} below) occur at $ \alpha $ if 
\begin{equation}
\label{eq:defflat}   \Spec_{ L^2 ( \mathbb C/\Lambda ) } D_{\rm{S}} ( \alpha ) = \mathbb C . \end{equation}
{When $ V $ satisfies certain symmetries, see \eqref{eq:defV} below, 
we have nontrivial solutions to} 
\begin{equation}
\label{eq:defpro}  D_{\rm{S}} ( \alpha ) v ( \alpha ) = 0 , \end{equation}
for all $ \alpha $. They are called {\em (symmetry) protected states} as their existence is due to symmetries of the operator.
 We are interested in the distribution of $ \alpha $'s for which \eqref{eq:defflat} holds and in the 
closely related WKB structure (as $ \alpha \to \infty $) of the protected states. 

\begin{figure}
\includegraphics[width=10cm]{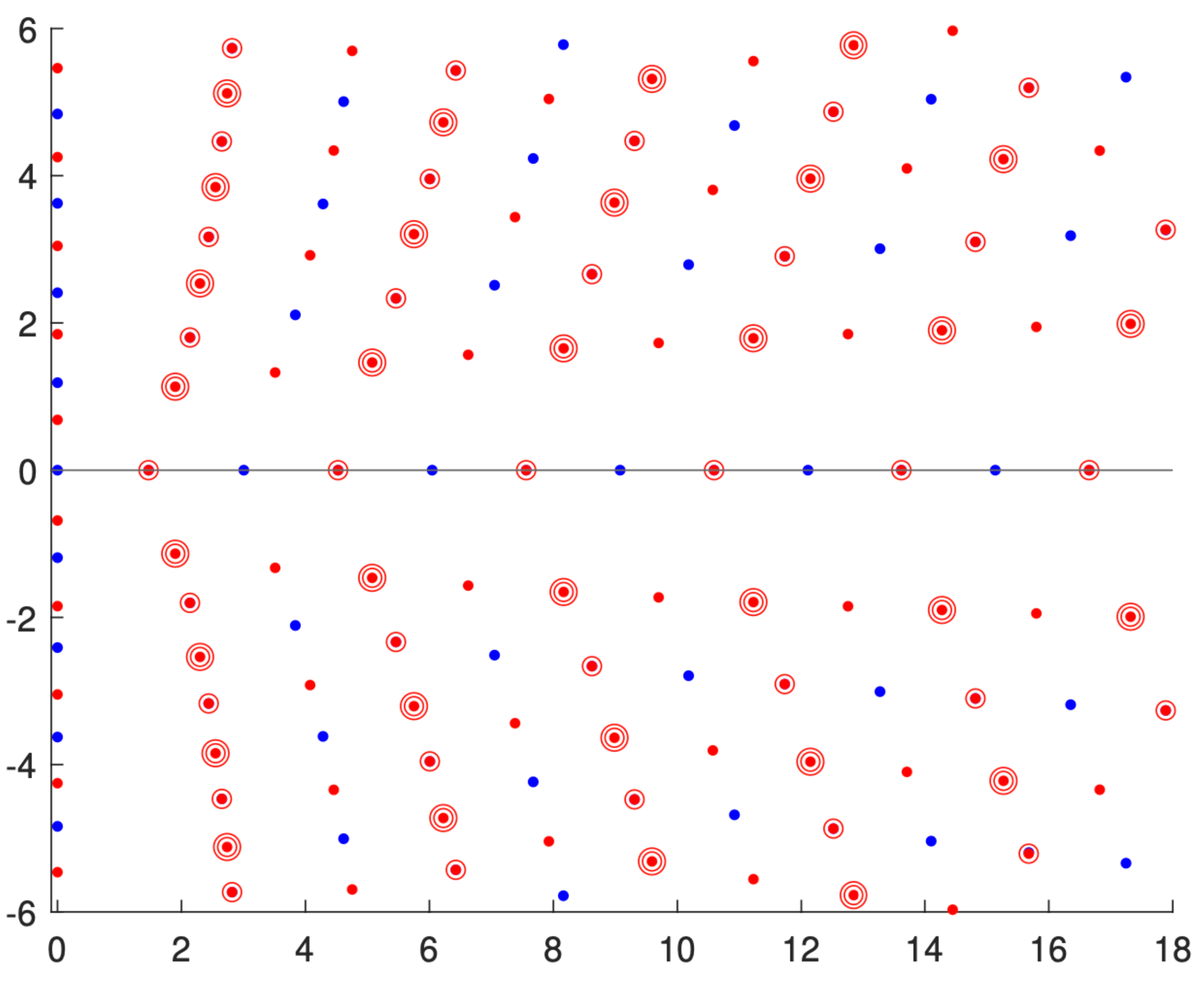}
\caption{\label{f:rot} The sets  $ \mathcal A $ ({\color{red} red}) and $  {\mathcal B } $ ({\color{blue} blue}) of Theorem \ref{t:1}  for the scalar model based on the 
Bistritzer--MacDonald potential \eqref{eq:defU2}. The multiple circles indicate multiplicity (in the sense of Theorem \ref{t:1}, $  m=1,2,3 $ all occur).
A movie showing the evolution of the bands $ k \mapsto E_j ( \alpha, k ) $, $ j =1, 2 $ (over the real axis in $ k $) 
as $ \alpha $ varies can be see at \url{https://math.berkeley.edu/~zworski/bands_scalar.mp4}. It illustrates the behaviour
of bands at real elements of $ \mathcal A \cup \mathcal B $.
}
\end{figure}

Because the properties of $ D_S ( \alpha ) $ in \eqref{eq:syst0} are described using a scalar operator 
\eqref{eq:scalar0}, we refer to this model as {\em a scalar model} of magic angles/flat bands. It was introduced in \cite{gaz4} with the hope of having a simpler setting for seeing the distribution of magic angles.
That model is derived from the the chiral model of twisted bilayer graphene (TBG) \cite{magic}, \cite{beta}, 
\cite{bhz2}, see \S \ref{s:chiral}.

The self-adjoint Hamiltonian associated to $ D_{\rm{S}} ( \alpha ) $ is given by 
\[   H ( \alpha ) := \begin{pmatrix} 0 & D_{\rm{S}} ( \alpha ) \\
D_{\rm{S}} ( \alpha )^* & 0 \end{pmatrix} , \]
and the bands are the eigenvalues of $ e^{ - i \langle z , k \rangle } H ( \alpha ) e^{ i \langle z, k \rangle } :
H^1 ( \mathbb C /\Lambda ; \mathbb C^4  ) \to L^2 ( \mathbb C/\Lambda, \mathbb C^4  ) $, 
$ z , k \in \mathbb C $, considered as functions of $ k $
(they are $ \pm $ the singular values of $D_{\rm{S}}(\alpha)+k$): 
\begin{equation}
\label{eq:defEjk}
  \cdots \leq - E_j ( \alpha, k ) \leq \cdots \leq - E_{1 } ( \alpha, k ) \leq E_{ 1} ( \alpha, k ) \leq \cdots \leq E_j ( \alpha, k ) \leq \cdots  
 \end{equation} 
(For the general theory of Bloch--Floquet band spectrum see \cite[\S 5.3]{notes} and for an introduction in this context 
\cite[\S 3.1,3.2]{survey}.) A flat band occurs at $\alpha$ such that $ E_1 ( \alpha, k ) \equiv 0 $ and that is equivalent to \eqref{eq:defflat}.

Local solutions to $ D_{\rm{S}} ( \alpha ) u = 0 $ (or equivalently solutions to $ P ( \alpha, 0 ) u_0 = 0 $) are given by local 
holomorphic sections of a rank-2 holomorphic (topologically trivial) vector bundle over $ \mathbb C/\Lambda $. {(This is valid in greater generality but we concentrate on the case in which we know flat bands occur.)}
We denote it by $ \mathscr E ( \alpha ) $ and it is defined in \S \ref{s:rank2}. The first theorem describes the connection between the structure of this vector bundle and
the behaviour of bands:
\begin{theo}
\label{t:1} 
Suppose that $ E_j ( \alpha , k ) $ are defined in \eqref{eq:defEjk} for $ D_{\rm{S}} ( \alpha ) $ with $ V $ satisfying
\eqref{eq:defV}. 
There exist discrete sets $ \mathcal A , \mathcal B \subset \mathbb C $ such that we have the following trichotomy:
\[ \text{$ \alpha \notin \mathcal A \cup \mathcal B   \Longleftrightarrow  \mathscr E  ( \alpha ) \text{ is indecomposable }   \Longleftrightarrow 
E_{\pm 1 } ( \alpha  , k ) \simeq \pm c |k|^2  $;}\]
\[ \text{$ \alpha \in  \mathcal B   \Longleftrightarrow \mathscr E ( \alpha )  = L_0 \oplus L_1,  L_j \text{ trivial} 
  \Longleftrightarrow   E_{\pm j} ( \alpha  , k ) \simeq \pm c_j |k|, \ j =1,2  $,  $ E_3 ( \alpha, k ) > 0 $;}  \]
\[ \text{$ \alpha \in  \mathcal A  \Longleftrightarrow \mathscr E ( \alpha ) = L \oplus L^* ,  $ {\rm{deg}}\! L = m $
  \Longleftrightarrow   E_j ( \alpha  , k ) \equiv 0 , \ |j| \leq m  $, $ E_{m+1} ( \alpha, k ) >  0 $,} \]
 where $ L_j $ and $ L$ are line bundles. 
 \end{theo}
We refer to $ m $ in the third statement of the theorem as the {\em multiplicity} of a magic $ \alpha \in \mathcal A $. 
Although the model we consider does not seem to be physically relevant, Li--Yang \cite{LY} proved that the same trichotomy 
appears in the case of twisted multiple layers of graphene -- see \cite{yang} for a mathematical presentation and references to the physics literature. The fact that double Dirac points {(that is, points $k $ at which the bands 
$ k \mapsto E_j ( \alpha, k ) $ in \eqref{eq:defEjk} meet conically, as in 
for $ k= 0$
the case of $ \alpha \in \mathcal B $ in Theorem \ref{t:1})}
meet 
can occur at a discrete set of angles was not present in the physics literature. The computations in \cite{LY} suggest that Dirac points should occur when the angle of twisting between two bilayer graphene wafers is close to 0.44${}^\circ$.

\subsection{Quantisation condition for real magic $ \alpha$'s}
One of the most striking observations made in \cite{magic} was a quantisation rule for 
real magic alphas in the chiral model (the definition is the same as \eqref{eq:defflat} but with 
the operator \eqref{eq:defD})  with the Bistritzer--MacDonald potential
\eqref{eq:defU}: if $ \alpha_1 < \alpha_2 < \cdots \alpha_j < \cdots $ is the 
 sequence of positive $ \alpha$'s for which \eqref{eq:flat} holds, then 
 \begin{equation}
 \label{eq:quant}  \alpha_{j+1} - \alpha_j  = \gamma + o ( 1 ) , \ \ j \to + \infty , \ \ \gamma \approx \tfrac32.
 \end{equation}
The more accurate computations made in \cite{beta} suggest that 
$  \gamma \approx 1.515 $. 
In the scalar model \eqref{eq:scalar} with $ V( z ) = U_{\rm{BM}} ( z ) U_{\rm{BM}}  ( -z ) $ where 
$ U_{\rm{BM}}  $ is given by \eqref{eq:defU} we numerically observe the following rule for real elements of $ \mathcal  A$:
real magic $ \alpha$'s are double (in the sense of Theorem \ref{t:1}), $ \alpha_{2j-1} = \alpha_{2j} $, $j = 1, 2 , \cdots $, 
\begin{equation}
 \label{eq:quants}  \alpha_{j+1} - \alpha_{j-1}  = 2 \gamma + o ( 1 ) , \ \ j \to + \infty ,  \end{equation}
where $ \gamma $ is the same as in \eqref{eq:quant}.  This "doubling" phenomenon in the passage from 
chiral to scalar models occurs for other potentials as well.

In this paper we discuss the connection between the WKB structure of solutions to \eqref{eq:defpro} and the distribution 
of $ \alpha$'s in the scalar model. In \S \ref{s:mechzero} we present the mechanism for zero creation in protected states 
$ v ( \alpha ) $ for the Bistritzer--MacDonald potential. That is related to the distribution of magic alphas -- see \S \ref{s:zeros}.
In \S \ref{s:Henry} we discuss the potential 
\begin{equation}
\label{eq:defHpot}  U ( z ) := i \overline{ U_{\rm{BM}}  ( z ) }^2 , \ \ \   V ( z ) := U ( z ) U ( -z )  ,
\end{equation}
for which the structure of the equation \eqref{eq:scalar0}  simplifies significantly and for which we have \eqref{eq:quant}
with $ \gamma = \frac14 $.
That is presented in the next theorem though we stress that we do {\em not} provide a full proof in this paper. At this point we are only able to provide a heuristic WKB argument supported by numerical evidence. Hence the next theorem is only a theorem in the ``physical" sense.
\begin{theo}
\label{t:2} 
The potential $ V $ is given in \eqref{eq:defHpot}  satisfies \eqref{eq:defV} and hence defines a set $ \mathcal A $ of
Theorem \ref{t:1}. 
If $ \mathcal A \cap \mathbb R_+  = \{ \alpha_j \}_{ j=1}^\infty $, $ \alpha_{j+1} \geq
\alpha_j $ then each $ \alpha_{2j-1}  = \alpha_{2j } $, $ j = 1, 2, \cdots $ (that is, each real $ \alpha $ has multiplicity $ 2 $) and 
\begin{equation}
\label{eq:deltal}    \alpha_{k} - \alpha_{k-2}  = \tfrac14 + o ( 1) , \  \ \ k \to \infty . \end{equation}
\end{theo} 
Numerically computed $ \alpha_j $'s for the scalar model based on \eqref{eq:defHpot} are 
\begin{center}
\begin{tabular}{rclcc}
\multicolumn{1}{c}{$k$} & &
\multicolumn{1}{c}{$\alpha_k$} &
& $\alpha_{k}-\alpha_{k-2}$ \\[2pt] \hline
1  && 0.1395 &&               \\
3  && 0.3803 && 0.2407              \\
5  && 0.6281 && 0.2478              \\
7  && 0.8772 && 0.2490              \\
9  && 1.1267 && 0.2494              \\
11 && 1.3764 && 0.2496              \\
13 && 1.6262 && 0.2497              \\
15 && 1.8760 && 0.2498              \\
17 && 2.1259 && 0.2498              \\
19 && 2.3758 && 0.2498              \\
\end{tabular}
\end{center} 
This shows that the heuristic argument for \eqref{eq:deltal} is confirmed numerically. The magic alphas for 
the scalar model and the potential $ V $ in \eqref{eq:defHpot} are shown in Figure \ref{f:magicH}. They have the 
same structure (on a different scale) as the magic alphas for the Bistritzer--MacDonald potential shows (in red) in 
Figure \ref{f:rot}. We also remark that for the chiral model with the potential $ U $ in \eqref{eq:defD} we have an 
analogue of \eqref{eq:quant} and \eqref{eq:quants} with $ \gamma = \frac18$: the spacing is halved and the multiplicities are simple. 

\begin{figure}
\includegraphics[width=12cm]{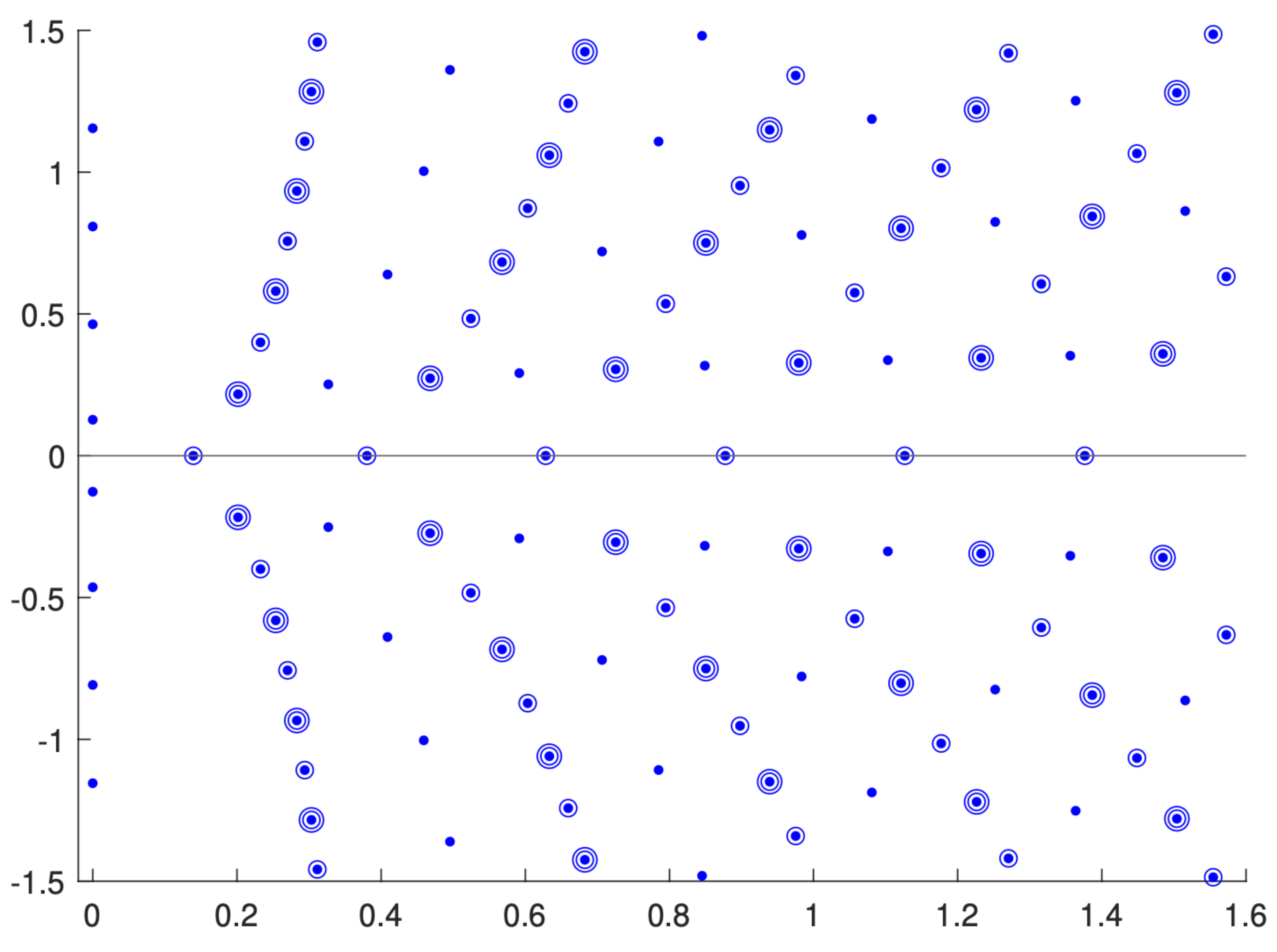}
\caption{\label{f:magicH} The sets $  {\mathcal A } $  for the scalar model with the potential $ V $ given by 
\eqref{eq:defHpot}. The multiple circles indicate multiplicity (in the sense of Theorem \ref{t:1}, $  m=1,2,3 $ all occur).
Except for different scaling the structure of the set is the same as for $ \mathcal A $ shown in red in Figure \ref{f:rot}.}
\end{figure}
\subsection{Nonvanishing potentials} 
We can study properties of solutions to $ P ( \alpha, k ) u = 0 $ for more general potentials. In that case, we do not expect
to have $ \alpha $'s for which \eqref{eq:defflat} holds but we could ask for which $ \alpha$'s and $ k $'s we have 
a non-trivial kernel of $ P ( \alpha, k ) $. One case to consider is that of 
\begin{equation}
\label{eq:V2W2}  V=W^2 \neq 0 ,
\end{equation}  as in that case the characteristic 
variety 
(that is the set of $ ( x, \xi ) \in T^* \mathbb C /\Lambda $ such that 
$ q ( x, \xi ) := ( 2 \bar \zeta )^2 - V (z ) = 0 $, $ \zeta =\frac12 ( \xi_1 - i \xi_2 ) $, $ z = x_1 + i x_2 $)
is a union of two disjoint tori. Those tori are not Lagrangian unless $ \Im \partial_z W \equiv 0 $.

Assuming \eqref{eq:V2W2} and that 
\begin{equation}
\label{eq:defW0} 
 W_0 := |\mathbb C/\Lambda |^{-1} \int_{\mathbb C/\Lambda} W ( z )  dm ( z ) \neq 0 , 
 \end{equation}
 we have the following {\em na\"{\i}ve} WKB argument for the existence of non-trivial kernel of $ P ( \alpha, k ) $:
 we construct  two global phase functions by solving
\[ 2 \partial_{\bar z } \varphi + k  = \pm \alpha W ( z ),  \  \    \varphi (z ) =  \pm ( \langle \alpha W_0, z \rangle + 
\alpha \psi  ) - \langle k , z \rangle , \ \  \psi ( z + \gamma ) = \psi ( z )  . \]
(Here $ \langle k , z \rangle = \Re (k \bar z ) $.) 
A  {\em na\"{\i}ve} quantisation condition is the given by 
\begin{equation}
\label{eq:naive}   \forall\, \gamma \in \Lambda,  \  \varphi ( z + \gamma ) \equiv \varphi( z )  \!\!\!\! \mod \! 2 \pi \mathbb Z
\  {\Longleftrightarrow} \  \alpha \in W_0^{-1} (  \Lambda^* \pm k )  . \end{equation}
It is not clear under what assumptions this gives an approximate values of $ k $ and $ \alpha $ for which 
we have a nontrivial solution to $ P ( \alpha, k ) u = 0 $. It certainly does not hold in general -- see Figure \ref{f:stokes}. We remark that for non-self-adjoint operators for which 
the characteristic variety is a {\em single} torus, close in some sense to a Lagrangian torus, Melin and Sj\"ostrand 
\cite{mess} obtained Bohr--Sommerfeld quantisation rules for the spectrum. Their methods do not seem to be applicable to our case. 

\begin{figure}
\centering
\includegraphics[width=13cm]{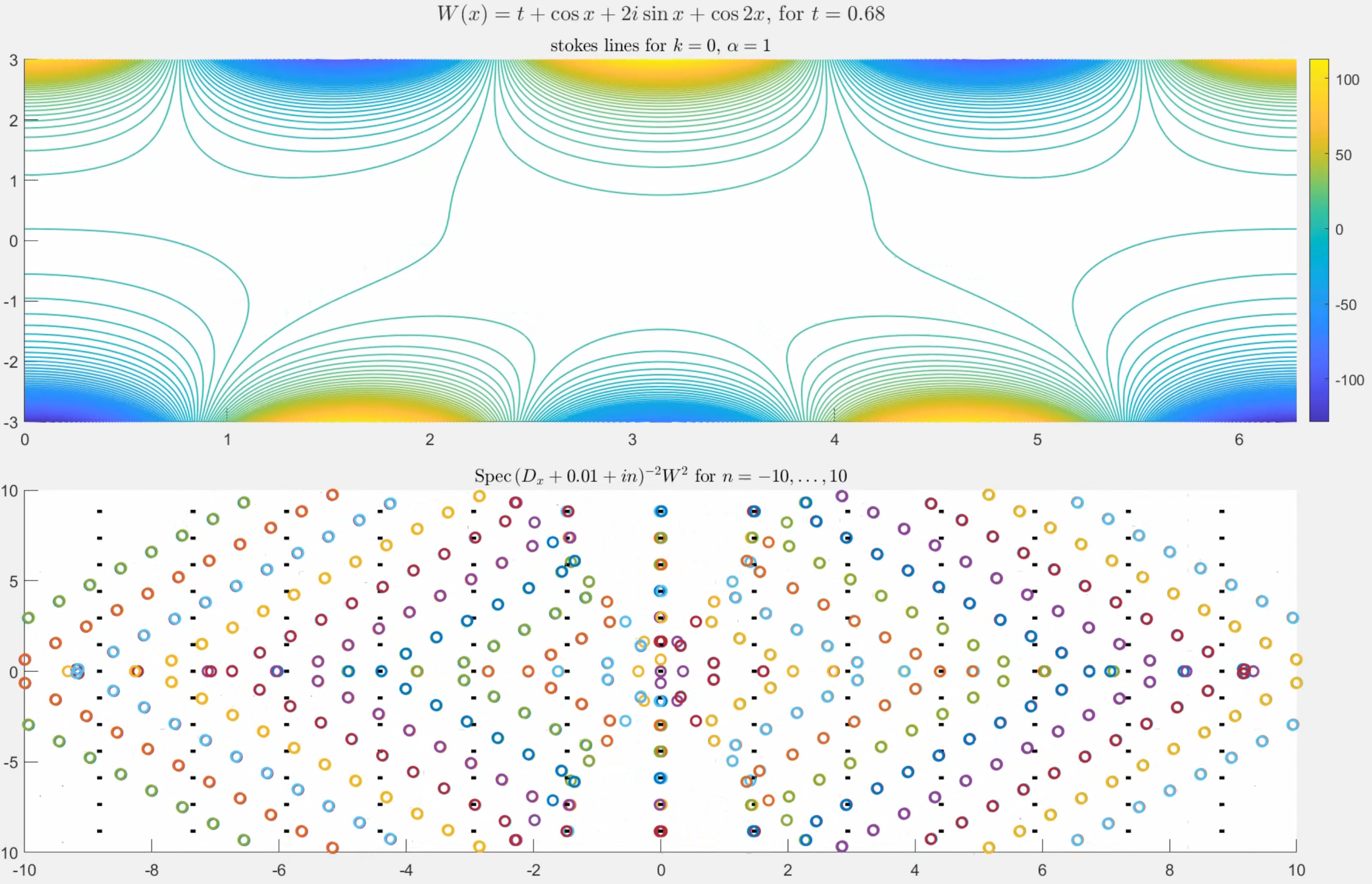} \includegraphics[width=13cm]{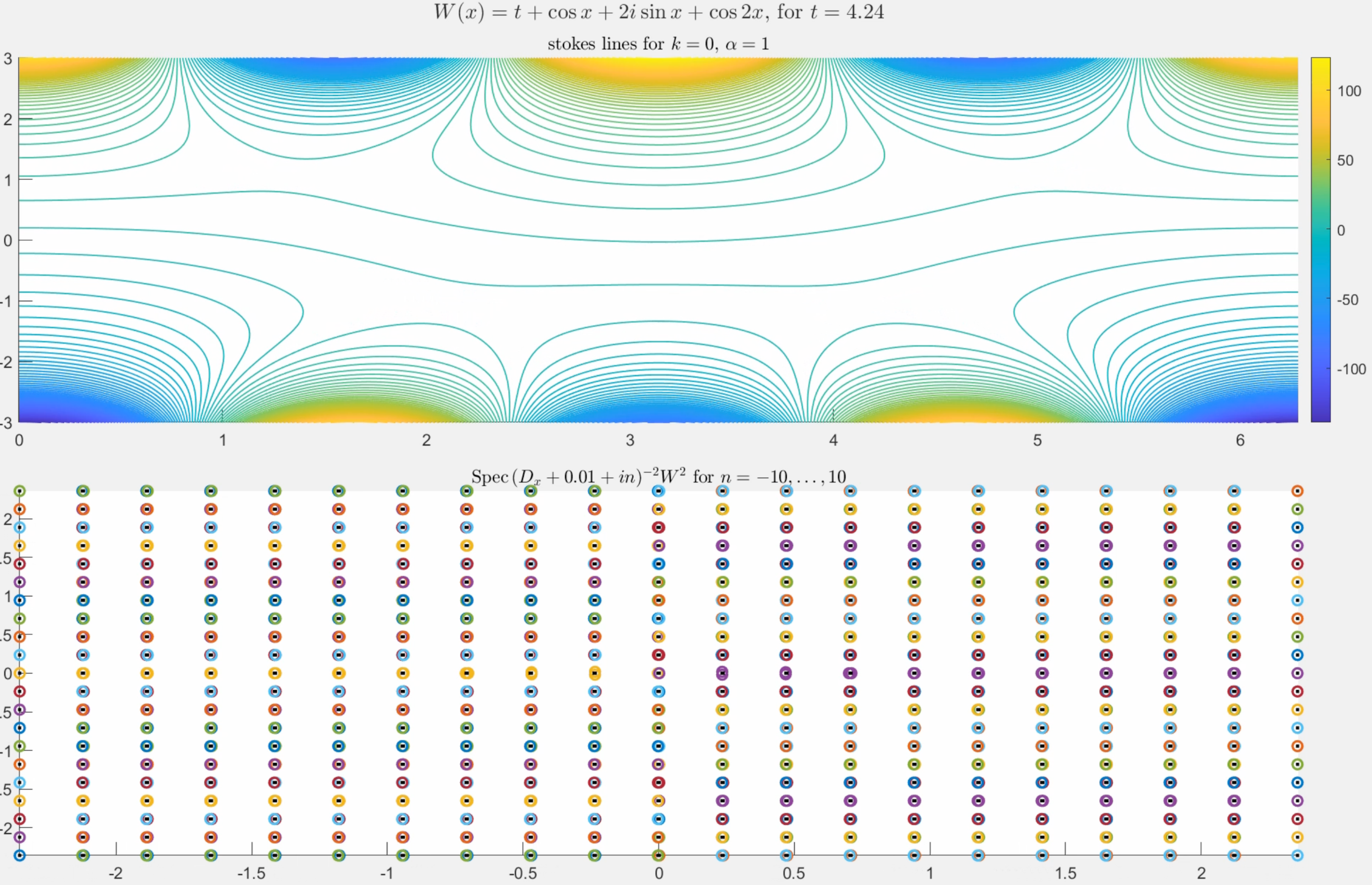}
\caption{\label{f:stokes} The bottom panels show sets of $ \alpha $'s (circles; coloured depending on the momentum variable of $ y $) for which $ ( D_x + i D_y )^2 - \alpha^2 W ( x ) $, $ ( x, y ) \in 
\mathbb R^2/2 \pi  \mathbb Z ^2 $ has a nontrival kernel. The top panels show curves on which 
$ \Im( W ( \gamma ( t ) ) \gamma' ( t ) )= 0 $ (Stokes lines). The first $ W $ satisfies \eqref{eq:W0} but does not have 
Stokes loops \eqref{eq:stol}; on the second $W $ has Stokes loops. The quantisation rule 
\eqref{eq:naive} (shown by black dots) does not  apply on the top and is very accurate on the bottom. For an animated version see \url{https://math.berkeley.edu/~zworski/stokes.mp4}.}
\end{figure}

In view of these complications we consider a yet simpler model in which one can use separation of variables and complex WKB. For that we take $ \Lambda : =  \mathbb Z^2 $ and $ W ( x, y ) = W ( x ) $, $ ( x, y ) \in \mathbb R^2 $. Taking the Fourier modes in $
 y $ we obtain a family of operators depending on $ \zeta = 2 \pi m i  + k $, where $ m \in \mathbb Z $ is a Fourier mode in $ y $ 
 and $ k $ is in the fundamental domain of the dual lattice {$ 2 \pi (\mathbb Z + i \mathbb Z)   $}:
 \begin{equation}
 \label{eq:Pzet}  P_\zeta  ( \alpha ) :=  ( D_x + \zeta )^2 - \alpha^2 W ( x )^2 . \end{equation}
 We assume that $ W ( x ) $ is an entire function satisfying $ W ( x + 1 ) = W (x) $.  In that case methods of complex WKB can be applied 
 but they normally require precise properties of a given $ W ( x ) $, typically a trigonometric polynomial -- see for instance 
 the work of Galtsev--Shafarevich \cite{gash} and references given there. For a brief account
 of complex WKB method and additional references we suggest \cite[Chapter 7]{SjBook}. Here we present a general result which gives 
  a condition under which the quantisation rule \eqref{eq:naive} holds:

\begin{theo}
\label{t:3} Suppose that  $ W ( x ) $ in an entire function satisfying $ W ( x + 1 ) = W ( x ) $. 
 Put $W_0:=\int_0^1W(x)\,dx$ and assume that $W_0\neq 0$.
If 
there exists a closed smooth non-self-intersecting curve $ {\gamma} : [ 0 , 1 ] \to \mathbb C/ \mathbb Z $, $ \gamma' ( t ) \neq 0 $, 
such that 
\begin{equation}
\label{eq:stol}  \Im  {\gamma}' ( t )  W_0^{-1} W ( {\gamma} ( t) ) =0 , \ \  W ( \gamma ( t ) ) \neq 0 , \ \ t \in [ 0, 1 ] ,
 \end{equation}
then for $ | \alpha | \gg 1$ and $ | \Re \zeta | \leq \pi $, 
\[  \ker ( (D_x+\zeta )^2 - \alpha^2 W ( x ) ) \neq \{ 0 \}  \  \Longleftrightarrow \ 
\exists \, n \in \mathbb Z \  \alpha = W_0^{-1} ( 2 \pi n \pm \zeta ) + 
 \mathcal 
O ( |\alpha|^{-1} ) . \]
\end{theo} 

The restriction on $ |\Re \zeta  | $ is not essential as for $ \ell \in \mathbb Z $, $ P_{\zeta + 2 \pi \ell }  =
e^{ - 2 \pi i \ell x }  P_{\zeta} e^{ 2 \pi i \ell x } $. 
The proof shares many features with standard complex WKB methods but we are asking different questions and allow for a very large class of potentials.  A more precise version is given in Theorem \ref{t:4}. An illustration of the result is given in Figure \ref{f:stokes}. One reason for presenting this case, much simpler than that considered in 
the heuristic Theorem \ref{t:2}, is to indicate the potential complications for operators involving the potentials 
$ U_{\rm{BM}} $ and $ i \overline{U_{\rm{BM}} }^2  $ for which flat bands are present.

We call such $ \gamma$'s satisfying \eqref{eq:stol} {\em Stokes loops}. A condition on $ W $ which guarantees their existence is given in Proposition \ref{p:stokes-perturb}:  $ \int_0^1 W ( x ) dx = W_0 \neq 0 $,  and 
\begin{equation}
\label{eq:conW}     
W ( z ) \neq 0 ,  \ \   | \Im ( W_0^{-1} W( z ) ) | \leq \Re ( W_0^{-1} W ( z ) ) \ \ \ \text{ for $ | \Im z |\leq 1 $.} \end{equation}
For instance, if $ f ( z ) $ is entire and $ f ( z + 1 ) = f ( z ) $ and $ \int_0^1 f ( z) dz  = 0  $, then $ W ( z ) = M + f ( z ) $ will satisfy \eqref{eq:conW} for 
$ M \geq \sqrt 2 \max_{ |\Im z | \leq 1 } | f ( z) | $. 

{The simplest case in which the relevance of Stokes loops becomes apparent is
that of $ \zeta = 0 $ and $ \alpha \in \mathbb R $. A heuristic argument for this goes as follows. We first consider the case of 
$ W (x)> 0 $, $ x \in \mathbb R $. In that case, standard WKB
analysis gives
two solutions, $ u_\pm ( x, \alpha  ) =  \Psi( x, \alpha )^{-\frac12} e^{\pm  i \alpha \int_0^x \Psi ( y , \alpha ) dy } $, $ \Psi ( x, h ) \simeq 
W ( x ) + \sum_{1}^\infty \alpha^{-j} \Psi_j ( x ) $. The  
asymptotic condition for having a periodic solution is then 
$ \alpha \int_0^1 \Psi ( y , \alpha )dy = 2 \pi n $, $ n \in \mathbb Z $, which to leading order is the quantisation condition in Theorem \ref{t:3}. Now suppose that for a holomorphic $1$-periodic $ W$
we have a curve satisfying \eqref{eq:stol} and, to simplify the notation assume that $ W_0 > 0 $. Then on $ \gamma$ the
holomorphic equation $ (  D_z^2 - \alpha^2 W ( z ) ) u ( z ) = 0 $ becomes 
\[ \left(  D_t ^2 
- \alpha^2 ( \gamma' ( t ) W ( \gamma ( t ) ))^2  -
\tfrac12  \gamma''' ( t ) + \tfrac14 \gamma''(t )^2 \right) [ e^{-\gamma' ( t )/2 } u ( \gamma ( t ) ) ] = 0 . \]
(We note that $ e^{-\gamma' ( t )/2 } $ is periodic.) The only difference with the real case is the presence of lower order terms
($ \mathcal O (1) $ corresponding to $ \mathcal O(h^2) $ if
$ h := 1/\alpha $ is the semiclassical parameter) and this suggests that the same condition for periodicity under $ t \mapsto t + 1 $, that is, the quantization condition 
in Theorem \ref{t:3}, holds. Since holomorphy then gives
$ u ( z + 1 ) = u ( z )  $, we obtain a periodic solution for
$ x \in \mathbb R $. To make this argument rigorous and to handle
the case of general $ \zeta $ we proceed by a careful semiclassical analysis of a monodromy operator of solutions to a system equivalent to \eqref{eq:Pzet}.}

\medskip\noindent {\sc Notation.} We write $ f = \mathcal O_ \rho (g)_H $ for $  \| f \|_H \leq  C_\rho g$, 
that is the constant depends on $ \rho $. When the context is clear either $ \rho $ or $ H$ 
may be dropped. We denote by $ \mathscr O ( U )$ holomorphic functions on an open set 
$ U \subset C^d $. If $ z_0 \in U $ where $ U $ is a topological space we write
$ \neigh_U (z_0 ) $ for an open neighbourhood of $ z_0 $ in $ U $.

\section{Review of the scalar model}

In this section we review the properties of the scalar model, starting with the chiral model from which it is derived.

\subsection{Chiral model}
\label{s:chiral}

The scalar model is derived from the the chiral model of twisted bilayer graphene (TBG) \cite{magic}, \cite{beta}, 
\cite{bhz2}:

\begin{equation}
\label{eq:defD}   
\begin{gathered} Q ( \alpha, k ) := D ( \alpha ) + k , \ \  
D ( \alpha ) := \begin{pmatrix} 2 D_{ \bar z } & \alpha U ( z ) \\
\alpha U ( - z ) & 2 D_{\bar z } \end{pmatrix} ,
\ \ \ \Omega  = \mathbb C , 
 \\  2D_{\bar z } = \tfrac 1 i ( \partial_{{x_1}} + i \partial_{x_2} ) , \ \
z = {x_1} + i x_2 \in \mathbb C , 
\end{gathered}
\end{equation}
where $ U $ satisfies 
\begin{equation}
\label{eq:defU}
\begin{gathered}   U ( z + \gamma ) = e^{ i \langle \gamma , K  \rangle } U ( z ) , \ \ U ( \omega z ) = \omega U(z) , \ \ \overline {U ( \bar z ) } = - U ( - z ) , \ \ \ 
\omega = e^{ 2 \pi i/3},  \ \\ \gamma \in \Lambda := \omega \mathbb Z \oplus \mathbb Z , \ \ K = \tfrac43 \pi, \ \ \langle z , w \rangle := \Re ( z \bar w ) . 
\end{gathered} 
\end{equation}
The constant $ K$ is determined up to sign by 
\[ \omega K \equiv K \not \equiv 0 \!\!\! \mod \Lambda^* , \ \ \
\Lambda^* :=  \frac {4 \pi i}  {\sqrt 3}  \Lambda , \] 
that it is a high symmetry point with respect to the dual lattice. Although the chiral model is not considered accurate for large $ \alpha $ (a dimensionless constant proportional the reciprocal of the angle of twisting), it remains popular in physics and its mathematical properties remain puzzling -- see \cite{survey}.

The principal example of $ U $ is given by the Bistritzer--MacDonald potential
\begin{equation}
\label{eq:defU2}
U_{\rm{BM}}  ( z ) =  - \tfrac{4} 3 \pi i \sum_{ \ell = 0 }^2 \omega^\ell e^{ i \langle z , \omega^\ell K \rangle }.
\end{equation}
It is, up to a real multiplicative constant, the unique potential satisfying \eqref{eq:defU} and having the lowest Fourier modes.

The scalar model is defined as follows: 
\begin{equation}
\label{eq:scalar}
P ( \alpha , k ): = ( 2 D_{\bar z } + k )^2 - \alpha^2 V ( z ) , \ \ \  V ( z ) := U ( z ) U (-z ) . 
\end{equation} 
It is derived from the chiral model by neglecting lower order terms, as $ \alpha \to \infty $, in the principally scalar reduction:
\[ \begin{gathered}  ( D ( - \alpha ) + k) ( D ( \alpha ) + k ) = P ( \alpha , k ) I_{\mathbb C^2 } + \alpha R_k ( \alpha ) , \\
R_k ( \alpha ) := \begin{pmatrix}   k \alpha^{-1} 4 D_{\bar z } + k^2 \alpha^{-1} &
2 D_{\bar z } U (z ) \\
- ( 2 D_{\bar z } U ) ( -z ) &
 k \alpha^{-1} 4 D_{\bar z } + k^2 \alpha^{-1}\end{pmatrix} .
\end{gathered} \]
We refer to \cite[\S 2.2]{hizw} for the discussion of this reduction and a semiclassical explanation why 
$  \alpha R_k ( \alpha ) $ is a lower order term.

The potential $ V $ in \eqref{eq:scalar} inherits symmetries from \eqref{eq:defU} but we can formulate them independently as 
\begin{equation}
\label{eq:defV} 
V ( z ) = V ( -z ) , \ \  V ( z + \gamma ) = V ( z ) , \ \ V ( \omega z ) = \omega^2 V ( z ) , \ \  \overline{ V ( \bar z ) } = V ( z ) . 
\end{equation}

 Just as for the chiral model of 
TBG, a flat band at zero for a given $ \alpha$ means that 
\begin{equation}
\label{eq:flat}  
\alpha \in \mathcal A \ \Longleftrightarrow  \ 
\  \forall \, k \in \mathbb C \  \ker_{ H^2 ( 
\mathbb C / \Lambda; \mathbb C  )} P ( \alpha, k  ) \neq \{ 0 \} .
 \end{equation}
 We denote the set of $ \alpha $'s for which \eqref{eq:flat} holds by $ \mathcal A  $, the same notation 
 used in the chiral model \cite{beta}. It follows from \cite{gaz4} that
 $ \mathcal A $ is a discrete subset of $ \mathbb C $, and that 
 \[    \mathcal A  = -  \mathcal A = \overline{  \mathcal A  } . \]
It is shown, together with the magic $ \alpha$'s for the chiral model \eqref{eq:defD}, both with $ U $ given by \eqref{eq:defU2}, in Figure \ref{f:A}.
\begin{figure}
\centering
\includegraphics[width=12cm]{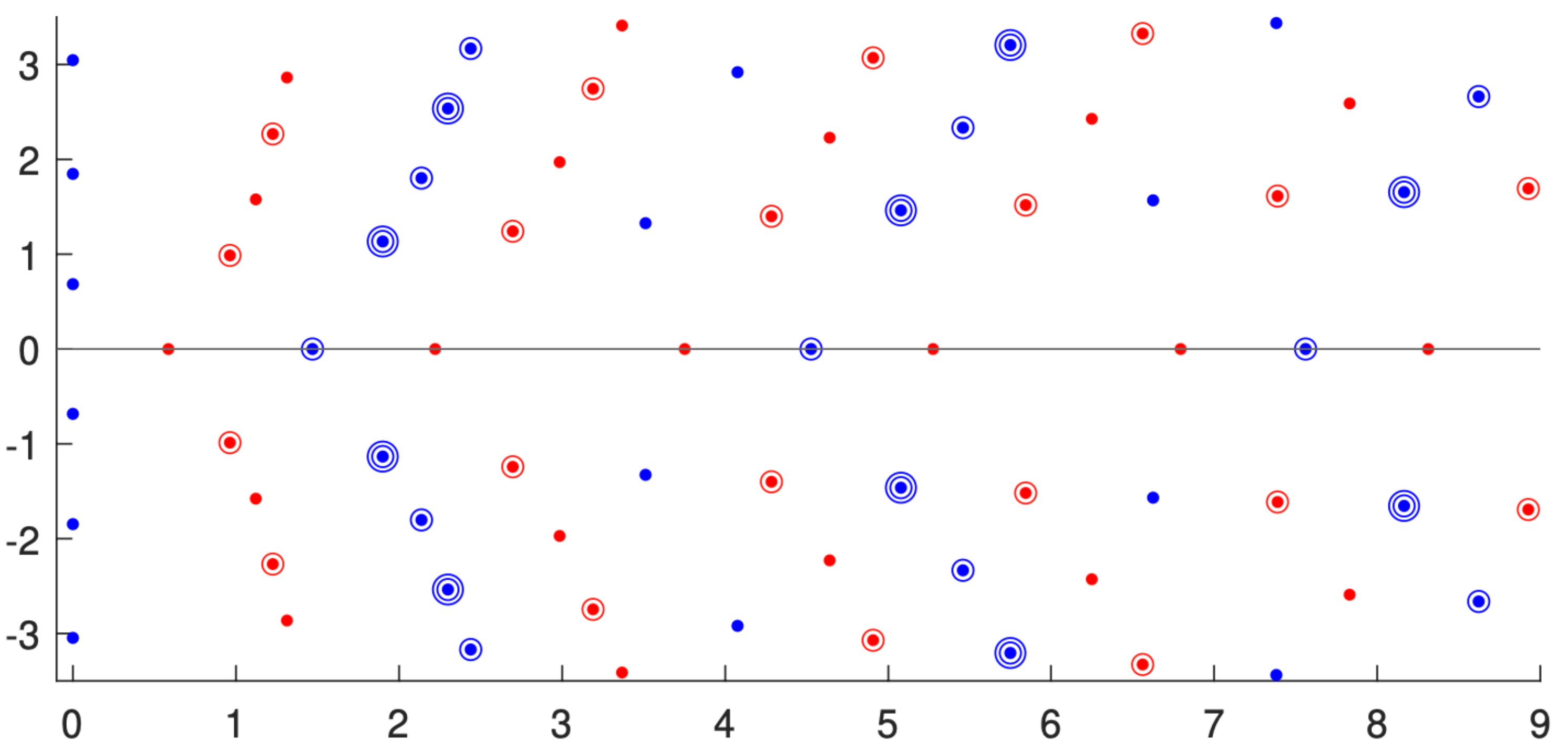}
\caption{\label{f:A} The sets $  {\mathcal A } $, of magic $ \alpha$'s  for the chiral model \eqref{eq:defD} 
(in {red}) 
and the scalar model \eqref{eq:scalar} (in {red}) with $U $ given by \eqref{eq:defU2}. Higher multiplicities are indicated by additional circles: in the scalar model the real $ \alpha$'s have multiplicity $2$. When we interpolate between the chiral model and the 
scalar model, the multiplicity two real $ \alpha$'s  split and travel in opposite directions to become
magic $ \alpha $'s for the chiral model: see \url{https://math.berkeley.edu/~zworski/Spec.mp4}, where the magic $\alpha$'s are plotted as a function of $t\in[0,1]$ which linearly interpolates between the chiral and scalar models.}
\end{figure}

We define multiplicity as follows:  if for $ \alpha \in \Omega $,  there is $k_0\in \mathbb{C}$ such that $\ker P(\alpha, k_0)=\{0\}$, then we define 
 \begin{equation}
\label{eq:mult}
m (\alpha, k ) := \frac{1}{ 2 \pi i } \tr \oint_{\partial D }   P( \alpha, \zeta ) ^{-1} 
\partial_\zeta P ( \alpha , \zeta ) d \zeta , 
\end{equation} 
where the integral is over the positively oriented boundary of a disc 
$ D $ which contains $ k$ as the
only possible pole of $ \zeta \mapsto P ( \alpha, \zeta ) $. 
Otherwise, we  put $ m ( \alpha, k ) =  
\infty $ for all $k\in \mathbb{C}$.

\subsection{Symmetries}
\label{s:sym}

To simplify notation, we write $ P ( k ) := P ( \alpha, k ) $, where $ P ( \alpha, k ) $ was defined in \eqref{eq:scalar0} with
$ V $ satisfying \eqref{eq:defV}.  

We define the following symmetries:
\[ \begin{gathered} \Omega u ( z ) := u ( \omega z ) , \ \ \ \mathscr L_\gamma u ( z ) = u ( z + \gamma ) , \ \ \ 
 \mathscr E_1 u ( z ) = u ( -z ) , \\\ 
 \mathscr A_1 u ( z ) = \overline{ u ( \bar z ) } , \ \ \ \  \mathscr R_1 u ( z ) = u ( \bar z ) ,
 \end{gathered}
 \]
with the following commutation properties with $ P ( k )  $: (here $ \gamma \in \Lambda $ defined in \eqref{eq:defU})
\begin{equation}
\label{eq:symk} 
\begin{gathered} \mathscr E_1 P (k ) = P (-k) \mathscr E_1, \ \ \
\mathscr L_\gamma P ( k )  = P ( k) \mathscr L_\gamma, \ \ \ P ( k ) \Omega = \omega \Omega P ( \omega k ) , \\
\mathscr A_1 P ( k )  = P( - \bar k )  \mathscr A_1,  \ \ \ 
\mathscr R_1P ( \alpha, k ) = P ( \bar \alpha, \bar k )^* \mathscr R_1 . 
\end{gathered} 
\end{equation}
It is useful to combine $ \mathscr E_1 $ and $ \Omega $:
\begin{equation}
\label{eq:defC}  \mathscr C_1 := \mathscr E_1 \Omega = \Omega \mathscr E_1 , \ \ \ 
P ( k ) \mathscr C_1  = - \tau \mathscr C_1 P ( \tau k ) , \ \ \ \mathscr C_1  u ( z ) = u ( \tau z ) , \ \ \ 
\tau := e^{ - \pi i / 3 } . 
\end{equation}
(Nothing is lost here since $ \mathscr C^3_1 = \mathscr E_1 $ and $ \mathscr E_1 \mathscr C_1  = \Omega $.)

We also note that for $ \mathbf a \in \Lambda^* $
\begin{equation}
    \label{eq:Pconj}
P ( k + \mathbf a ) = e^{- i \langle z , \mathbf a \rangle } P ( k ) e^{ i \langle z , \mathbf a \rangle }, \ \ 
u ( z )  \mapsto  e^{ i \langle z , \mathbf a \rangle } u ( z ) \text{ is unitary on $ L^2 ( \mathbb C/ \Lambda ) $.}\end{equation}
Using these symmetries, definition of the multiplicity \eqref{eq:mult} and \cite[Theorem 1]{gaz4} we obtain
\begin{equation}
\label{eq:mult1}  m ( \alpha, k ) = \left\{ \begin{array}{ll}   \ \ \infty & \alpha \in \mathcal  A , \\
2 \indic_{ \Lambda^* } ( k ) & \alpha \notin \mathcal  A . \end{array} \right. \end{equation}
In particular,
\[  \forall \alpha \in \mathbb C , \  \ker P ( \alpha, 0 ) \neq \{ 0 \} .\]

We now relate $ P ( k ) $ to the operator \eqref{eq:syst0}:
\begin{equation}
\label{eq:defDS}
D_{\rm{S}} = D_{\rm{S}} ( \alpha ) = \begin{pmatrix} 2 D_{\bar z} & \alpha V ( z ) \\
\alpha & 2 D_{\bar z } \end{pmatrix} 
: H^1 ( \mathbb C/\Lambda; \mathbb C^2 ) \to L^2 (\mathbb C/\Lambda ; \mathbb C^2 ) .
 \end{equation}
We have 
\begin{equation}
\label{eq:D2P}    ( D_{\rm{S}} + k )  \begin{pmatrix} 
u_1 \\ u_2  \end{pmatrix} = \begin{pmatrix} f \\ g \end{pmatrix} 
\ \Longleftrightarrow \ \left\{ \begin{array}{l} 
 \alpha u_1 = ( 2 D_{\bar z } + k ) u_2 - g,  \\ P ( k ) u_2 = \alpha f + ( 2 D_{\bar z } + k ) g .
\end{array} \right. 
\end{equation}
Hence, with $ \pi_1 ( ( u, v )^t ) = u $, $ \pi_2 ( ( u, v )^t ) = v $, $ \iota f := ( f , 0 )^t $, 
\begin{equation}
\label{eq:P2D2} 
\begin{gathered}  P ( k )^{-1} f = \alpha^{-1} \pi_2  ( D_{\rm{S}} + k )^{-1} \iota f, \ \ \ 
\partial_k P ( k ) P ( k )^{-1} f = 2 \pi_1 ( D_{\rm{S}} +  k )^{-1} \iota f , \\
( D_S + k)^{-1} = \begin{pmatrix}  \tfrac12 \partial_k  P ( k ) P (k)^{-1} & \tfrac12  \alpha^{-1} \partial_k  P ( k ) P(k)^{-1} \partial_k  P ( k ) - 
\alpha^{-1} \\
\alpha P(k)^{-1} & \tfrac12 P (k)^{-1}   \partial_k  P ( k )  \end{pmatrix} 
.
\end{gathered}
\end{equation}

For $ \alpha \notin \mathcal  A $,  \eqref{eq:P2D2} gives (with integrals is over a positively oriented circle around 
$ 0 $)
\begin{equation}
\label{eq:D2R} 
\begin{gathered}
\Pi ( \alpha ) := \frac{1}{ 2\pi i } \oint_0 ( D_S ( \alpha ) + \zeta )^{-1} d \zeta 
= \tfrac12  \begin{pmatrix}  a ( \alpha ) & b ( \alpha ) \\
c ( \alpha ) & d ( \alpha ) \end{pmatrix}, \\ a ( \alpha ) := \frac{1}{ 2 \pi i } \oint \partial_k P ( k ) 
P( k )^{-1} dk , \ \  d ( \alpha ) := \frac{1}{ 2 \pi i } \oint  P( k )^{-1} \partial_k P ( k ) dk . \end{gathered} 
 \end{equation}
Hence 
\begin{equation}
\label{eq:Pi2AD}  \tr \Pi ( \alpha ) = \tfrac12 ( \tr a( \alpha ) + \tr d ( \alpha ) ) = \tr a ( \alpha ) = m ( \alpha, 0 ) . 
\end{equation}
(The equality of the traces of $ a ( \alpha )  $ and $ d( \alpha ) $ follows from the Gohberg--Sigal theory -- see for instance \cite[Theorem C.11]{res}.) {We note that we have analogous formula for 
$ m ( \alpha, k ) $ by taking an integral around $ k$ in the definition of 
$ \Pi ( \alpha ) $.}

From \eqref{eq:mult} we conclude that 
\begin{equation}
\label{eq:D2R2}
 ( D_S ( \alpha ) + k )^{-1} = \frac{ \Pi( \alpha)  }k - \frac{D_S \Pi( \alpha ) } {k^2} + A ( k ) , 
 \end{equation}
and $ k \mapsto A ( k ) $ is holomorphic near $ k = 0 $.

We record the following symmetries of $ D_S $: the translation symmetry $ \mathscr L_\gamma D_{\rm{S}} = D_{\rm{S}} \mathscr L_\gamma $, and 
\begin{equation} 
\label{eq:defCS}
\mathscr C_2  \begin{pmatrix} u_1 \\ u_2 \end{pmatrix} := 
\begin{pmatrix} \bar \tau u_1 ( \tau \bullet ) \\
u_2 ( \tau \bullet ) \end{pmatrix}, \ \ \ 
\mathscr C_2 ( D_{\rm{S}} ( \alpha )  + k ) = \tau ( D_S ( \alpha ) + \bar \tau k ) \mathscr C_2 ,  \ \ \ 
\mathscr C_2  \Pi  = \Pi \mathscr C_2  , \end{equation}
\begin{equation*}
\mathscr R_2  
\begin{pmatrix} u_1  \\ u_2 \end{pmatrix} := 
\begin{pmatrix}   u_2 ( \bar \bullet )\\
u_1 ( \bar \bullet ) \end{pmatrix}, \ \ \  \mathscr R_2 (  D_{\rm{S}} (\alpha)  + k )  =  ( D_{\rm{S}} (\bar \alpha) + \bar k )^*  \mathscr R_2 , 
\end{equation*}
\begin{equation*}
\mathscr A_2  
\begin{pmatrix} u_1  \\ u_2    \end{pmatrix} := 
\begin{pmatrix}  -  \bar u_2 \\
\ \ \bar u_1  \end{pmatrix}, \ \ \  \mathscr A_2 (  D_{\rm{S}} (\alpha)  + k )  = 
- ( D_{\rm{S}} (\alpha) -  k )^*  \mathscr A_2 \end{equation*}
We note that 
\begin{equation}
\label{eq:C2R}  \mathscr R_2 \mathscr C_2  = \bar  \tau   \mathscr C_2^{-1}  \mathscr R_2 , 
\ \ \  \ \mathscr A_2 \mathscr C_2  = \tau  \mathscr C_2 \mathscr A_2 .
\end{equation}

A self-adjoint operator associated to $ D_{\rm{S} }  $ is given by 
\[  H_{\rm S } ( \alpha, k )  := \begin{pmatrix} 0 & D_{\rm{S}} (\alpha)  + k  \\
( D_{\rm{S} } ( \alpha ) + k )^* & 0 \end{pmatrix} : 
H^1 ( \mathbb C/\Lambda; \mathbb C^4 ) \to L^2  ( \mathbb C/\Lambda; \mathbb C^4 ) \]
It enjoys the standard chiral symmetry
\begin{equation}
\label{eq:defW}   \mathscr W H (\alpha,  k ) \mathscr W = - H ( \alpha, k ) , \ \ \  \mathscr W \begin{pmatrix}   \mathbf u \\ \mathbf v \end{pmatrix}
:= \begin{pmatrix}  -  \mathbf u \\ \ \  \mathbf v \end{pmatrix} , 
\ \ \  \mathbf u, \mathbf v \in L^2 ( \mathbb C/\Lambda ; \mathbb C^2 ). 
\end{equation}
The symmetries of $ P ( k ) $ give us also 
\begin{equation}
\label{eq:defC2}  \mathscr C_4 H ( \alpha,  k ) = H ( \alpha, \bar \tau k ) \mathscr C_4 , \ \ \  \mathscr C_4  \begin{pmatrix}  \mathbf
u \\ \mathbf v  \end{pmatrix}
:= \begin{pmatrix}  \ \mathscr C_2  \mathbf u \\  \tau \mathscr C_2 \mathbf v \end{pmatrix} , \ \ \ 
\mathscr C_4 \mathscr W = \mathscr W  \mathscr C_4. \end{equation}
Finally, 
\begin{equation}
\label{eq:defR2}  \mathscr R_4 H (\alpha ,  k ) = H ( \bar \alpha ,   \bar k ) \mathscr R_4, \ \ \
\mathscr R_4 \begin{pmatrix}  \mathbf  u \\ \mathbf  v  \end{pmatrix}
:= \begin{pmatrix}   \mathscr R_2 \mathbf  v  \\ \mathscr R_2 \mathbf  u  \end{pmatrix} ,  \ \ \ 
\mathscr C_4 \mathscr R_4  =  \mathscr R_4 \mathscr C_4^{-1} ,
 \end{equation}
 \begin{equation}
\label{eq:defA2} \mathscr A_4 H (\alpha ,  k ) = - H ( \alpha ,   - k ) \mathscr A_4, \ \ \
\mathscr A_4 \begin{pmatrix}  \mathbf  u \\ \mathbf  v  \end{pmatrix}
:= \begin{pmatrix}   \mathscr A_2 \mathbf  v  \\ \mathscr A_2 \mathbf  u  \end{pmatrix} ,  \ \ \ 
\mathscr C_4 \mathscr A_4  =  \mathscr A_4 \mathscr C_2 ,
 \end{equation} 
 We also define  for $\ell \in \mathbb Z / 6\mathbb Z$ (dropping $\mathbb C $ when $ n =1 $)
\[  L^2_\ell ( \mathbb C/\Lambda, \mathbb C^n ) := \{ 
u \in L^2 ( \mathbb C/\Lambda, \mathbb C^n ) : \mathscr C_n u = \tau^\ell u \} , \ \ \ n = 1, 2, 4, \]
noting that
\begin{equation}
\label{eq:L2ell} 
\begin{gathered}  L^2_\ell (  \mathbb C/\Lambda;  \mathbb C^2 ) = 
L^2_{\ell + 1}  (  \mathbb C/\Lambda ) \oplus L^2_{\ell}  (  \mathbb C/\Lambda ), \\ 
L^2_\ell (  \mathbb C/\Lambda; \mathbb C^4 ) = 
L^2_{\ell}  (  \mathbb C/\Lambda ; \mathbb C^2) \oplus L^2_{\ell-1}  (  \mathbb C/\Lambda ; \mathbb C^2 ), \\
  \mathscr R_1 : L^2_{\ell } ( \mathbb C/\Lambda ) \to L^2_{- \ell } ( \mathbb C /\Lambda ) , \ \ 
\mathscr R_2 :  L^2_{\ell } ( \mathbb C/\Lambda ; \mathbb C^2 ) \to L^2_{-\ell-1}  ( \mathbb C/\Lambda ; \mathbb C^2 ), \\
\mathscr R_4 , \mathscr A_4  : L^2_{\ell } ( \mathbb C/\Lambda ; \mathbb C^4 ) \to L^2_{-\ell}  ( \mathbb C/\Lambda ; \mathbb C^2 ).
\end{gathered}
\end{equation}
We then have 
\begin{equation}
\label{eq:L2map}  
\begin{gathered}  D_{\rm{S}} ( \alpha) : ( H^1 \cap L^2_\ell ) ( \mathbb C/\Lambda ; \mathbb C^2 ) \to 
L^2_{\ell + 1}  ( \mathbb C/\Lambda ; \mathbb C^2 ) , \\
H_\ell  ( \alpha ) := H ( \alpha ) := H ( \alpha, 0 ) : ( H^1 \cap L^2_\ell ) ( \mathbb C/\Lambda ; \mathbb C^4 ) \to 
L^2_{\ell}  ( \mathbb C/\Lambda ; \mathbb C^4 ) . 
\end{gathered}
\end{equation}

\section{Protected states and magic alphas in the scalar model}
\label{s:scal0}

This section is devoted to the proof of Theorem \ref{t:1}. We first investigate the family of protected states,
that is nontrivial elements of the kernel of $ D_{\rm{S}} ( \alpha ) $ which exist for all $ \alpha $. In the process we will define
the discrete sets $ \mathcal A $ and $ \mathcal B $ which appear in Theorem \ref{t:1}. We then discuss the presence of Dirac points for $ \alpha \in \mathcal B $ and the relation between

\subsection{Family of protected states}
\label{s:scal}

{For completeness we start with a direct proof of \eqref{eq:mult1}.}
\begin{lemm}
\label{l:210}
{There exists a discrete set $ \mathcal A \subset \mathbb C $ such that for 
$ \alpha \notin \mathcal A $, $  k \mapsto ( D_{\rm{S}} ( \alpha ) + k )^{-1} $
is meromorphic on $ \mathbb C $ and, in the notation of \eqref{eq:mult},
\[ m ( \alpha, k ) = 2 \indic_{ \Lambda^* } ( k )  . \]
For $ \alpha \in \mathcal A $, $ \Spec D_{\rm{S}} ( \alpha) = \mathbb C $.}
\end{lemm}
\begin{proof} 
{In the proof we use \eqref{eq:Pi2AD} (see the comment below that formula) and, when $ k \mapsto ( D_{\rm{S}} ( \alpha ) + k ) ^{-1} $
is meromorphic, write
\begin{equation}
\label{eq:newdefm} m ( \alpha, k ) = \frac{1}{ 2\pi i } \oint_k ( D_S ( \alpha ) + \zeta )^{-1} d \zeta. \end{equation}
We then observe that \eqref{eq:Pconj} and \eqref{eq:defCS} show 
\begin{equation}
\label{eq:malpha} m ( \alpha, k + a ) , \ \ a \in \Lambda^*, \ \ \ 
m ( \alpha, \omega k ) = m ( \alpha, k ) .
\end{equation} 
For any fixed $k_0 \notin \Lambda^*$, 
$ m ( 0 , k_0 ) = 0 $, and hence (using analytic Fredholm theory 
for $ \alpha \mapsto D_{\rm{S}} ( \alpha ) + k_0 $, see \cite[\S C.3]{res}),  
$m(\alpha,k_0) =  0$ except on a discrete set of $ \alpha$'s which we denote by $\mathcal A_{k_0} $.
    We claim that
    \begin{equation}
        \label{eq:claima}
        \forall \, \alpha \in \mathbb C\setminus \mathcal A_{k_0} , \ k \in 
        \mathbb C \ \
    m(\alpha,k) = 2\indic_{\Lambda^*}(k) . 
    \end{equation}}
{To start we note that $\mathbb C\setminus \mathcal A_{k_0} $ is connected and for all $\alpha \in \mathbb C \setminus \mathcal A_{k_0} $, $m(\alpha,k) = 0$ except for a discrete set of $ k$'s (again using analytic Fredholm theory). 
    Let $Z \subset \mathbb C \setminus \mathcal A_{k_0} $
    be the set of $\alpha$ with $m(\alpha,k) = 2 \indic_{\Lambda^*}(k)$ for all $k$.
    The set $Z$ is closed in $\mathbb C \setminus \mathcal A_{k_0} $ (this follows from \eqref{eq:newdefm}) and nonempty since $0 \in Z$.}

{    Suppose $\alpha_1 \in Z$ and $k_1 \in \mathbb C$.
    If $k_1\not\in \Lambda^*$ then $m(\alpha_1,k_1) = 0$ and for $(\alpha,k) \in U \times V $,
    $ U = \neigh_{\mathbb C } ( \alpha_1 ) $, 
    $ V = \neigh_{\mathbb C } (k_1 ) $, $ m(\alpha, k) = 0$
    (invertibility of $ D_{\rm S } ( \alpha ) + k :
    H^1 ( \mathbb C/\Lambda; \mathbb C^2 ) \to  L^2( \mathbb C/\Lambda; \mathbb C^2 ) $ is an open condition).  }   
    
{If $k_1 \in \Lambda^*$  then, since $ \alpha_1 \in Z $, 
$ m ( \alpha_1, k ) = 0 $ for $ k \in D( k_1, 4 \delta) \setminus 
D (k_1 , \delta/2 ) $, $ \delta \ll 1 $. Again, since invertibility is an open condition, there exist $ U = \neigh_{\mathbb C } ( \alpha_1 ) $, 
     such that for 
    $ \alpha \in U $ and $ k \in D ( k_1, 3 \delta ) \setminus D ( k_1, \delta ) $, $ m ( \alpha, k ) = 0 $. It follows
    from \eqref{eq:newdefm} and continuity of 
$ U \times \partial D ( k_1, 2 \delta  ) \ni ( \alpha, k )
    \mapsto ( D_{\rm{S}} ( \alpha ) + k ) ^{-1},  $
that for $ \alpha \in U $, 
\begin{equation}
\label{eq:summ}
     \begin{split} \sum_{ k \in D ( k_1,  \delta ) } m ( \alpha, k ) & = \frac{1}{2 \pi i}
\tr \int_{\partial D ( k_1, 2 \delta ) }  ( D_S ( \alpha ) + \zeta )^{-1} d \zeta \\
& =\frac{1}{2 \pi i} \tr \int_{\partial D ( k_1, 2 \delta ) }  ( D_S ( \alpha_1 ) + \zeta )^{-1} d \zeta = m ( \alpha_1, k_1 ) = 2.\end{split} \end{equation}
Now, if $ m ( \alpha, k ) \neq 0 $ for $ k \in D( k_1, \delta ) \setminus 
\{ k_1 \} $, then, using \eqref{eq:malpha}, and the fact that
$ \omega k_1 \equiv k_1 \mod \Lambda^* $, 
\[ m ( \alpha, k_1 + \omega ( k - k_1 ) = m ( \alpha , \omega k_1 + \omega( k -k_1 ) ) = m ( \alpha, \omega k ) = m ( \alpha, k ) \geq 1, \]
and consequently $ \sum_{ k \in D ( k_1, \delta ) } m ( \alpha, k ) \geq 3 $, contradicting \eqref{eq:summ}. This means that for $ \alpha \in U $ and
$ k \in D( k_1, \delta ) $, $ m ( \alpha, k ) = 2 \indic _{\Lambda^*} ( k ) $.}

{We conclude that every $ ( \alpha_1, k_1 ) $, $ \alpha_1 \in Z $
and $ k_1 \in \mathbb C $ has a neighbourhood in which 
$ m ( \alpha, k ) = 2 \indic_{\Lambda^* } ( k ) $. In view of 
periodicity in \eqref{eq:malpha} to obtain validity for all $ k $ we only need to check it for $ k \in \mathbb C/\Lambda^* $. Since that set is compact, the local result implies that every $ \alpha_1 \in Z $ has
a neighbourhood in which  the equality in \eqref{eq:claima} holds. That means that
$ Z $ is open in $ \mathbb C \setminus \mathcal A_{k_0}  $ and as it is closed, 
it is equal to $ \mathbb C \setminus \mathcal A_{k_0}  $. This proves \eqref{eq:claima} and hence the first part of the lemma with $ \mathcal A = \mathcal A_{k_0}  $.}

{Finally, we note that $ \mathcal A_{k_0}  $ is independent of $ k_0 
\notin \Lambda^* $. In fact, 
\[ \mathbb C \setminus \mathcal A_{k_0} 
\subset \{ \alpha \in \mathbb C : \, \forall\, k \in 
\mathbb C \ \,  m ( \alpha , k) = 2 \indic_{\Lambda^* } ( k ) \} 
\subset \mathbb C \setminus \mathcal A_{k_1} , \]
where the second inclusion follows from the fact that we can take
$ k = k_1 $, so that $ m ( \alpha, k_1 ) = 0 $. Hence we can drop $ k_0 $ in $ \mathcal A_{k_0} $. Since by definition,  
$D( \alpha ) + k_0 $ is not invertible for $ \alpha \in \mathcal A $
and any $ k_0 \notin \Lambda^*$, if follows that $ \Spec D_{\rm{S}} ( \alpha ) = \mathbb C $, completing the proof.}
\end{proof}

We have the following result proved using methods similar to those in \cite{beta}, \cite{bhz2} and
\cite{BHWY} but with some new twists (no pun intended).

\begin{prop} 
\label{p:simple}
Let $ \mathcal A $ be the set of magic $ \alpha $'s defined in \eqref{eq:flat}.  There exists a unique discrete set $ {\mathcal B} $,  with $ \mathcal B \cap \mathcal A = \emptyset $,  and such that for $ \alpha \neq 0 $, 
\begin{equation}
\label{eq:dimker}  \dim \ker P ( \alpha, 0 ) = \dim \ker D_{\rm{S}} ( \alpha ) = \left\{ 
\begin{array}{ll} 
 1 , & \alpha \notin \mathcal A \cup { \mathcal B } , \\
 2, & \alpha \in { \mathcal B } .
 \end{array} \right.
 \end{equation}
Moreover, there exists a holomorphic family
 unique up to multiplication by a holomorphic function
\begin{equation}
\label{eq:holu} \mathbb C \ni \alpha \mapsto u ( \alpha ) \in \ker  P ( \alpha, 0 ) \setminus \{ 0 \} , \ \ \
u ( \alpha ) \in L^2_0 ( \mathbb C/\Lambda ; \mathbb C ) .  \end{equation}
\end{prop}

\noindent
{\bf Remarks.} 1.The proposition and \eqref{eq:P2D2} show that for $ \Pi ( \alpha ) $ defined in \eqref{eq:D2R} we have
\begin{equation}
\label{eq:Jor1}   D_{\rm{S}} ( \alpha ) \Pi ( \alpha ) \neq 0 , \ \ \ \alpha \notin  \mathcal A \cup { \mathcal B } , \end{equation}
that is, $ D_{\rm{S}} ( \alpha ) $, has a Jordan block block structure at its eigenvalues (which are given by 
$ \Lambda^* $. {The way that $ D_{\rm{S}} (\alpha ) $ was defined we do not have \eqref{eq:dimker}
 at $ \alpha = 0$. Instead, $ \dim \ker P ( 0 , 0 ) = 1 $ and $ \dim \ker D_{\rm{S} } ( 0 ) = 2 $, and from the point of view of
 Theorem \ref{t:2}, we should include $ 0 $ in $ \mathcal B $. To have an agreement between the system and the scalar operator
 we could have considered
 \begin{equation}
 \label{eq:DtS}  \widetilde D_S ( \alpha ) := \begin{pmatrix} 2 D_{\bar z } & \alpha^2 V ( z ) 
 \\ 1 & 2 D_{\bar z } \end{pmatrix} = \begin{pmatrix} 1 & 0 \\ 0 & \alpha \end{pmatrix}^{-1} D_S ( \alpha ) 
 \begin{pmatrix} 1 & 0 \\ 0 & \alpha \end{pmatrix} , \end{equation}
 To maintain a closer analogy to the chiral model \eqref{eq:defD} we keep $ \alpha $ as a linear coupling constant.}

\noindent
2. As we will see, the set $  { \mathcal B } $ has an interesting interpretation in terms of band structure: 
for $ \alpha \in \mathcal  A $, the operator $ H ( \alpha ) $ has a flat band. For 
$ \alpha \notin \mathcal A \cup  { \mathcal B } $, the bands touch tangentially at $ 0 $, 
while for $ \alpha \in  { \mathcal B }  $, the bands exhibit a Dirac
point at $ 0 $. This will be explained in \S \ref{s:Dirac}.

\noindent
3. For $ \alpha \notin \mathcal A \cup \mathcal B $, the unique solution to $ P ( \alpha, 0 ) u = 0 $ 
satisfies symmetries inherited from the equation:
\begin{equation}
\label{eq:symu}
u ( z ) = u ( -z ) , \ \ \ u ( \omega z ) = u ( z ) , \ \ \ \overline{ u ( \bar z ) } = u ( z ) , \ \ \ u ( z + \gamma) = u ( z ) , \ \ \
\gamma \in \Lambda. 
\end{equation}
We note that this also implies that (since $ e^{ \frac13 \pi i} = - \omega^2 $)
\[  u ( e^{ \frac13 \pi i} z ) = u ( z ), \ \ \ \ u ( - \bar z ) = \overline{ u ( z ) } . \]
Analytic continuation shows that these properties remain true for the family $ \mathbb C \ni \alpha \mapsto u ( \alpha ) $.

\begin{proof}
We use
$ \widetilde D_S ( \alpha ) $ defined in 
\eqref{eq:DtS} and the corresponding self-adjoint family
\[  \widetilde H ( \alpha ) :=  \begin{pmatrix} 0 & \widetilde D_{\rm{S}} ( \alpha ) \\
\widetilde D_{\rm{S}} ( \alpha )^* & 0 \end{pmatrix}  .\]
Away from $ 0 $ these are equivalent to $ D_S ( \alpha ) $ and $ H( \alpha ) $ using \eqref{eq:DtS}. 
Since $\mathscr C_2$ commutes
with the conjugation in \eqref{eq:DtS} 
we still have the mapping properties
(\ref{eq:L2map}).

We have 
\[  \widetilde H( 0 ) = \begin{pmatrix} & & 2 D_{\bar z} & 0 \\
& & 1 & 2 D_{\bar z } \\
2 D_z &  1 & & \\
0 & 2 D_z \end{pmatrix}, \ \ \
\ker \widetilde H ( 0 ) = \{ \mathbf e_1 , \mathbf e_4 \}, \ \ \ 
\mathbf e_1 \in L^2_{-1} , \ \ \mathbf e_4 \in L^2_1 .  \]
so that in the notation of \eqref{eq:L2map}
\begin{equation}
\label{eq:dimker0}  \dim \ker \widetilde H_\ell ( 0  ) = \left\{ \begin{array}{ll} 
1 & \ell =  \pm 1 \\
0 & \text{otherwise. } 
\end{array} \right.
 \end{equation}
In view of \eqref{eq:defC2}, $ \mathcal W \widetilde H_{\ell} ( \alpha  ) \mathcal W = - \widetilde H_\ell ( \alpha ) $, and hence
the spectrum of $ \widetilde H_{\ell} ( \alpha ) $ is symmetric about $ 0 $. 
At $ \alpha = 0 $,  $ 0 \in \Spec (\widetilde H_{\ell } ( 0 ) ) $ only for $ \ell = \pm 1 $, and then it is a simple eigenvalue.
Hence, the multiplicity of $ 0 $ as an eigenvalue of $ \widetilde H_{\ell } (\alpha) $, $ \ell = \pm 1 $ has to be odd, while
that $ \widetilde H_\ell ( \alpha ) $, $ \ell \notin \{ \pm 1 \} $ has to be even. In particular, 
\[ \dim \ker \widetilde H_{\pm 1}  ( \alpha  ) \geq 1 . \]
In view of \eqref{eq:L2map} 
\begin{equation}
\label{eq:H2D}  
\begin{gathered} 
\ker \widetilde H_\ell ( \alpha ) =  \ker \widetilde D_{\rm{S}} ( \alpha )^* |_{ L^2_{\ell} }  \oplus \{ 0 \}  +   \{ 0 \} \oplus \ker \widetilde D_{\rm{S}} ( \alpha )  |_{ L^2_{\ell -1}  }  ,\\
\widetilde D_{\rm{S}} ( \alpha ) |_{ L^2_{\ell -1}  } : L^2_{\ell -1 } \to L^2_\ell, \ \ \ 
\widetilde D_{\rm{S}} ( \alpha )^* |_{ L^2_{\ell} } =  ( \widetilde D_{\rm{S}} ( \alpha ) |_{ L^2_{\ell -1}  })^*= 
L^2_{\ell} \to L^2_{\ell-1} . \end{gathered}
\end{equation}
Since it is true at $ \alpha = 0$, the continuity of (simple) eigenfunctions shows that for $ |\alpha | \ll 1 $, 
\[ \ker \widetilde H_1 ( \alpha) = \{ 0 \} \oplus \ker \widetilde D_{\rm{S}} ( \alpha )  |_{ L^2_{0}  }, \] 
and 
\[ \ker \widetilde H_{-1} ( \alpha  ) = \mathscr A_4 \ker \widetilde H_{1} ( \alpha ) = 
\ker \widetilde D_{\rm{S}} ( \alpha )^* |_{ L^2_{-1} }  \oplus \{ 0 \} 
. \]
This also shows that
\[ \dim \ker \widetilde D_S ( \alpha ) = 1, \ \ \ | \alpha | \ll 1 , \]
as $ \dim \ker \widetilde D_S ( \alpha ) $ is upper semicontinuous 
(see for instance \cite[Theorem 2.5.3]{HNotes}). 
Since $ {\rm{rank}} \, \Pi ( \alpha ) = 2 $ (this follows from \eqref{eq:Pi2AD}), we have 
$ \widetilde D_{\rm{S}} ( \alpha ) \Pi ( \alpha ) \neq 0 $ for $ |\alpha | \ll 1 $, that is, 
\[  \ker \widetilde D_{\rm{S}} ( \alpha ) = \widetilde D_{\rm{S}} ( \alpha ) \Pi ( \alpha ) L^2 . \]
In view of the mapping properties in \eqref{eq:L2map}, we conclude that $ \ran\Pi ( \alpha ) = 
\mathbb C u ( \alpha ) \oplus \mathbb C  v ( \alpha ) $, where 
\[  u ( \alpha ) \in \ker \widetilde D_{\rm{S}} ( \alpha ) \subset L^2_0 ( \mathbb C /\Lambda ; \mathbb C^2 ), \ \ 
v ( \alpha ) \in L^2_{-1}  ( \mathbb C /\Lambda ; \mathbb C^2)  , \ \ |\alpha| \ll 1 .  \]
In particular, if $ \pi_\ell : L^2 \to L^2_\ell $ are the orthogonal projections, we have for 
$ |\alpha | \ll 1 $, 
\begin{equation}
\label{eq:proj}
( \pi_0 + \pi_{-1} ) \Pi ( \alpha ) = \Pi ( \alpha ) . 
\end{equation}
We now note that for $ \mathbb C \setminus \mathcal A\ni \alpha \mapsto 
\Pi ( \alpha ) $ is a holomorphic family of projections (this follows from applying $ \partial_{\bar \alpha } $ to 
the definition of $ \Pi ( \alpha ) $ in \eqref{eq:D2R}). We then see that 
\eqref{eq:proj} holds for $ \alpha \in \mathbb C \setminus \mathcal A$.

Since $ \mathbb C \setminus \mathcal A \ni \alpha \mapsto 
\tr \pi_0 \Pi ( \alpha ) \pi_0 $ is holomorphic and hence identically equal to $ 1 $, 
and since $ \pi_0 \Pi ( \alpha ) \pi_0 = \pi_0  \Pi ( \alpha ) $ (the projections 
$ \pi_\ell $ commute with $ \Pi ( \alpha) $ thanks to \eqref{eq:defCS}), we 
conclude that
\begin{equation} 
\label{eq:dimpi0}  \dim \pi_0 \Pi ( \alpha ) L^2 = 1 ,  \ \ \alpha \in \mathbb C \setminus \mathcal A. 
\end{equation} 
Hence, if \eqref{eq:dimker} does not hold then 
\[  \ker \widetilde D_{S } ( \alpha ) |_{ L^2_{-1} } \neq \{ 0 \} , \]
that is, there exists a non-zero $ v \in L^2_{ -1 } ( \mathbb C/\Lambda ) $ such that 
$ P ( \alpha, 0 ) v = 0 $ (we then have $ [ 2 D_{\bar z } v , v ]^t \in 
\ker \widetilde D_{\rm{S} } |_{  L^2_{-1} ( \mathbb C/\Lambda; \mathbb C^2 ) } $). Since
\[  D_{\bar z}^2 : L^2_{-1} ( \mathbb C/\Lambda ) \to L^2_1 ( \mathbb C/\Lambda ) \]
is invertible, and $ u \mapsto V u $, takes $ L^2_{-1} \to L^2_1 $, we have
\[  P ( \alpha , 0 ) |_{ L^2_{-1} } = ( 2 D_{\bar z })^2 |_{ L^2_{-1} } ( I - \alpha^2 K ) , \ \ \
K := [( 2 D_{\bar z })^2 |_{ L^2_ {-1} }]^{-1} V : L^2_{-1}  \to L^2_{-1} . \]
The operator $ K $ is compact and hence $ I - \alpha^2 K $ is invertible outside of a discrete set of 
$ \alpha$'s.
We define the set $ \mathcal B $ as follows
\begin{equation}
\label{eq:defB} 
\alpha \in \mathcal B \ \Longleftrightarrow \   \ker P  ( \alpha , 0 ) |_{ L^2_{-1} } \neq \{ 0 \}  \ \text{ and } 
\alpha \notin \mathcal A . 
\end{equation}
We note that for $ \alpha \in \mathcal B $, $ \dim \ker \widetilde D_{\rm{S}}  ( \alpha ) = 2 $. In fact, if there existed three 
independent solutions $ \widetilde D_{\rm{S}} ( \alpha ) u_j = 0 $, for any $ z_0 \in \mathbb C/\Lambda $ we could find 
$ \gamma \in \mathbb C^3 \setminus \{ 0 \} $ such that  $  \sum_{j=1}^3 \gamma_j u_j ( z_0 ) = 0 $ (since 
$ u_j ( z_0 ) \in \mathbb C^2 $ these are two equations). Since $ \sum_{j=1}^3 \gamma_j u_j \neq 0 $ (from 
the independence of $ u_j $'s)
the now standard argument (see \cite{magic} and \cite[\S 6]{survey})
shows that $ \alpha \in \mathcal  A $.

To see \eqref{eq:holu} we modify the proof of \cite[Proposition 2.3]{bhz2}. For that we fix
$ \alpha $ and define a new operator
\[   \widetilde H ( \zeta ) := \begin{pmatrix} 0 & \widetilde D_{\rm{S}} (\alpha + \zeta )   \\
 \widetilde D_{\rm{S} } ( \alpha +\bar \zeta )  ^* & 0 \end{pmatrix}, \]
which has the property that $  \widetilde H( \zeta )  = \widetilde H ( \bar \zeta )^* $. 
We can restrict this operator to $ L^2_{1} ( \mathbb C/\Lambda; \mathbb C^4 ) $ and, as we already know that
$ \ker \widetilde D_S ( \alpha + \zeta ) |_{L^2_0 } $ is nontrivial, 
$ \dim \ker \widetilde H_1 ( \zeta ) \geq 1 $.  Rellich's theorem 
\cite[Chapter VII, Theorem 3.9]{kato}, then shows that there exists a holomorphic family 
$ \neigh_{\mathbb C } ( 0 ) \ni \zeta \to \widetilde u ( \zeta ) \in   \ker \widetilde H_1 ( \zeta )  \setminus \{ 0 \} $
(it corresponds to the $ 0 $ eigenfunctions of $ \widetilde H_{1} ( \zeta )  $). If $ \alpha + \zeta \notin 
\mathcal  A $, then $ \widetilde u ( \zeta ) = ( 0 , \widetilde {\mathbf u } ( \zeta ) )$, 
$ \ker \widetilde D_{\rm{S}} ( \alpha + \zeta ) |_{ L^2_0 } =  \mathbb C \widetilde {\mathbf u } ( \zeta )  $. This construction is
local near any $ \alpha \in \mathbb C $ and by partition of unity we can find
a smooth family $ \alpha \mapsto \mathbf u_1 ( \alpha ) \in \ker \widetilde D_{\rm{S}} ( \alpha )|_{ L^2_0 } $.

To modify to obtain a holomorphic family we first note that 
\[  0 =  \partial_{\bar \alpha } ( \widetilde D_{\rm{S}}  ( \alpha ) \mathbf u_1 ( \alpha ) ) = 
\widetilde D_{\rm{S}}  ( \alpha ) ( \partial_{\bar \alpha } \mathbf u_1  ( \alpha ) ) 
\ \Longrightarrow \   \partial_{\bar \alpha } \mathbf u_1 ( \alpha ) \in \ker\widetilde D_{ \rm{S}} ( \alpha) |_{ L^2_0 } .  \]
Using the fact that for $ \alpha \notin \mathcal  A $ this kernel is one dimensional we see that
\begin{equation}
\label{eq:falfa} \partial_{\bar \alpha } \mathbf u_1 ( \alpha ) =  f( \alpha ) \mathbf 
u_1( \alpha ), \ \ \ f ( \alpha ) := 
\frac{ \langle  \partial_{\bar \alpha }\mathbf  u_1 ( \alpha ),  \mathbf u_1 ( \alpha )\rangle}{ \| \mathbf u_1 ( \alpha ) \|^2} ,
 \  \ \ \alpha \notin \mathcal A  . \end{equation}
The formula for $ f ( \alpha ) $ defines a smooth function $ f \in C^\infty ( \mathbb C ) $ and that means that
the first formula in \eqref{eq:falfa} holds everywhere. 
The equation $ \partial_{\bar \alpha} F ( \alpha )  = f ( \alpha ) $ 
(see for instance \cite[Theorem 4.4.6]{H1} applied with $ P = \partial_{\bar \alpha } $
and $ X = \mathbb C $)
can be solved 
with $ F \in C^\infty ( \mathbb C ) $. This shows that 
$  \mathbf u ( \alpha ):= \exp ( - F ( \alpha ) ) \mathbf u_1 ( \alpha ) $ is 
holomorphic. Since $  \mathbf u ( \alpha ) = [ 2 D_{\bar z } u ( \alpha ) , u ( \alpha ) ] $, we obtain \eqref{eq:holu}.
\end{proof}

\noindent
{\bf Remark.} For $ \alpha \in \mathbb R \setminus \{ 0 \} $, having a kernel of $ D_{\rm{S}} ( \alpha )|_{L^2_{-1} }  $
is a natural possibility for another reason than the one presented before \eqref{eq:defB}.  It implies that $ \ker H_{0} ( \alpha ) \neq 0 $. 
 Since 
$ \ker H_{\ell} ( 0 ) = 0 $, $ \ell \neq \pm 1 $,  the chiral symmetry shows that $ \dim \ker H_\ell ( \alpha ) 
\equiv 0 \mod 2 $ and if $ \dim \ker H ( \alpha ) > 2 $, then there exists a unique (in view of \eqref{eq:dimpi0})
$ \ell \neq \pm 1 $ such that $ \dim \ker H_\ell ( \alpha ) = 2 $ (For $ \alpha \in 
\mathbb R \setminus \mathcal  A $, we have $ \dim \ker H ( \alpha ) \leq 4 $, 
as otherwise $ \mathscr R_4 H( \alpha ) = H ( \alpha ) \mathscr R_4 $ would show that 
$ \dim \ker D_{\rm{S}} ( \alpha ) > 2 $ contradicting \eqref{eq:Pi2AD}.)
 Since $ \mathscr R_4 : \ker H_\ell  ( \alpha ) \to H_{-\ell} ( \alpha ) $, that means that 
 $ \ell \equiv - \ell \mod 6 $, that is $ \ell = 0 $ or $ \ell = 3 $. The case $ \ell = 0 $ corresponds to 
 a nontrivial kernel of $ D_{\rm{S}} ( \alpha )|_{L^2_{-1} }$ (see \eqref{eq:L2map}). The case $ \ell = 3 $ is
 excluded by the assumption that $ \alpha \notin \mathcal  A $. In fact, we would then have a
 non-trivial kernel of $ D_{\rm{S}} ( \alpha ) |_{ L^2_{2 }} $. But that would mean that we would have
 $v $ with $ D_{\rm{S}} ( \alpha ) v = 0 $,
 \( [ \bar \tau v_1 ( \tau z ) , v_2 ( \tau z ) ]^t = [
   \tau^2 v_1 ( z  ) , \tau^2 v_2 ( z ) ]^t \),
 that is $ v_1 ( \tau z ) = - v_1 ( z ) $ and
 $ v_2 ( \tau z ) = \tau^2 v_2 ( z ) $.
 In particular $ v ( z ) = 0 $ and the theta function argument (see \cite{magic} and \cite[\S 6]{survey})
shows that $ \alpha \in \mathcal  A $ (see Proposition \ref{p:theta}).

\subsection{Dirac points for a discrete set of $ \alpha $'s} \label{s:Dirac}

Suppose that $ \alpha \in {\mathcal B} $ defined in \eqref{eq:defB}. Then 
$ 0 \in \Spec  D_{\rm{S}} ( \alpha ) $ has geometric (and algebraic) multiplicity $ 2 $ and there exist
\begin{equation}
\label{eq:defuv}  \mathbf u \in \ker D_{\rm{S}} ( \alpha )|_{ L^2_0 }  , \ \ \ \mathbf v \in \ker D_{\rm{S}} ( \alpha ) |_{L^2_{-1} }, \ \ 
\langle \mathbf u , \mathbf v \rangle = 0 , \ \ \ \| \mathbf u \| = \| \mathbf v \| = 1 . \end{equation}
We consider the Wronskian of $ \mathbf u $ and $ \mathbf v $,
\begin{equation}
\label{eq:Wr} \Wr ( \mathbf u , \mathbf v ) = u_1 v_2 - u_2 v_1 , \end{equation}
 and note that 
it is non-zero.
Indeed, note $u_1(0) = 0$ and $v_2(0) = 0$ by their symmetries.
If $\Wr(\mathbf u, \mathbf v) = 0$, then either $ \mathbf u $ or $\mathbf v  $ would have a zero
at $ 0 $ and then $ \alpha \in \mathcal  A $ -- see Proposition \ref{p:theta}. We can then set up a Grushin problem for 
$ H ( \alpha, 0 ) $ which remains invertible when we replace $ H( \alpha, 0 ) $ by $ H ( \alpha , k ) $ 
for small $ k $:
\begin{equation}
\label{eq:Grushin1}  
\begin{gathered} \mathscr H (\alpha , k ) = \begin{pmatrix} H ( \alpha, k ) & R_- \\
R_+ & 0 \end{pmatrix} : H^1 ( \mathbb C/\Lambda; \mathbb C^4 ) \times \mathbb C^4 \to 
L^2 ( \mathbb C/\Lambda; \mathbb C^4 ) \times \mathbb C^4 , \\  R_+ := R_-^* , \ \ 
R_-  u_- = u_{-,1} \begin{pmatrix} 0 \\ \mathbf u \end{pmatrix} 
+ u_{-,2} \begin{pmatrix} 0 \\  \mathbf v  \end{pmatrix} + 
u_{-,3} \begin{pmatrix} \mathscr A_2 \mathbf u  \\  0 \end{pmatrix} +
u_{-,3} \begin{pmatrix} \mathscr A_2 \mathbf v  \\  0 \end{pmatrix} . 
\end{gathered}
\end{equation}
Then
\[ \mathscr H ( \alpha, k )^{-1}  = \begin{pmatrix} E ( \alpha, k ) & E_+ ( \alpha, k ) \\
E_-  ( \alpha, k ) & E_{-+}  ( \alpha, k ) \end{pmatrix} ,   \ \ 
E_\pm ( \alpha, 0) = R_{\pm } , \]
where  (see \cite[Proposition 2.12]{notes}) 
\begin{equation}
\label{eq:Emp}  \begin{split} 
E_{-+} ( \alpha, k ) & = - R_+ \begin{pmatrix} 0 & k \\
\bar k & 0 \end{pmatrix} R_- + \mathcal O ( |k|^2 ) \\ & = 
 \begin{pmatrix} 0 &   \bar k A^*  \\  k A  & 0 \end{pmatrix} +\mathcal O ( |k|^2 ) , \ \ \ 
A := \tfrac{\sqrt 3 } 2  W \begin{pmatrix} \ \  0 & 1 \\
- 1 & 0 \end{pmatrix} 
\end{split}
 \end{equation}
 where $ W = W ( \mathbf u , \mathbf v ) $. 
 This follows from the following calculations:
 \[ \begin{split}   & \langle \mathbf u , \mathscr A_2 \mathbf u \rangle = 
 \int_{\mathbb C/\Lambda } u_1  \overline{( - \bar u_2 )}  + u_2 \overline{(\bar u_1 ) } dm ( z) = 0 , 
\ \ \  \langle \mathbf v ,\mathscr A_2 \mathbf v \rangle = 0, \\
  \ \ \
 & \langle \mathbf u , \mathscr A_2 \mathbf v \rangle =  
  \int_{\mathbb C/\Lambda } u_1  \overline{( - \bar v_2 )}  + u_2 \overline{(\bar v_1 ) } dm ( z) = 
\tfrac {\sqrt 3 } 2 \Wr ( \mathbf v , \mathbf u ) = - \langle \mathbf v , \mathscr A_2 \mathbf u \rangle, 
\end{split} 
  \]
and  
\[ \begin{split}
 \begin{pmatrix} 0 & k \\
\bar k & 0 \end{pmatrix} R_- u_- & = 
 \begin{pmatrix} 0 & k \\
\bar k & 0 \end{pmatrix}  \left( u_{-,1} \begin{pmatrix} 0 \\ \mathbf u \end{pmatrix} 
+ u_{-,2} \begin{pmatrix} 0 \\  \mathbf v  \end{pmatrix} + 
u_{-,3} \begin{pmatrix} \mathscr A_2 \mathbf u  \\  0 \end{pmatrix} +
u_{-,3} \begin{pmatrix} \mathscr A_2 \mathbf v  \\  0 \end{pmatrix}\right) \\
& = 
u_{-,1} \begin{pmatrix} k \mathbf u \\ 0  \end{pmatrix} 
+ u_{-,2} \begin{pmatrix}  k  \mathbf v  \\ 0 \end{pmatrix} + 
u_{-,3} \begin{pmatrix} 0 \\ \bar k \mathscr A_2 \mathbf u   \end{pmatrix} +
u_{-,3} \begin{pmatrix}  0 \\ \bar k \mathscr A_2 \mathbf v   \end{pmatrix}
\end{split} \]

To study the eigenvalues near $ 0 $ of $ H ( \alpha, k ) $ for $ |k| \ll 1 $ we observe that 
the Grushin problem \eqref{eq:Grushin1} is well posed when $ H ( \alpha, k ) $ is replaced by
$ H ( \alpha, k ) - \lambda $, $ |\lambda | \ll 1 $. Using \cite[Proposition 2.12]{notes} we then 
see that small eigenvalues of $ H ( \alpha, k ) $ are given by the zeros of 
$ \lambda \to \det ( E^\lambda_{-+}  ( \alpha, k ) ) $ where, in the notation of \eqref{eq:Emp}, 
\[ \begin{split}  E^\lambda_{-+} ( \alpha, k ) & = E_{-+} ( \alpha, k ) - \lambda E_- ( \alpha, k ) E_+( \alpha, k ) + 
\mathcal O ( \lambda^2 ) \\
& =  \begin{pmatrix} 0 &   \bar k A^*  \\  k A  & 0 \end{pmatrix} - \lambda + \mathcal O ( \lambda ^2 + |k|^2 ) .
\end{split}  \]
The eigenvalues of the first term on the right hand side are double and given by $ \pm \frac {\sqrt 3} 2 |k| |W| $ and hence
 (with the labelling convention of \eqref{eq:defEjk})
\begin{equation}   
\label{eq:Dirac}
 E_{\pm j } ( \alpha, k ) = \pm \tfrac {\sqrt 3} 2   | \Wr ( \mathbf u, \mathbf v ) | |k| + \mathcal O ( |k|^2 ) , \ \ \ \alpha \in 
{ \mathcal B}, \ \  j = 1, 2 , \end{equation}
while all the other eigenvalues are non-zero. 
As in \cite{magic}, the Wronskian is interpreted as the Fermi velocity, that is the slope of the Dirac cone.
However, it occurs only for a discrete set of $ \alpha $ at which the cone is double.

When $ \alpha \notin { \mathcal  B} \cup \mathcal A  $, then $ \ker D_S ( \alpha ) = 
\mathbb C \mathbf u $, 
and setting up the Grushin problem using $ ( 0 , \mathbf u )^t $ and $ ( \mathscr A_2 \mathbf u, 0 )^t $, 
shows that 
\begin{equation} 
\label{eq:flat1}
E_{\pm 1 } ( \alpha , k ) = \mathcal O ( |k|^2 ) , \ \  |E_{\pm 2 } ( \alpha, k )| > 0 , \ \ \ 
\alpha  \notin { \mathcal  B} \cup \mathcal A.
\end{equation}

\subsection{Zeros of $ z \mapsto u ( \alpha, z ) $}
\label{s:zeros}

We start with some general considerations and first prove the following fact. {It is an adaptation to the scalar model of the central argument of
\cite{magic}.}
\begin{prop}
\label{p:theta}
In the notation of \eqref{eq:holu} let $ \mathbf u ( \alpha ) = [ 2D_{\bar z } u ( \alpha ) ,  u ( \alpha ) ]^t $.
Then 
\begin{equation}
\label{eq:theta}  \alpha \in \mathcal  A \ \Longleftrightarrow \ 
\exists \, z_0 \in \mathbb C/\Lambda \ \mathbf u ( z_0 ) = 0 . 
\end{equation}
\end{prop}
\begin{proof}
The proof of $ \Longleftarrow $ is the same as the proof for the chiral model in \cite[\S 6]{survey} but we present a slightly modified version for completeness. {It is based on the properties of the resolvent (see \cite[\S 10.2.3]{notes}): 
\[  \begin{gathered} R(  k ) := ( 2 D_{\bar z} +  k )^{-1} :
L^2 ( \mathbb C/\Lambda) \to L^2 ( \mathbb C/\Lambda ) 
, \ \  k 
\notin \Lambda^* ,\\
 R ( k ) f ( z ) = \int_{\mathbb C/\Gamma } G_{ k } ( z - w )  f ( w ) d m ( w ) , \ \ \ 
( 2 D_{\bar z } +  k ) G_{ k } ( z ) = \delta_{0 }( z )  . \end{gathered}\]
The function $ G_{ k} $ is periodic and $ z - w $ is defined 
modulo $ \Lambda $. (We can describe $ G_{k} $ explicitly using 
theta functions but we do not need that.) If $ \mathbf u ( \alpha ) \in \ker D_{\rm{S}} ( \alpha ) $ is the protected state from Proposition \ref{p:simple} then for $ k \notin \Lambda^* $, 
\[  ( D_{\rm{S}}(\alpha ) + k ) ( G_k ( z - z_0 ) \mathbf u ( \alpha, z ) ) = \delta_{z_0 } ( z ) \mathbf u ( z_0 ) . \]
Hence if $ \mathbf u $ has a zero at $ z_0 $ we obtain an eigenfunction
for every $ k $.}

Now suppose that 
$ \alpha \in \mathcal  A $ and 
$ \mathbf u $ does {\em not} vanish. Since $ 
\Spec D_{\rm{S}} ( \alpha ) = \mathbb C $, for every $ k $ (or equivalently some $ k \notin \Lambda^* $)
there exists $ \mathbf v \neq 0 $ such that $( D_{\rm{S}} ( \alpha ) + k ) \mathbf v = 0 $. Then the
Wronskian \eqref{eq:Wr} is seen to satisfy $ 2 D_{\bar z } \Wr  = - k \Wr  $.
The general solution of this equation is given by 
\[ \Wr ( z, \bar z ) = e^{ \frac i 2 ( k \bar z + \bar k z ) } w ( z) , \ \ \ w \in \mathscr O ( \mathbf C ) . \]
Since $ \Wr $ is periodic and $ z \mapsto e^{ \frac i 2 ( k \bar z + \bar k z ) }  $ is a bounded function,
$ w $ has to be constant, and for $ k \notin \Lambda^* $ that constant has to vanish.

{The vanishing
 of $ \Wr $ means that $ \mathbf u $ and $ \mathbf v $ are parallel and hence, since $ \mathbf u $ is assumed not to vanish, there exists $ F \in C^\infty ( \Omega )
$ such that $ \mathbf v = F \mathbf u $. Applying $ 2 D_{\bar z}$ to both sides shows that
$ 2 D_{\bar z } F =  - k F $.  
Hence 
for some $ f \in C^\infty ( \CC  ) $,  
\begin{equation}
\label{eq:u2v}  \mathbf v( z, \bar z )  = e^{ - \frac i 2 ( k \bar z + \bar k z ) } f ( z) \mathbf u ( z, \bar z ) ,  \ \ \ \ 
\partial_{\bar z } f = 0 . \end{equation}
Since $ \mathbf v  $ and $ \mathbf u $ are be periodic, and 
$ \mathbb C \ni z \mapsto  e^{ - \frac i 2 ( k \bar z + \bar k z ) } $ is bounded (but not periodic for $ k \notin \Lambda^*$), we conclude that $ f $ is bounded and hence constant. But then 
$ \mathbf v $ is not periodic, wich provides a contradiction.}
\end{proof}

We have the following
refinement of Proposition 1. To formulate it, we define the notion of multiplicity:
\begin{defi}
For a solution $ \mathbf u \in \ker D_{\rm{S}} ( \alpha ) $ we define the multiplicity of 
a zero at $ z_0 $ as 
\begin{equation}
\label{eq:defM}    M_{\mathbf u }  ( z_0 ) := \max \{ m :  [ \partial_z^{m-1} \mathbf u ] (z_0 ) = 0 \} , 
\end{equation}
with the convention that $ M_{\mathbf u } ( z_0 ) = 0 $ if $ \mathbf u ( z_0 ) \neq 0 $. 
\end{defi}

Since
$ \mathbf u = [ 2 D_{\bar z } u , u ]^t $, $ P ( \alpha ) u = 0$,  and  $
 m = M_{\mathbf u }  ( z_0 ) $ is equivalent to $ u ( z) = 
( z - z_0)^m v ( z ) $, $ v ( z_0 ) \neq 0 $, $ v \in C^\infty ( \neigh_{\mathbb C } ( z_0 ) )$. {{We note here
that $ \mathbf u ( z_0 ) = 0 $ and $ D_{\rm{S}} ( \alpha ) \mathbf u (z ) = 0$ imply that
$ \partial_{\bar z}^k \mathbf u ( z_0 ) = 0 $ for all $ k $. That means $ \partial_{\bar z }^k u ( z_0 ) = 0 $ for all $ k $
and hence we can factor out a positive power of $ ( z - z_0 ) $.}

\begin{prop}
\label{p:zeros}
In the notation of \eqref{eq:holu} let $ \mathbf u ( \alpha )  = [ 2D_{\bar z } u ( \alpha ) ,  u ( \alpha ) ]^t $.
For $ \alpha \in \mathcal  A $
\begin{equation}
\label{eq:zeros}  \dim \ker D_{\rm{S}} ( \alpha ) = \sum_{ z \in \mathbb C/\Lambda } M_{\mathbf u ( \alpha ) } ( z ) =: m ( \alpha ) . \end{equation}
This means that for all $ k \in \mathbb C /\Lambda^* $, 
\[  \dim \ker H ( \alpha, k ) = 2 m ( \alpha ) , \]
which is the number of flat bands at $ 0 $.
\end{prop}
\begin{proof}
We  prove that for $ \alpha \in \mathcal  A $,  $ \dim \ker D_{\rm{S}} ( \alpha ) \geq m ( \alpha )$. For that suppose first that $ \mathbf u ( \alpha ) $ has distinct zeros $ z_0 , \cdots , z_{m-1} $ 
each of multiplicity one. 
We can assume that $ m \geq 2 $ as otherwise there is nothing to prove. We then choose
\[  w_j \notin \{ z_j \}_{j=0}^{m-1} , \ \ \ \ \ \sum_{ j=0}^{m-1} w_j \equiv  \sum_{j=0}^{m-1}  z_j , 
\]
and put
\[ \mathbf u_j ( z, \bar z ) = \frac{ \theta ( z - w_0 ) \theta ( z - w_j ) }{
\theta ( z - z_0 ) \theta ( z- z_j ) }  \mathbf u ( z, \bar z ). \]
The factor in front of $ \mathbf u $ is a meromorphic function with simple poles at 
$ z_0 $ and $ z_j $ (see \cite[\S 3.1]{bhz2} for the review of properties of $ \theta $ in a related context)
but the properties of the zeros show that $ \mathbf u_j \in \ker D_{\rm{S}} ( \alpha ) $. The functions
$ \mathbf u_j $ are linearly independent as 
\[ c_0 + \sum_{ j=1}^{n-1} c_j  \frac{ \theta ( z - w_0 ) \theta ( z - w_j ) }
{ \theta ( z - z_0 ) \theta ( z- z_j ) } \equiv 0 \ \ \Longrightarrow \ c_0 = c_1 = \cdots = c_{m-1} = 0.\]
(Put $ z = z_j $, $ j > 0 $ to see that $ c_j = 0 $ for $ j > 0 $.) The case of higher multiplicities is dealt with similarly
but using meromorphic functions with poles of higher order. This proves that $ \dim \ker D_S ( \alpha ) \geq m ( \alpha ) $.

To prove the opposite inequality, suppose that $ \mathbf u $, $ \mathbf v_j $, $ 1 \leq j\leq M-1 $,  
span $ \ker D_{\rm{S}} ( \alpha ) $, $ M = \dim \ker D_{\rm{S}} ( \alpha ) $. Then 
the Wronskians \eqref{eq:Wr}, $ \Wr( \mathbf u , \mathbf v_j ) = 0 $ as $ \mathbf u $ has to vanish somewhere
by Proposition \ref{p:theta} and the Wronskian is constant.
 But as in \eqref{eq:u2v}, 
\begin{equation}
\label{eq:u2vj}   \mathbf v_j ( z , \bar z ) = f_j ( z ) \mathbf u ( z, \bar z  ) , \ \text{ \ \  $ f_j ( z ) $, 
$ 1 \leq j \leq m - 1 $, meromorphic,} 
\end{equation} 
 $ f_0 \equiv 1 $ and $ \{ f_j \}_{j=0}^{M-1} $ 
linearly independent. The span of $ f_j$'s is contained in the space of meromorphic functions\footnote{For an accessible introduction see Terry Tao's blog {\url{https://terrytao.wordpress.com/2018/03/28/246c-notes-1-meromorphic-functions-on-riemann-surfaces-and-the-riemann-roch-theorem/}} and the version of the Riemann--Roch theorem needed here is in formula (7) there.} $ L ( D ) $
 where 
$ D $ is the divisor defined by the zeros of $ \mathbf u  $ included according to their multiplicities
defined in \eqref{eq:defM}. Since $ \dim L ( D ) = \deg D = m( \alpha )  $ (see \eqref{eq:zeros}), 
this means that the number of zeros of $ \mathbf u $ has to be greater than or equal to
$ M $, the dimension of $ \ker D_{\rm{S}} ( \alpha ) $. 
\end{proof}

\noindent
{\bf Remarks.} 1. The same argument applies in the chiral model and settles
\cite[Problem 14]{survey} -- see \cite[Theorem 1]{bhz3}.

\noindent
2. Although we invoked a basic version of the Riemann--Roch theorem the proof that the number of poles
of $ f_j$'s has to be greater than $ M$ is explicit. If the poles are all simple then
\[ f_j ( z ) = \sum_{ k=1}^{K(j)} \lambda_k (j) \frac{\theta' ( z - a_k ( j ) )  }{ \theta ( z - a_k (j ) ) } + c(j), \ \ \ \ 
\sum_{ k=1}^{K(j)} \lambda_k (j) = 0 .\]
For $ f_0 \equiv 1$, $ f_1, \cdots , f_{M-1} $ to be linearly independent we need 
the number of distinct zeros $ \{ a_k ( j ) : 1 \leq k \leq K( j ) , 1 \leq j \leq M-1 \} $ to be $ \geq M $. It is not difficult to modify this construction to the case of poles of higher order.

\begin{figure} 
\centering
  \includegraphics[width=13cm]{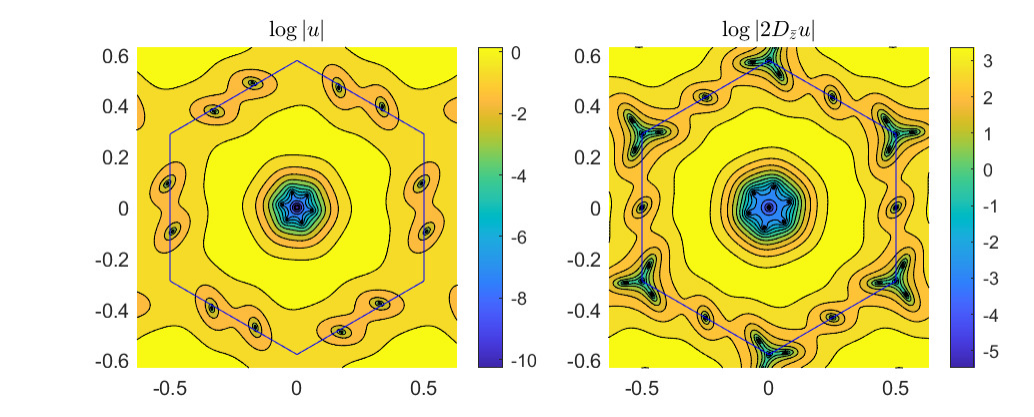}
\caption{\label{f:m1} When multiplicity of the band is one, the only possible zero of $ \mathbf u $ is at  $ 0$ 
and that is illustrated in this figure. For the potential \eqref{eq:defU} in the scalar model with $ V ( z ) = U_{\rm{BM}} (z ) 
U_{\rm{BM}} ( - z ) $ this occurs only for complex $ \alpha$'s.}
 \end{figure}

\subsection{Multiplicity two}
\label{s:mult2}

Suppose that $ \dim \ker D_S ( \alpha )  = 2 $ and $ \alpha \in \mathcal  A $. 
We remark that numerical analysis of the Bistritzer--MacDonald potential
\eqref{eq:defU2} (that is, $ V ( z ) = U ( z ) U ( -z ) $ with that $ U $) suggest this is the case 
for $ \mathcal  A \cap \mathbb R $.

The assumption  $ \dim \ker D_S ( \alpha )  = 2 $ means that 
$ \dim \ker H( \alpha, 0 ) = 4 $ and consequently $ \dim \ker H ( \alpha, k ) = 4 $ for all $ k $. 
Since $ \dim \ker H_\ell ( 0 , 0 ) \equiv \dim \ker H_\ell ( \alpha, 0) \!\! \mod 2 $, in view of \eqref{eq:dimker0}, 
there exists $ \ell \neq \pm 1 $ such that $ \dim \ker H_{\ell } ( \alpha, 0 ) = 2 $. But then \eqref{eq:defA2}
and \eqref{eq:L2map} imply that $ \ell \equiv - \ell \!\!  \mod 6 $, that is, $ \ell = 0, 3 $.

The kernel of $ D_{\rm{S}} ( \alpha ) $ is consequently spanned by $\mathbf u \in L^2_0 ( \mathbb C/\Lambda; 
\mathbb C^2 ) $ and $ \mathbf v \in L^2_\ell ( \mathbb C/\Lambda; 
\mathbb C^2 )$, $ \ell $ is either $ -1  $ (the case of  $\dim \ker H_{0 } ( \alpha, 0 ) = 2$)
 or $ 2 $ (the case of  $\dim \ker H_{3 } ( \alpha, 0 ) = 2$). 
The Wronskian \eqref{eq:Wr}, $ \Wr ( \mathbf u , \mathbf v ) = 0 $ as otherwise we would have 
\eqref{eq:Dirac} and no flat bands. We note that $ u_1 ( \pm z_S ) = u_1 ( 0 ) = 0 $, 
and $ v_2 ( \pm z_S ) = v_2 ( 0 ) = 0 $ (for both $ \ell = -1 $ and $ \ell = 2 $). Hence
\[  u_2 ( \pm z_S ) v_1 ( \pm z_S ) = u_2 ( 0 ) v_1 (0) = 0.\]
Hence, we have two exclusive cases ($ \mathbf u $ cannot vanish at three points as then Proposition 
\ref{p:zeros} would contradict the multiplicity two assumptions):
\begin{equation}
\label{eq:choice}  \mathbf u ( 0 ) =0  \ \text{ or } \ \mathbf u ( \pm z_S ) = 0 
\end{equation}
Suppose first that $ \mathbf u ( 0 ) = 0 $ and $ \mathbf u  ( \pm z_S ) \neq 0 $. Then $ \mathbf v ( \pm z_S ) = 0 $, 
and we can use the following 
 well known 
\begin{lemm}
\label{l:well}
The unique (up to a multiplicative constant) meromorphic function with two poles at $ \pm z_S$ on $ \mathbb C/\Lambda $ 
and a (double) zero at $ 0 $
satisfies
\begin{equation}
\label{eq:tauF}  F ( \tau z ) = \tau^2 F ( z ), \ \ \ \tau = e^{ - \pi i /3 } .  
\end{equation}
\end{lemm}
\begin{proof}
We first note that $ F $ is unique as for any periodic meromorphic function the sum of the poles is congruent to the sum of the zeros and the sum of residues adds up to $ 0 $. We can then give an explicit formula for 
$ F$:
\[ 
\begin{split} F ( z) & = \sum_{ \gamma \in \Lambda } \frac{1}{ z + \gamma - z_S } 
- \frac{1}{ z + \gamma + z_S } - \frac{1}{ \gamma - z_S } 
 + \frac 1{ \gamma + z_S } 
\\
& = - 2 z_S \sum_{ \gamma \in \Lambda } \frac{ z^2 + 2 z \gamma }{ ( z+ \gamma +z_S ) ( z + \gamma - z_S )  (\gamma - z_S )( \gamma + z_S )} \\
\end{split} 
\]
This shows absolute convergence and $ F ( z ) = F ( - z )$. Hence to obtain \eqref{eq:tauF} it suffices
to show that $ F ( \omega z ) = \bar \omega F ( z ) $. In fact, $ \omega z_S = z_S + \gamma_0 $, 
$ \gamma_0 \in \Lambda $, and 
\[ 
\begin{split} F ( \omega z) & = \bar \omega \sum_{ \gamma \in \Lambda } \frac{1}{ z +  \gamma - \gamma_0  - z_S  } 
- \frac{1}{ z + \gamma + \gamma_0 + z_S  } - \frac{1}{ \gamma - \gamma_0 - z_S } 
 + \frac 1{ \gamma + \gamma_0  + z_S }  \\
 & = \bar \omega  F ( z ) + \sum_{ \gamma \in \Lambda } \frac{1}{ z + \gamma + z_S  } 
- \frac{1}{ z + \gamma + 2 \gamma_0 + z_S  } - \frac{1}{ \gamma  + z_S } 
 + \frac 1{ \gamma + 2 \gamma_0  + z_S } .
 \end{split} \]
The sum on the right hand side vanishes as it converges uniformly and defines a periodic holomorphic function vanishing at $ 0 $. 
This completes the proof. 
\end{proof}

\noindent
{\bf Remark.} The function $ F $ can be expressed using theta functions (see \cite[\S 3.1]{bhz2} for the definitions in the this context) 
\[ F ( z ) = \frac{ \theta ( z )^2 }{ \theta ( z - z_S ) \theta ( z+ z_S )}, \]
and the conclusion follows from \cite[(3.1)]{bhz2} and the fact that $ \bar \omega z_S = z_S + \gamma_0 $, 
$ \gamma_0 \in \Lambda $.

Since $ F ( z ) \mathbf v $ is not a scalar multiple of $ \mathbf v $, we have $ \mathbf u = c_0 F ( z ) \mathbf v \in L^2_{\ell + 2 } $. But for $ \ell = -1, 2 $, 
$ \ell + 2 \not \equiv 0 \mod 6 $ which is a contradiction. Hence only the second option in \eqref{eq:choice} is allowed. But then \eqref{eq:tauF} implies that $ \mathbf v = c_0 F ( z ) \mathbf u \in L^2_{2} $. 
We summarise this with 
\begin{prop}
\label{p:m2}
The multiplicity of $ \alpha \in \mathcal  A $ is 2 (that is, $ \dim \ker H( \alpha, k ) = 4 $
for all $ k $) if and only if $ \mathbf u ( \alpha ) = [ 2 D_{\bar z} u ( \alpha) , u ( \alpha )  ]^t $ in \eqref{eq:holu} 
vanishes at $ \pm z_S $. In that case, 
\[ \ker D_S ( \alpha ) = \mathbb C \mathbf u \oplus \mathbb C \mathbf v , \ \ \ \mathbf u \in L^2_0 ( \mathbb C/\Lambda;
\mathbb C^2 ) ,  \ \ \mathbf v \in L^2_2 ( \mathbb C/\Lambda;
\mathbb C^2 ) .\]
\end{prop}

\noindent
{\bf Remark.} We observe that in the case of $ m ( \alpha ) = 2 $ we can write
\begin{equation}
\label{eq:ker2th}  \ker D_{\rm S } ( \alpha ) = \{ c \frac{ \theta ( z - z_1 ) \theta ( z + z_1 ) } { 
\theta ( z - z_S ) \theta ( z + z_S ) } \mathbf u ( \alpha ) , \ \ z_1 \in \mathbb C/\Lambda, \ \ 
c \in \mathbb C \} . \end{equation}
The Weierstrass identity \cite[(3.3)]{voca} 
gives 
\[ \theta ( z - z_S ) \theta ( z + z_S ) \theta ( z_1 )^2 + \theta ( z_S - z_1 ) \theta ( z_S + z_1 ) 
\theta ( z )^2 =  \theta ( z_S )^2 \theta ( z - z_S ) \theta ( z + z_S ) \]
which means that, with $ \mathbf v = F  \mathbf u $, ($ F ( z)$ defined in Lemma \ref{l:well})
\[  \frac{ \theta ( z - z_1 ) \theta ( z + z_1 ) } { 
\theta ( z - z_S ) \theta ( z + z_S ) } \mathbf u  
= \frac{\theta ( z_1)^2}{ \theta(z_S )^2}  \mathbf u + \frac{ \theta ( z_S - z_1 ) \theta ( z_S + z_1 ) }
{\theta( z_S )^2 } \mathbf v . \]

\subsection{Local structure of solutions}
\label{s:loc}

We start with some general comments: we consider the following {\em non-holomorphic} embedding of 
$ \mathbb R^2 \simeq \mathbb C  $ in $ \mathbb C^2 $:
\[ F:  \mathbb C \to \mathbb C^2 , \ \ \ F: z \mapsto ( z, \bar z ) \in \mathbb C^2 . \]
The image of $ F$,  $ X := \{ ( z, \bar z ) : z \in \mathbb C \} $, is a totally real linear
subspace (over $ \mathbb R $) of $ \mathbb C^2 $ in the sense that $ X \cap i X = \{ 0 \} $. 
In particular, (we use the notation $ \neigh_{\Omega}  ( \Omega_0 ) $ for an open neighbourhood of $ \Omega_0 $ in 
$ \Omega $, $ \mathscr O ( \Omega ) $ denotes holomorphic functions in $ \Omega  $, $ \zeta = ( \zeta_1, \zeta_2 
\in \mathbb C^2 $)
\begin{equation}
\label{eq:w2z} 
\begin{gathered} u \in \mathscr O ( \neigh_{\mathbb C^2 } ( X ) ) \ \Longrightarrow \ \left\{ 
\begin{array}{l} 
\partial_{\zeta_1}  u ( \zeta  ) |_{ \zeta = ( z, \bar z ) } = \partial_z [ u ( z, \bar z ) ] , \\
\partial_{\zeta_2} u ( \zeta ) |_{ \zeta = ( z, \bar z) } = \partial_{\bar z } [ u ( z, \bar z ) ] ,
\end{array} \right. \\
u \in \mathscr O ( \neigh_{\mathbb C^2 } ( X ) ), \ \ u|_X  = 0 \ \Longrightarrow  \ 
u \equiv 0 . 
\end{gathered} 
 \end{equation}

With this in mind, we consider the following problem (we drop $ h $ or $ \alpha = 1/h $ which can be put into $ V $):
\begin{equation}  
\label{eq:system} 2  D_w \mathbf v_j  ( z, w ) = \begin{pmatrix} 0 & \alpha V ( z, w )  \\
\alpha  & 0 \end{pmatrix} \mathbf v_j , \ \ \ \mathbf v_j ( z , \bar z_0  ) = \mathbf e_j ,
\end{equation}
where $ \mathbf e_1 = [ 1 , 0 ]^t$, $ \mathbf e_2 = [0 , 1 ]^t $, for $ z $ near $ z_0 $ and 
$ w $ near $ \bar z_0 $.

\begin{lemm}
\label{l:locODE}
The system \eqref{eq:system} has unique solutions 
$ \mathbf v_j \in \mathscr O ( {\rm{neigh}}_{ \mathbb C^2 } ( z_0 , \bar z_0 ) ; \mathbb C^2 ) $.
\end{lemm}

\begin{proof}
For a fixed $ z $, the standard theory shows existence of unique solutions which are
holomorphic in $ w $. (If $ w \mapsto V ( w, z ) $ is entire, we have a global solution in $ w $ -- see \cite[\S 7.2]{SjBook}.) We now see that
$ \partial_{\bar z } \mathbf u_j $ solves the equation \eqref{eq:system} but 
with the initial condition $ \partial_{\bar z } \mathbf u_j ( z , \bar z_0 ) = 0 $. But that means that
$ w \mapsto \partial_{\bar z } \mathbf u_j ( w, z) = 0 $. 
\end{proof}

\begin{lemm}
\label{l:locPDE}
There exists $ \mathbf u_j \in \mathscr O ( 
\neigh_{ \mathbb C^2 } ( z_0, \bar z_0 ) ) $ solving 
\[ 2  D_w \mathbf u_j  ( z, w ) = \begin{pmatrix} 0 & \alpha V ( z,w )  \\
 \alpha  & 0 \end{pmatrix} \mathbf u_j , \]
such that any solution $ \mathbf u \in C^\omega ( \neigh_{\mathbb C} ( z_0 ) ) $ solving
$ D_{\rm{S} } (\alpha ) \mathbf u = 0 $, (see \eqref{eq:syst0}) is given by 
\begin{equation}
\label{eq:u2c12}  \mathbf u ( z, \bar z ) = c_1 ( z ) \mathbf u_1 ( z, \bar z ) + c_2 ( z ) \mathbf u_2 ( z , \bar z ) , \ \ 
c_j \in \mathscr O ( {\neigh_{\mathbb C} ( z_0 ) }) . \end{equation}
\end{lemm}

\noindent 
{\bf Remark.} The assumption of analyticity of $ V $, made throughout this paper is not necessary here.
Using \cite[Proposition 19.1]{KS} applied in the simplest case of $ m = 1 $ gives the lemma for 
$ V \in C^\infty ( \mathbb C ) $ without the necessity of complexifying $ \mathbb C $ to $ \mathbb C^2 $. 
 Consequently, 
the conclusion of this section work for a more general class of potentials. We opt for a more elementary 
approach under the stronger assumption of real analyticity. 

\begin{proof}[Proof of Lemma \ref{l:locPDE}] 
We have many different choices for $ \mathbf u_j $ and one is given in Lemma \ref{l:locODE}. 
We note that $ \mathbf v_j ( z_0, \bar z_0  ) = \mathbf e_j $ and hence 
$ \mathbf v_j ( z, \bar z ) $ are linear independent for $ z \in \neigh_{\mathbb C} ( z_0 ) $. 
If $ D_{\mathrm{S} } \mathbf u = 0 $ near $ z_0 $ then $ u $ is real analytic in 
$ \neigh_{\mathbb C} ( z_0 ) \simeq \neigh_X ( z_0 , \bar z_0 ) $
(by the standard ellipticity result for operators with analytic coefficients). Hence it extends to 
$ \widetilde { \mathbf u } \in \mathscr \neigh_{\mathbb C^2 } ( z_0 , \bar z_0 ) $ and (as in
\eqref{eq:w2z}) 
\[  2 D_{w } \widetilde {\mathbf u } ( z, w) = \begin{pmatrix} 0 & \alpha V ( z, w )  \\
\alpha & 0 \end{pmatrix} \widetilde {\mathbf u } ( z, w)   . \]
Hence for $ ( z, w ) \in \neigh_{\mathbb C^2 } ( z_0 , \bar z_0 ) $, 
\[ \widetilde  {\mathbf u } ( z, w ) = c_1 ( z ) \mathbf u_1 ( z, w ) + c_2 ( z ) \mathbf u_2 ( z, w ) , \]
and, as $ \widetilde  {\mathbf u } ( z , w ) $ is holomorphic in both variables, the coefficients 
$z \mapsto c_j $ are also holomorphic. Restricting to $ X $ gives \eqref{eq:u2c12}.
\end{proof}

\subsection{Rank-2 holomorphic vector bundle}
\label{s:rank2}
The local structure of solutions shows the existence of the following rank 2 topologically trivial 
holomorphic vector bundle, $ \mathscr E  $, over $ \mathbb C/\Lambda $. We have 
\[ \mathscr E  = \mathbb C/\Lambda \times \mathbb C^2 , \ \ \ \pi : \mathscr E   \mapsto \mathbb C /\Lambda . \]
To define the following covering of $ \mathbb C/\Lambda $: for $ z_0 \in \mathbb C^2 $ we choose
$ U = \neigh_{\mathbb C/\Lambda } ( z_0 ) $ such that \eqref{eq:u2c12} holds with some 
linearly independent (at each $ z \in U $), $ \mathbf u_j $, $ j = 1,2 $. We then define
\[  \begin{gathered} 
g_U = g_{U, \mathbf u_1, \mathbf u_2 } : U \times \mathbb C^2 \to \pi^{-1} (U) , \\
 g_{U} ( z  , \zeta  ) \mapsto ( z , \zeta_1 \mathbf u_1 ( z ) + \zeta_2  \mathbf u_2 ( z ) ) \subset 
 \mathbb C/\Lambda \times \mathbb C^2 , \end{gathered} \]
 noting that, as $ \mathbf u_j ( z ) $ are independent, this is a map is an isomorphism 
 $ \mathbb C^2 \to \pi^{-1} ( z ) $. 
If $ V $  is another neighbourhood with a local basis of the kernel of $ D_S $ over $ V $ given by  $\mathbf v_1 , \mathbf v_2  $, then Lemma \ref{l:locPDE} shows that
\[ \mathbf v_j ( z, \bar z ) = c_{1j}  ( z ) \mathbf u_1 ( z, \bar z ) + c_{2j} ( z) \mathbf u_2 ( z, \bar z ) , \ \ \ 
c_{kj} \in \mathscr O ( U \cap V ) . \]
 Hence for $ z \in U \cap V $, 
\[   g_U^{-1}  \circ g_V |_{ \{ z \}  \times \mathbb C^2}  ( \zeta ) = 
C ( z ) \zeta  , \ \  C ( z ) := \begin{pmatrix} c_{11} ( z ) & c_{12} ( z ) \\
c_{21} ( z ) & c_{22} ( z ) \end{pmatrix} , \]
is holomorphic.

If we now include $ \alpha $ we obtain a family of holomorphic vector bundles $ \mathscr E  ( \alpha ) $:
their local sections correspond to local solutions to  $ D_{\rm{S} } ( \alpha ) \mathbf u = 0 $.

We can describe $ \mathscr E  ( \alpha ) $ at different values of $ \alpha $ and see how it fits in the 
classification of holomorphic rank 2 vector bundles over tori -- see \cite[\S 4]{daly} for a gentle introduction
to that theory. 

We start with the following observation:
\begin{equation}
\label{eq:det2tr}
\text{ The determinant line bundle, $ \wedge ^2 \mathscr E  ( \alpha ) $, is trivial.}
\end{equation}
\begin{proof}[Proof of \eqref{eq:det2tr}] In our choice of $ \mathbf u_j (z) $, defined in $ U$, we can 
always arrange that in $ U $, $  \Wr ( \mathbf u_1 , \mathbf u_2 ) := \mathbf u_1 \wedge \mathbf u_2 \equiv 1 $. In fact, since $ D_{\bar z }  \Wr ( \mathbf u_1, \mathbf u_2 ) = 0$, the Wronskian is holomorphic 
and non-vanishing, so we can replace $ \mathbf u_1 $ by $ \mathbf u_1 / \Wr ( \mathbf u_1, \mathbf u_2 ) $. 
The transition matrix is given by 
\[ C ( z ) = \begin{pmatrix}  \ \ \mathbf v_1 \wedge \mathbf u_2  & \ \ \mathbf v_2 \wedge  \mathbf u_2   \\
 - \mathbf v_1 \wedge  \mathbf u_1  & - \mathbf v_2\wedge  \mathbf u_1 \end{pmatrix} , \ \ \ z \in U \cap V . \]
But the transition factor for $ \wedge ^2 \mathscr E  ( \alpha ) $ is then 
\[ \det C ( z ) = \det \begin{pmatrix}  \mathbf v_1 \wedge \mathbf u_1  &  \mathbf v_1 \wedge  \mathbf u_2  \\
 \mathbf v_2 \wedge  \mathbf u_1  &  \mathbf v_2 \wedge  \mathbf u_2 \end{pmatrix}  =  ( \mathbf v_1 \wedge \mathbf v_2 ) ( \mathbf u_1 \wedge \mathbf u_2 ) \equiv 1. \]
This means that constants are global sections.
\end{proof}

\noindent
{\bf Case 1:} $ \alpha \notin  { \mathcal B } \cup \mathcal A$. In this case
there exists a unique global section satisfying $ D_{\rm S } ( \alpha )  \mathbf u = 0$, 
$ \mathbf u \in C^\omega ( \mathbb C/\Lambda ) $. We note that $ \mathbf u $ is a holomorphic
section of $ \mathscr E  ( \alpha ) $ but obviously it is {\em not} a holomorphic function on $ \mathbb C/\Lambda $.
Since $ \alpha \notin \mathcal  A $, $ \mathbf u $ does not vanish and hence it defines a trivial
line subbundle, $ L ( \alpha ) $ of $ \mathscr E ( \alpha ) $. The quotient line bundle $ \mathscr E  ( \alpha ) / L ( \alpha ) $
has an explicit representation based on the proof of Proposition \ref{p:simple}. We showed there
that for $ \alpha \notin  { \mathcal B } \cup \mathcal A  $, there exists $ 
\mathbf v $ such that $ D_{\rm S} ( \alpha ) \mathbf v = \mathbf u $. We also remark that 
$ \mathbf v (z) $ is independent of $ \mathbf u ( z) $ at every point. In fact, in the notation of \eqref{eq:Wr}, 
\[ \begin{split}  2 D_{\bar z } W ( \mathbf u, \mathbf v ) & = 
2 D_{\bar z }  u_1 v_2 + 2 D_{\bar z } v_2 u_1 - 2 D_{\bar z } u_2 v_1 - 2 D_{\bar z } v_1 u_2 
\\& = 
\alpha V u_2 v_2 + ( v_1 + u_2 ) u_1 - u_1 v_1 - ( \alpha V v_2 + u_1 ) u_2 = 0 .
\end{split} \]
Hence dependence at some point would imply that  $ W ( \mathbf  u, \mathbf v ) \equiv 0 $.
Arguing as in \eqref{eq:u2v}, $ \mathbf u ( z, \bar z ) = f ( z ) \mathbf v ( z, \bar z ) $ for 
some meromorphic function $ f $. But then $ \mathbf u $ would have zeros which contradicts 
the assumption $ \alpha \notin \mathcal A$ and Proposition \ref{p:theta}.
 
We now define the following element of $ C^\infty ( \mathbb C; \mathbb C^2 ) $:
\begin{equation}
\label{eq:deftu}    \widetilde {\mathbf u }  ( z  )  := \tfrac{ i } 2 \bar z  \mathbf u - \mathbf v  \end{equation}
which satisfies
\begin{equation}
\label{eq:Dtu}   D_{\rm{S}} ( \alpha )  \widetilde {\mathbf u }   = 0 , \ \ \ 
\widetilde {\mathbf u } ( z + \gamma ) =  \widetilde {\mathbf u } ( z) +  \tfrac{ i } 2  \bar \gamma \mathbf u ( z) . 
\end{equation}
Since $ \mathbf u ( z ) $ are $ \mathbf v ( z )  $ are independent, 
$ \mathbf u $ and $ \widetilde{\mathbf u}  $ can be used as a basis, 
near $ z_0 $ and over $ \mathscr O ( \neigh_{\mathbb C/\Lambda } ( z_0 ) ) $, of local solutions to $ D_{\rm{S}} ( \alpha ) \mathbf w = 0 $.

From \eqref{eq:Dtu} we see that $ \widetilde {\mathbf u } $ defines a non-vanishing
 section of $ \mathscr E ( \alpha ) / L ( \alpha ) $.
To see that $ \mathscr E (\alpha) $ is indecomposable, suppose that $ \mathscr E (\alpha) = L \oplus M $. Since $ \mathscr E (\alpha) $ and $ L $ are trivial 
so is  $M $ which consequently 
has a global nonvanishing section, $ \mathbf w $. Then $ [ \mathbf w ] 
= c [\widetilde { \mathbf u } ]$  in $ \mathscr E (\alpha)/L $ and we can assume that assume $ c = 1 $. This means that 
\[ \mathbf w ( z, \bar z )=  \widetilde{ \mathbf  u} ( z, \bar z ) +
F ( z, \bar z ) \mathbf u ( z, \bar z ) .\]
Applying the equation we see that $ D_{\bar z } F = 0 $, while \eqref{eq:Dtu} gives
$ F ( z + \gamma ) = F ( z ) - i  \bar \gamma /2 $. Consequently
$ f(z, \bar z ) :=  F ( z ) +   i  \bar z /2$ 
is periodic and solves $ \partial_{\bar z } f = i/2  $. But that is impossible 
as we see by integrating both sides over $ \mathbb C/\Lambda $. 
 
 \noindent
 {\bf Case 2:} $ \alpha \in {\mathcal B } $. This is the easiest case as the proof of 
 Proposition \ref{p:simple} shows that $ \ker D_{\rm{S}} ( \alpha ) $ is spanned by two independent
 solutions $ \mathbf u \in L^2_0 ( \mathbb C/\Lambda; \mathbb C )  $ and $ \mathbf v \in 
  L^2_{-1} ( \mathbb C/\Lambda; \mathbb C )$ (see \eqref{eq:defuv}). In particular, they never vanish and define two trivial line bundles, $L_0 (\alpha ) $ and $ L_{-1} ( \alpha ) $ and
  \[  E ( \alpha ) = L_0 ( \alpha ) \oplus L_{-1} ( \alpha ) . \]
  
  \noindent
  {\bf Case 3:} $ \alpha \in \mathcal  A $. In view of Proposition \ref{p:zeros}, 
the dimension of the space of holomorphic sections of $ \mathscr E  ( \alpha ) $ is $ m ( \alpha ) $, 
which is also the number of zeros of each section. From \eqref{eq:u2vj} we see that 
at each point $ z $,  $ \{ \mathbf w ( z ) : \mathbf w \in \ker D_{\rm{S}} ( \alpha ) \} \simeq \mathbb C $, 
and hence $ \ker D_{\rm{S}} ( \alpha ) $ defines a line bundle, $ L ( \alpha ) $  over $ \mathbb C/\Lambda $. The proof of Proposition \ref{p:zeros} shows that this is the line bundle associated to 
the divisor given by the zeros of any fixed element of $ \ker D_{\rm{S}} ( \alpha )$.

In the notation of
\cite[Definition 2.18]{daly}, $ \mathscr E  $ (we drop $ \alpha $ here) is decomposable if 
$ {\rm{Ext}}^1 ( \mathscr E /L, L ) = \{ 0 \} $. Since $ \mathscr E  $ is trivial, the degree of $ \mathscr E /L $ is equal to $ - m $, where $ m > 0 $
is the degree of $ L $. We have, by Serre's duality (\cite[Proposition 2.21]{daly}, with trivial $ K_X $ for tori), $ {\rm{Ext}}^1 ( \mathscr E /L , L ) \simeq H^1 ( ( \mathscr E /L)^* \otimes L ) \simeq  H^0 ( ( \mathscr E /L ) \otimes L^* )^* $. But the last space is trivial as 
the degree of the line bundle $ ( \mathscr E /L ) \otimes L^*  $ is $ - 2 m < 0 $. 
 Hence, indeed, $ \mathscr E  $ is decomposable.

As noted in \eqref{eq:det2tr}, the determinant bundle of $ \mathscr E  ( \alpha ) $ is trivial, and if 
$ \mathscr E  ( \alpha ) = L ( \alpha ) \oplus M ( \alpha ) $ then 
$ M ( \alpha ) = \det \mathscr E  ( \alpha )  \otimes L ( \alpha)^* = L( \alpha )^* $. Hence we can identify $ \mathscr E  ( \alpha ) $ with
$ L ( \alpha ) \oplus L ( \alpha )^* $.

\begin{figure} 
\centering
 \includegraphics[width=13cm]{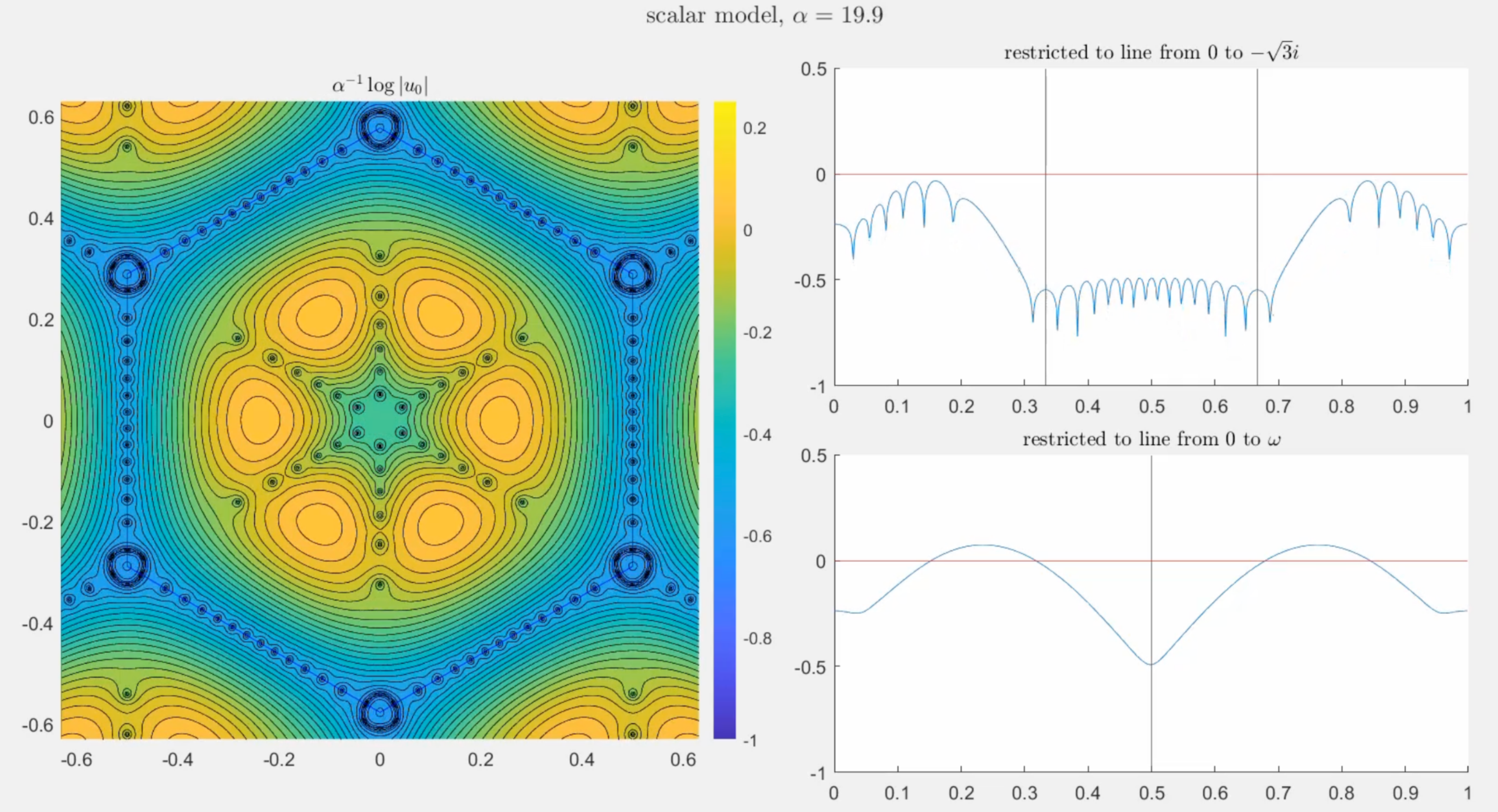}
\caption{\label{f:WKB} On the left: contour plot of $ z \mapsto \alpha^{-1} \log | u ( \alpha, z ) | $, $ \alpha = 20 $, where
$  ( 2 D_{\bar z } )^2 - \alpha^2 U ( z ) U ( -z ) ) u ( \alpha, z ) = 0 $ and $ U $ is the Bistritzer--MacDonald potential \eqref{eq:defU2}. We see the structure of zeros and the regions of exponential decay predicted in \cite{hizw}. On the right:
the plot of $ z \mapsto \alpha^{-1} \log | u ( \alpha, z ) | $ over the the interval connected the center to the edge of the hexagon, continued through the edge and then to the translated center (top) and over the interval connecting the 
centers and perpendicular to the edge of the hexagon. For an animated version with varying $ \alpha $ see
\url{https://math.berkeley.edu/~zworski/scalar_protected.mp4}.}
 \end{figure}

 \section{WKB structure of protected states}
 \label{s:WKB}
 
Numerical computations indicate that the family of protected states, $ u ( \alpha ) $, (see Proposition \ref{p:simple}) 
has a structure of a WKB state -- see Figure \ref{f:WKB} for the case of the Bistritzer--MacDonald potential \eqref{eq:defU2}
It was shown in \cite{hizw} that the ($L^2$-normalised) protected state (in the scalar and chiral models) for that potential is exponentially small, that is of size $ \exp( - c_0 \alpha ) $ in a neighbourhood of the hexagon. Numerically we see
that outside of any neighbourhood of the 6 maxima of $ | u ( \alpha ) | $ we have exponential decay of the solution. 
Mathematically, the key property behind the exponential decay near the hexagon is the vanishing of the 
Poisson bracket $ \{ q , \bar q \}|_{q^{-1} ( 0 ) } $, $ q := ( 2 \bar \zeta )^2 - V ( z ) $ on the hexagon and its non-vanishing outside of the hexagon (except at the center). We refer to \cite{hizw} for the detailed discussion of this.

The structure changes for other potentials and an interesting case is provided by 
\begin{equation}
\label{eq:Hpot}  U ( z ) = i \overline { U_{BM} ( z ) }^2 ,
\end{equation} 
which satisfies the conditions \eqref{eq:defU} but for which $ \{ q , \bar q \}_{ q^{-1} ( 0 ) } \equiv 0 $. That implies an existence of a global solution to the eikonal equation and suggests different WKB structure.

 In this section we will first discuss the zero creation at the corners of the hexagon for the Bistritzer--MacDonald potential. 
 We will then present a heuristic derivation of a WKB state for the model based on \eqref{eq:Hpot} and an argument towards the proof of Theorem \ref{t:2}.

 \subsection{Mechanism for the creation of zeros for the Bistritzer--MacDonald potential}
 \label{s:mechzero}
 For the scalar model numerical calculation show that the real elements of $ \mathcal A $ have multiplicity 2. 
 In view of the arguments presented in \S \ref{s:mult2} this means that the protected state 
 $ \mathbf v ( \alpha ) $, $ D_{\rm{S}} ( \alpha ) \mathbf v ( \alpha ) = 0 $ has simple zeros (in the sense of \eqref{eq:defM}) at the 
 corners of the hexagon. Using \eqref{eq:syst0} we see that 
 \[ \mathbf v ( \alpha , z) = \begin{pmatrix} \alpha^{-1} 2 D_{\bar z } u ( \alpha, z )  \\
 u ( \alpha, z )  \end{pmatrix} , \ \ \   (( 2 D_{\bar z } )^2 - \alpha^2 V ( z )) u ( \alpha, z ) = 0 . \]
 The structure of $ u ( \alpha, z ) $ was illustrated in Figure \ref{f:WKB}.
 
 In this section we will describe an approximate local mechanism for the creation of zeros of $ u $ and 
 $ \partial_{\bar z } u $ at the corner. The symmetries of $ u $ will play a crucial role here.

We start by describing $ V ( z ) = U_{\rm{BM} }  ( z ) U_{\rm{BM} } ( -z ) $, where $ U_{\rm{BM} } $ is given in
\eqref{eq:defU2} near the corner of the hexagon $ z_S $: 
\begin{equation}
\label{eq:V2Ai} V ( z_S + w ) = - i a \bar w + b w^2 + \mathcal O ( |w|^3 ), \ \ \ a = \tfrac{32}3 \pi^3 , \ \ b = \tfrac{32}9 \pi^4 + \mathcal O ( |w|^3 ) . \end{equation} 
The model operator near the corner (writing $ z = w $) is then given by 
\begin{equation}
\label{eq:modelP}  P = ( 2  D_{\bar z } )^2 + \alpha^2( i a \bar z - b z^2)   . \end{equation}
To represent the solutions using Airy functions we write 
\[ z = -  i \beta w , \  \ \  \bar z = i \beta \bar w ,   \ \ D_{\bar{z }} = -  \beta^{-1}   \partial_{\bar w } ,   \]
so that if $ \beta  = ( 2 / \alpha )^{\frac23} a ^{-\frac13} $ then
\[ \begin{split} P & = (2h/ \alpha )^2 \partial_{\bar w }^2  - a \alpha \bar w + b \alpha^2 w^2 \\
& = 
4 ( \alpha/2 )^{4/3} a^{\frac13}  ( \partial_{\bar w}^2 - \bar w + c \alpha^{-\frac23} w^2 ) , \ \ \ 
c :=  2^{\frac23} b/ a^{\frac43} =  \tfrac 23 ( \tfrac {16} 3 )^{\frac13} > 0 . \end{split}  \]
We keep in mind that
\begin{equation}
\label{eq:z2w}   w = i \beta ( \alpha )  z  , \ \ \ \beta ( \alpha ) := (2/a)^{\frac23}  a^{-\frac13} .  \end{equation} 
Let $ u $ denote the solution in the original variables of \eqref{eq:scalar0}. We want to deduce the symmetries of
 (local) solutions to $ P v  = 0 $ pretending that the operators agree. In view of \eqref{eq:symu}
 and the fact that $ \omega z_S \equiv z_S\!\! \mod \! \Lambda $, $  \bar z_S = - z_S \!\! \mod \! \Lambda $,  we have 
\begin{equation}
\label{eq:symv} \begin{gathered}  v (  w ) = u ( z_S + \omega i \beta w ) = 
u ( \omega ( z_S +  i \beta w  ) ) = u ( z_S + \omega i \beta w  ) = v ( \omega w  ) , \\
  v  ( \bar  w ) = u ( z_S +  i \beta \bar  w ) = u  ( \overline{ - z_S - i \beta w  } ) = 
\overline{ u ( -z_S - i \beta w  ) } = \overline{ u  ( z_S + i \beta w ) } = \overline{ v   ( w ) } . \end{gathered} \end{equation} 
(Note that this last consequence is global in the sense that we use global properties of our solution.)

\begin{figure} 
\centering
\includegraphics[width=6cm]{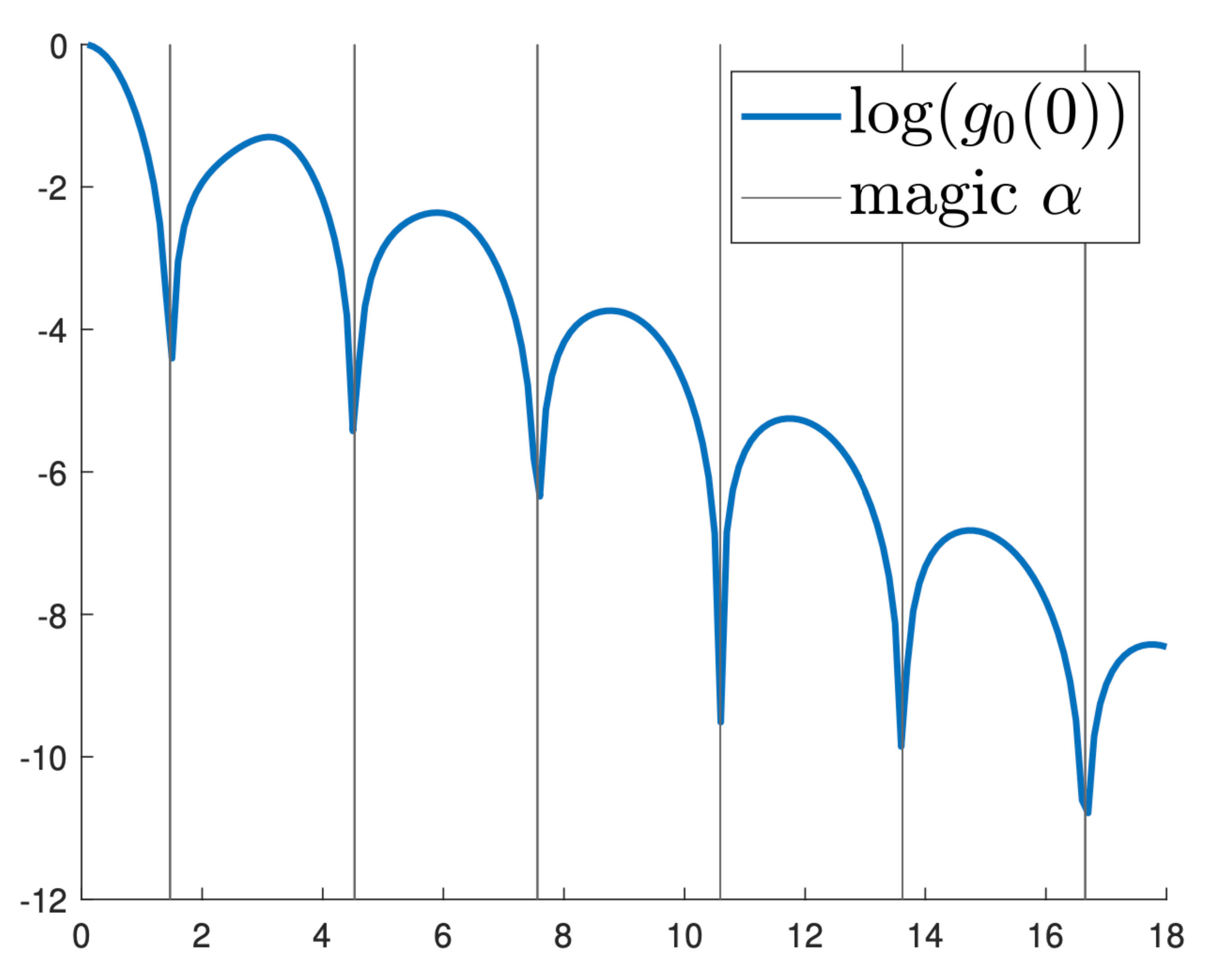}
\caption{\label{f:g0} Illustration of $ g_0( 0 , \alpha ) $ in \eqref{eq:decv} and \eqref{eq:f2g} -- it is calculated by 
numerically evaluating the Wronskian of the protected state with $ Ai_2 ( \bar w - \alpha^{-\frac23} c w^2 ) $.
We see exponential decay explained in \cite[Theorem 1]{hizw} and the zeros at the magic angles of the scalar model 
for the potential \eqref{eq:defU2}.}
  \end{figure}

We recall some basic facts about Airy functions\footnote{See \url{https://mtaylor.web.unc.edu/wp-content/uploads/sites/16915/2018/04/airyf.pdf} for a quick but thorough introduction.}. The first Airy function is given by 
\[ Ai ( s ) =  \frac{1}{ 2\pi } \int_{\mathbb R } e^{ i st +i  t^3/3 } dt , \ \ \  \partial_s^2 Ai ( s) - s A i ( s ) = 0. \]
By either considering the equation satisfied by $ Ai $ or by deforming the contour in the integral we see that
$ \mathbb C \ni \xi \mapsto Ai ( \xi )$ is  holomorphic and $ Ai ( \overline \xi ) = \overline{ Ai ( \xi ) } $. 
That allows to define 
\begin{equation}
\label{eq:defAij}  
Ai_j ( \xi ) := \sum_{ k=0}^{j} \omega^{ kj} Ai ( \omega^k \xi ) . 
\end{equation}
It follows that $ Ai_1 \equiv 0 $ and that $ Ai_0 $ and $ Ai_2 $ give a basis of solutions to 
the Airy equations ($ Ai_0 ( 0 ) = 3 Ai ( 0 ) \neq 0 $, $ Ai_0' ( 0 ) = 0 $, $ Ai_2 ( 0 ) = 0 $, $ Ai_2'( 0 ) = 3 Ai'(0) \neq 0 $). We also have the following properties 
\begin{equation}
\label{eq:propAij} 
\overline{ Ai_j ( \xi ) }  
= Ai_j ( \overline \xi ), \ \ \
Ai_j ( \omega \xi ) 
=  \omega^{-j} Ai_j ( \xi ) . 
\end{equation}

\begin{figure} 
\centering
\includegraphics[width=4cm]{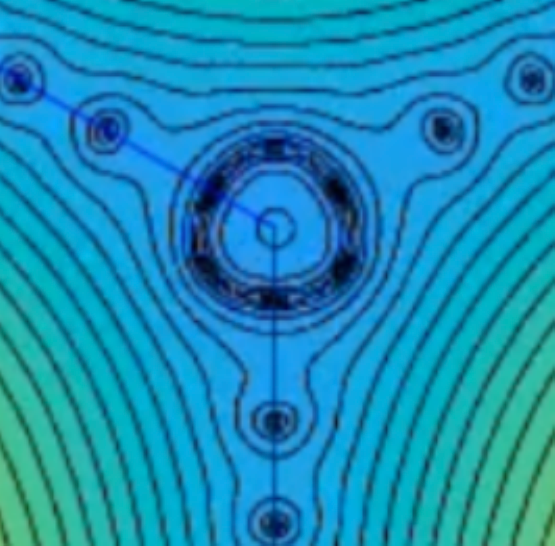} \includegraphics[width=4cm]{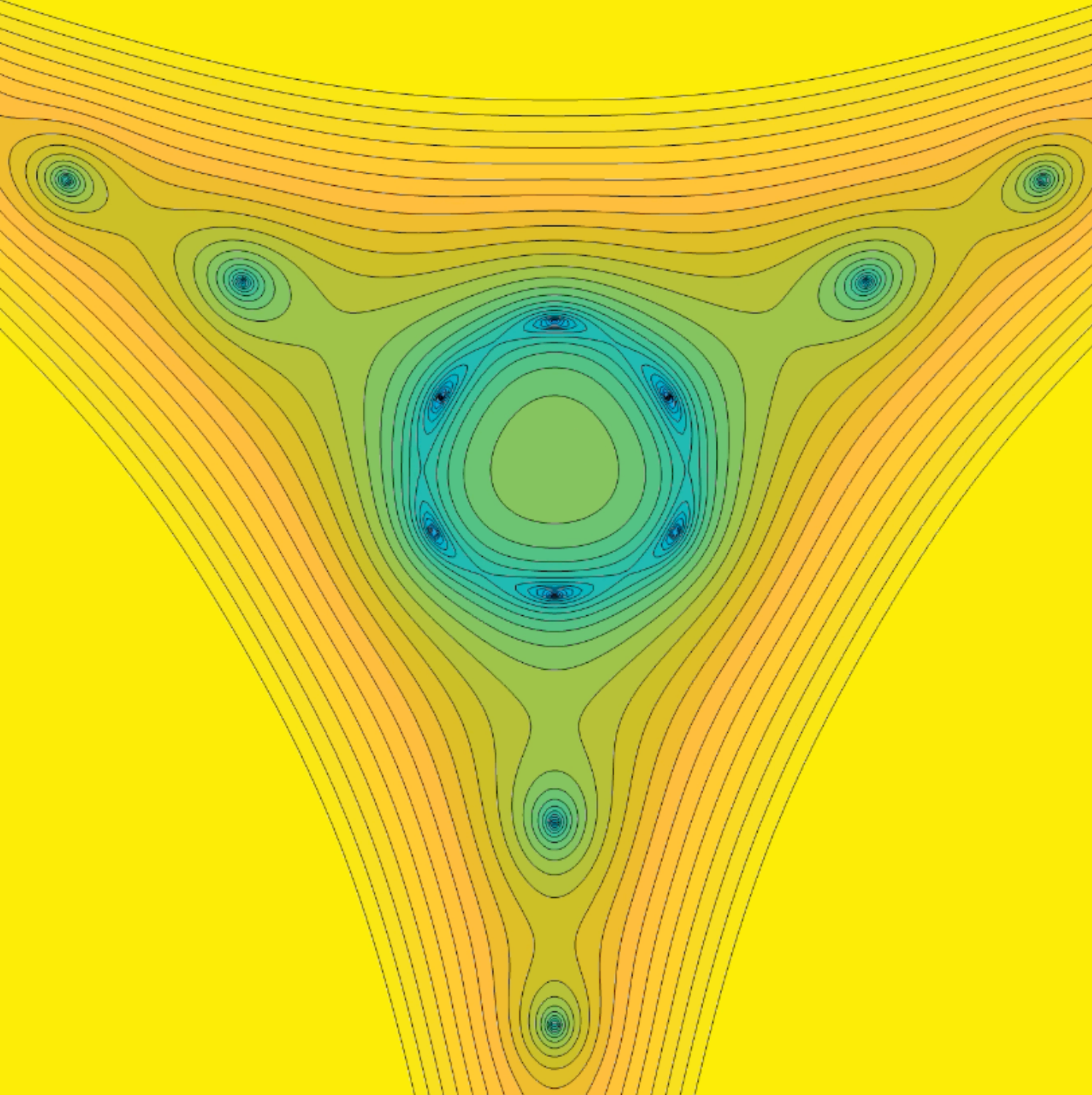}
\caption{\label{f:zerocr} On the left a snapshot from the movie referenced in Figure \ref{f:WKB} showing the evolutions of 
zeros at the corner; on the right a snapshot from \url{https://math.berkeley.edu/~zworski/airy_movie.mp4} showing the 
same structure for \eqref{eq:decv} with $ g_0 $ given by \eqref{eq:apg0} with $ c_j = 1 $, $j < 4 $, $ c_4 = 0 $.}
  \end{figure}

We now consider solutions to  $ P v = 0 $  where $ P$ is  \eqref{eq:modelP} but in coordinates \eqref{eq:z2w}. 
In view of \eqref{eq:V2Ai} these solutions approximate solutions to $ P ( \alpha, 0 ) u = 0 $ near the 
corners of the hexagon. They are given by 
\begin{equation}
\label{eq:decv}  v ( w , \bar w ) = f_0 ( w, \alpha ) Ai_0 ( \bar w - \alpha^{-\frac23} c w^2 ) + f_2 ( w, \alpha  ) Ai_2 ( \bar w - \alpha^{-\frac23} c w^2 ) ,
\end{equation} 
where $ f_j $, $ j = 0, 2 $ are holomorphic in $ w $. 
Combining \eqref{eq:symv} with \eqref{eq:propAij} we obtain,
with  $ g_j $, $j =1,2 $,  holomorphic, 
\begin{equation}
\label{eq:f2g}   f_0 ( w , \alpha ) = g_0 ( w^3, \alpha ) , \ \ \ f_2 ( w, \alpha ) = w g_2 ( w^3 , \alpha ) , 
\ \ \  g_j ( \bar \zeta , \alpha ) = \overline { g_j ( \zeta, \alpha ) } .
\end{equation}
Numerically we can compute $ g_j $'s by taking Wronskians of the actual solution to $ P ( \alpha, 0 )u = 0 $
with $ Ai_j ( \bar w -  \alpha^{-\frac23} c w^2 ) $. The resulting functions $ g_j $ are close to being holomorphic 
(the Airy functions provide approximate solutions for $ P ( \alpha, 0 ) $ -- see \eqref{eq:V2Ai}; they are holomorphic when solutions to $ P v = 0 $ are decomposed using \eqref{eq:decv}) and satisfy 
\[   \begin{gathered} g_0' ( 0, \alpha ) \neq 0  , \ \ \ 
g_2 ( 0, \alpha  ) \neq 0 . 
\end{gathered} \]
and  the zeros of $  g_0 ( 0, \alpha ) $ determine the distribution of $ \alpha$'s. Figure \ref{f:g0} shows the plot of the
numerically computed $ \log | g_0 ( 0 , \alpha )| $ as function of $ \alpha $ (using \eqref{eq:decv} and \eqref{eq:f2g}). 
Exponential decay and the way zeros are created suggests  that 
\begin{equation} 
\label{eq:apg0} g_0 ( \zeta , \alpha ) \simeq e^{ - c_0 \alpha } ( c_1 \zeta + c_2 \sin( c_3 \alpha + c_4 ) ) . 
\end{equation}
That choice of $ g_0 $ in \eqref{eq:decv} and \eqref{eq:f2g} matches the the observed mechanism of zero creation -- see Figure \ref{f:zerocr}.

\begin{figure} 
\centering
\includegraphics[width=7cm]{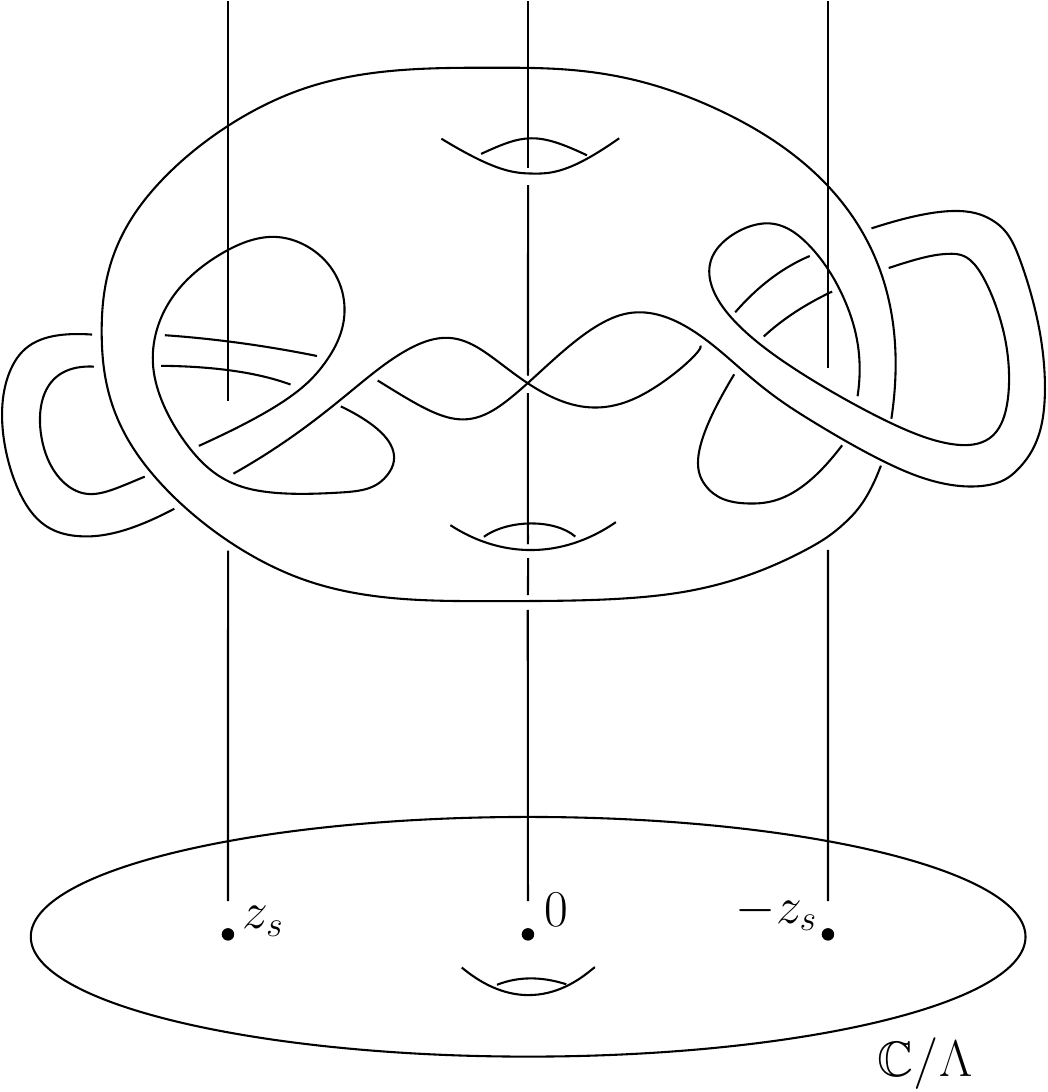} 
\caption{\label{f:BH} An attempt to illustrate the characteristic variety $ q^{-1} ( 0 ) $ for 
the classical symbol of \eqref{eq:scalar0}, $ q ( x, \xi ) = ( 2 \bar \zeta )^2 - V ( z ) $. It is a genus two surface
with singular projection to the base torus above the vertices of the hexagon (the stacking points $ \pm z_S $)
and a single singularity above the center of the hexagon. That singularity is almost a normal crossing $ \bar \zeta^2 - z^2 = 0$  but the two 
surfaces, $ \bar \zeta \simeq \pm z $ crossing transversally at the center (and momentum zero) are not smooth but $ C^{2,\gamma} $, 
$ \gamma < 1 $.  Elsewhere, 
the surfaces is real analytic.}
  \end{figure}

Ideally one would like to connect solutions approximated by \eqref{eq:decv} with other local WKB solutions to obtain an approximation for $ g_0 ( 0 , \alpha ) $ coming from a global solution: knowing $ c_3 $ in \eqref{eq:apg0} would then given a spacing in the values of real $ \alpha $'s as in \eqref{eq:quants}. An attempt in that direction was made in \cite{renu} in a somewhat different spirit than proposed in this section. The difficulty lies in a complicated structure of the characteristic variety $  q^{-1} ( 0 ) $, $ q ( x, \xi ) = ( 2 \bar \zeta )^2 - V ( z ) $, $ \zeta = \frac12 ( \xi_1 - i \xi_2 ) $, $ z = x_1 + i x_2 $ (this is the classical observable which is semiclassical quantised, with $ h = 1/\alpha $ to \eqref{eq:scalar0} -- see \cite{hizw}). 
An attempt to illustrate that structure is made in Figure \ref{f:BH}. The exponential decay of the solution near the hexagon  adds to the many complications.

\begin{figure} 
\centering
\includegraphics[width=5cm]{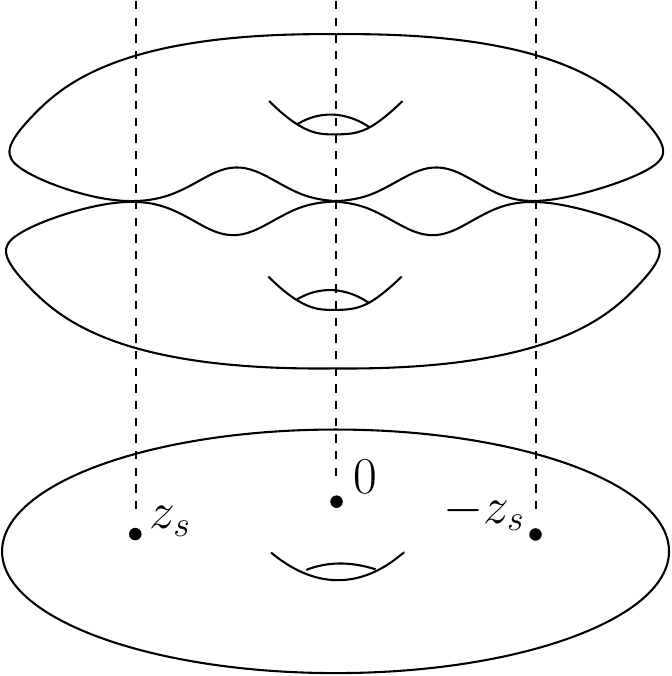}
\caption{\label{f:Hch}  The characteristic variety $ q^{-1} ( 0 ) $ for 
the classical symbol of \eqref{eq:scalar0}, $ q ( x, \xi ) = ( 2 \bar \zeta )^2 - V ( z ) $ with $ V ( z ) = U ( z ) U ( -z ) $
and $ U ( z ) = i \overline{ U_{\rm{BM}} ( z ) }^2 $. It is a union of graphs of $ \pm d \varphi ( z )$, where
$ \varphi $ is real and is given in \eqref{eq:defph}. These two real Lagrangian tori touch at the 
edges of the hexagon and at its center. That complicates the WKB structure of solutions which nevertheless is much simpler than that for the characteristic variety shown in Figure \ref{f:BH}.}
  \end{figure}

 \subsection{An example with a quantisation condition for magic alphas}
 \label{s:Henry} 
 
We now consider $ U $ given by \eqref{eq:Hpot}. The properties \eqref{eq:defU} are inherited from the same 
properties being satisfied by $ U_{\rm{BM}} $. Hence, the set of magic alphas $ \mathcal A$ is well defined 
both for the corresponding chiral and scalar models. Here we will consider the latter.
What changes is the structure of the characteristic variety $ q^{-1} ( 0 ) $, $ q := ( 2 \bar \zeta)^2 - V ( z ) $. 
It consists of two Lagrangian tori with three high tangency common points -- it is shown in Figure \ref{f:Hch}.

The simpler structure of the characteristic variety is due the existence of a global phase function. To see it we 
write the potential as 
\begin{equation}
\label{eq:Vbar2}
\begin{gathered} 
    V(z) = W(z)^2 =  U (z ) U ( - z ) , \ \ \   U ( z ) := i \overline {U_{\rm{BM}} ( z ) }^2 , \\
     W(z) := i \overline{U(z)U(-z)} = - 2i\left(\frac{4\pi}{3}\right)^2
    \sum_{l=0}^2 \omega^l \cos \left( \left\langle z, \omega^{l} {4\pi i}/{\sqrt{3}}\right\rangle\right).
    \end{gathered} 
\end{equation}
The eikonal equation for $ ( 2 D_{\bar z } ) - \alpha^2 W ( z )^2 $ (where we think of $ \alpha = 1/h $ as a semiclassical parameter) reads 
\begin{equation}
\label{eq:eiko}
2 \partial_{\bar z } \varphi ( z) = W ( z ) , 
\end{equation}
and it is solved by a real valued function:
\begin{equation}
\label{eq:defph} 
    \varphi ( z ) := \frac{8\pi}{3\sqrt{3}}\sum_{l = 0}^2 \sin\left(\left\langle z,{4\pi i}/{\sqrt{3}}\omega^l\right\rangle\right). 
   \end{equation}
Hence 
\begin{equation}
\label{eq:Hch}   
\begin{gathered} q^{-1} ( 0 ) = \{ ( x, d_x \varphi ( x)  ) : x \in \mathbb C/\Lambda \} \cup \{ ( x , - d_x \varphi (x ) ) : x \in \mathbb C/\Lambda \} , \\
q ( x, \xi ) = (2 \bar \zeta )^2  - W ( z)^2 , \ \ \ \zeta = \tfrac12 ( \xi_1 - i \xi_2 ) , \ \  z = x_1 + i x_2 . 
\end{gathered}
\end{equation}
In particular, we have 
\begin{equation}
\label{eq:0bra}  \{ q , \bar q \}|_{ q^{-1} ( 0 ) } = 0 ,\end{equation}
and consequently there is no mechanism for exponential decay of solutions to $ P ( \alpha, 0 ) u = 0 $ as in the case of 
$ U =  U_{\rm{BM}} $ -- see \cite{hizw}.

The phase function \eqref{eq:defph} enjoys the following natural symmetries:
\begin{equation}
\label{eq:symph}
\begin{gathered}
 \varphi ( -z ) = \varphi ( \bar z ) = - \varphi ( z ) , \ \ \ \varphi ( \tau z ) = - \varphi ( z ) , \ \ \ \tau = e^{ - i \pi / 3 } .
 \end{gathered} 
 \end{equation}
For future reference we also consider the behaviour of $ \varphi $ at the vertex, $  z = z_S $ and at the center $ z = 0$, 
of the hexagon. For that we first note that \eqref{eq:V2Ai} gives 
\[ W ( z_S +w  ) = i \overline {V_{\rm{BM}} ( z )}
= i ( - i a w + b \bar w^2 + \mathcal O ( |w|^3) ))  = a w + \mathcal O ( |w|^2 ) ,\]
so that 
\begin{equation}
\label{eq:phi_at_zS}
\varphi( z_S + w ) = \tfrac12 a |w|^2  + \mathcal O ( |w|^3 )  , \  \ \ a = \tfrac{32} 3 \pi^3 .
\end{equation}
At the center we have 
\begin{equation}
\label{eq:phi_at_0}
\begin{split} \varphi( z ) & = \frac{8\pi}{3\sqrt{3}}\sum_{\ell = 0}^2 \sin\left(\frac{2\pi i}{\sqrt{3}} ( \bar z \omega^\ell - z \bar \omega^\ell ) \right)
 = -i \frac{ 32 \pi^4}{ 81} \sum_{\ell = 0}^2  \left( \bar z \omega^\ell - z \bar \omega^\ell \right)^3  + \mathcal O ( |z|^5 )\\
 & = -i \frac{ 32 \pi^4 }{27} ( \bar z^3 - z^3 ) + \mathcal O ( |z|^5 ) 
  = - \frac{ 64 \pi^4 }{27} \Im z^3 + \mathcal O ( |z|^5 ) . 
\end{split} 
\end{equation}

\begin{figure} 
\centering
\includegraphics[width=12cm]{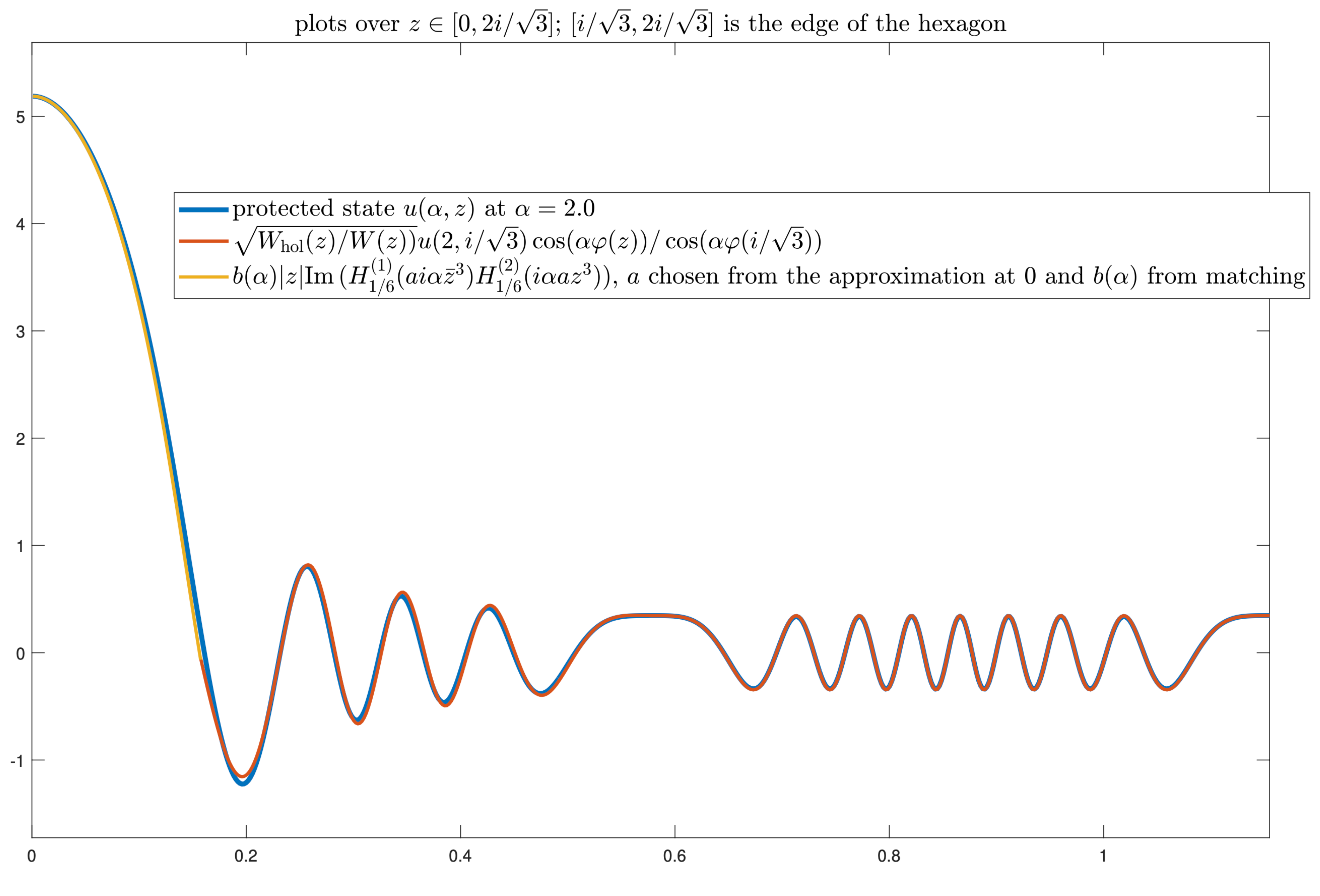}
\caption{\label{f:on_the_edge} The protected state for the potential \eqref{eq:Hpot} over a segment on the
imaginary axis (it is real valued there) compared to WKB approximations near the center of the hexagon
(based on Hankel functions) and away from the center (based on a standard WKB approximation multiplied by a
suitable holomorphic function). Although we do not present a rigorous proof of this approximation the agreement with the
numerics is quite striking. For an animated version with varying $ \alpha $, see \url{https://math.berkeley.edu/~zworski/on_the_edge.mp4}.}
  \end{figure}

We will now propose an approximate solution to $ P ( \alpha, 0 ) u = 0 $ satisfying the symmetries of
the protected state \eqref{eq:symu}. We will construct $ u $ over the closed parallelogram, $ \Sigma $, spanned
by $ \omega  $ and $ - \omega^2 $ (this points are $ \Lambda $-congruent to 0, the centre of the hexagon) and $
(\omega - \omega^2)/3 $, $ 2 ( \omega - \omega^2 )/3 $ (vertices of the
hexagon). We remark that 
\begin{equation}
\label{eq:defSig}   \Sigma_\pm := \Sigma \cap \{ \pm \Re z \geq 0 \}  , \ \ \ \Sigma_0 :=  \Sigma \cap \{ \Re z = 0 \} , 
\end{equation} 
where the $ \pm $ sets are  fundamental domains for the actions of the lattice $ \Lambda $ and 
$ z \mapsto e^{ i \pi /3 } z $ while $ \Sigma_0 $ is the edge of the hexagon.
 The approximate solution is pieced from an approximate solution valid  near $\Sigma_0$ 
 ("near the edge") and an approximate solution valid near $\omega, -\omega^2$ 
 ("near the centre").

\begin{figure} 
\centering
\includegraphics[width=10cm]{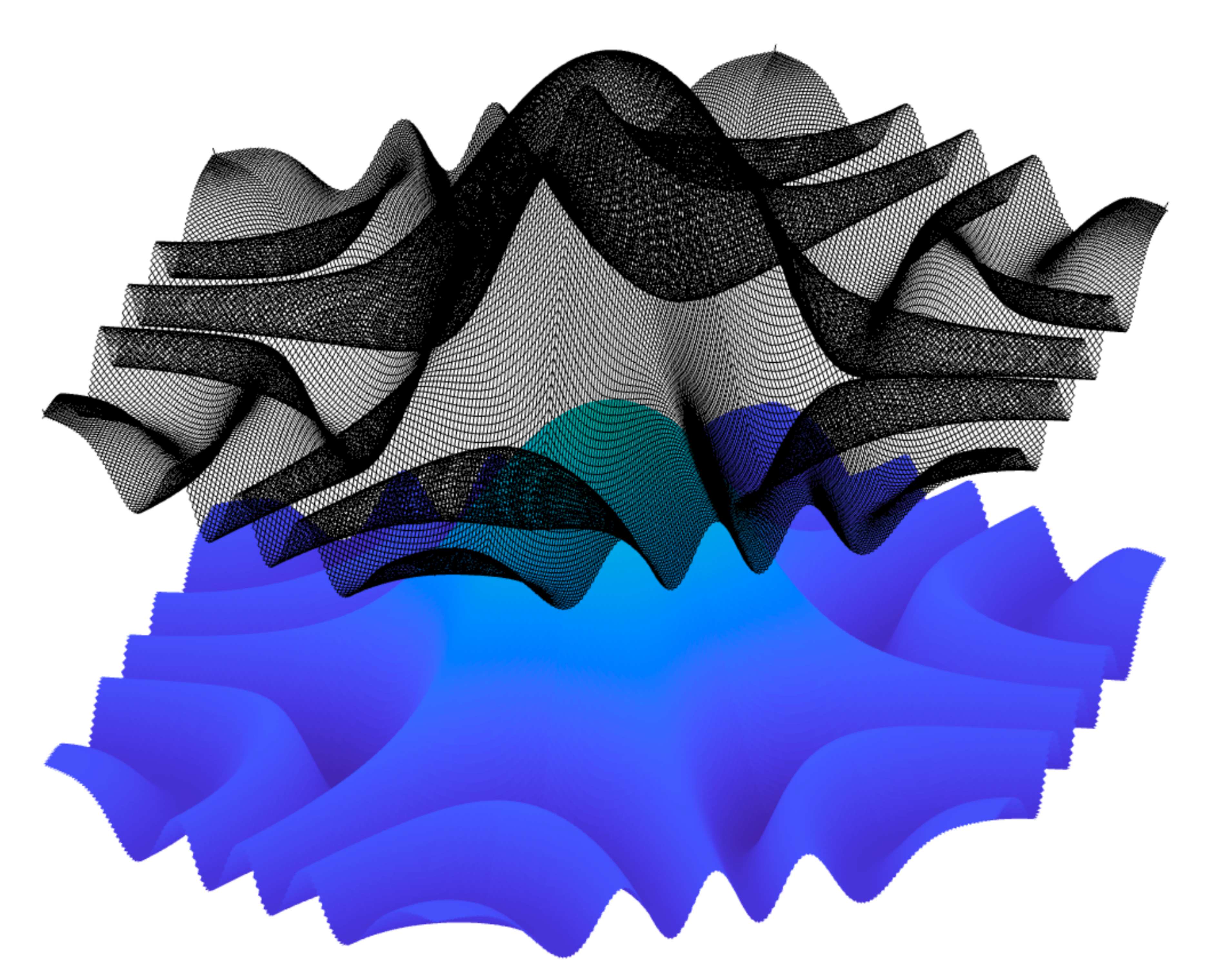}
\caption{\label{f:global} Bottom:  the (dominant) real part of the protected state for $ \alpha = 1 $ over the fundamental domain given by the hexagon. Top: the matched ans\"atze \eqref{eq:cos_ansatz} and \eqref{eq:hankel_ansatz}. The strong agreement in the ``eyeball norm" suggests a rigorous agreement of finer approximations -- see Figure \ref{f:numbers} for some numerical comparisons. A rotating version of this figure can be found at \\
\url{https://math.berkeley.edu/~zworski/rotate_compare.mp4}.}
  \end{figure}

The ultimate goal is to count zeros on the edge as a function of $ \alpha $: a creation of a zero at the vertex at a magic
$ \alpha $ produces two zeros on the edge (see Figure \ref{f:WKB}). Hence counting of those zeros will determine the distribution of real magic $ \alpha$'s.

\begin{figure} 
\centering
\includegraphics[width=12cm]{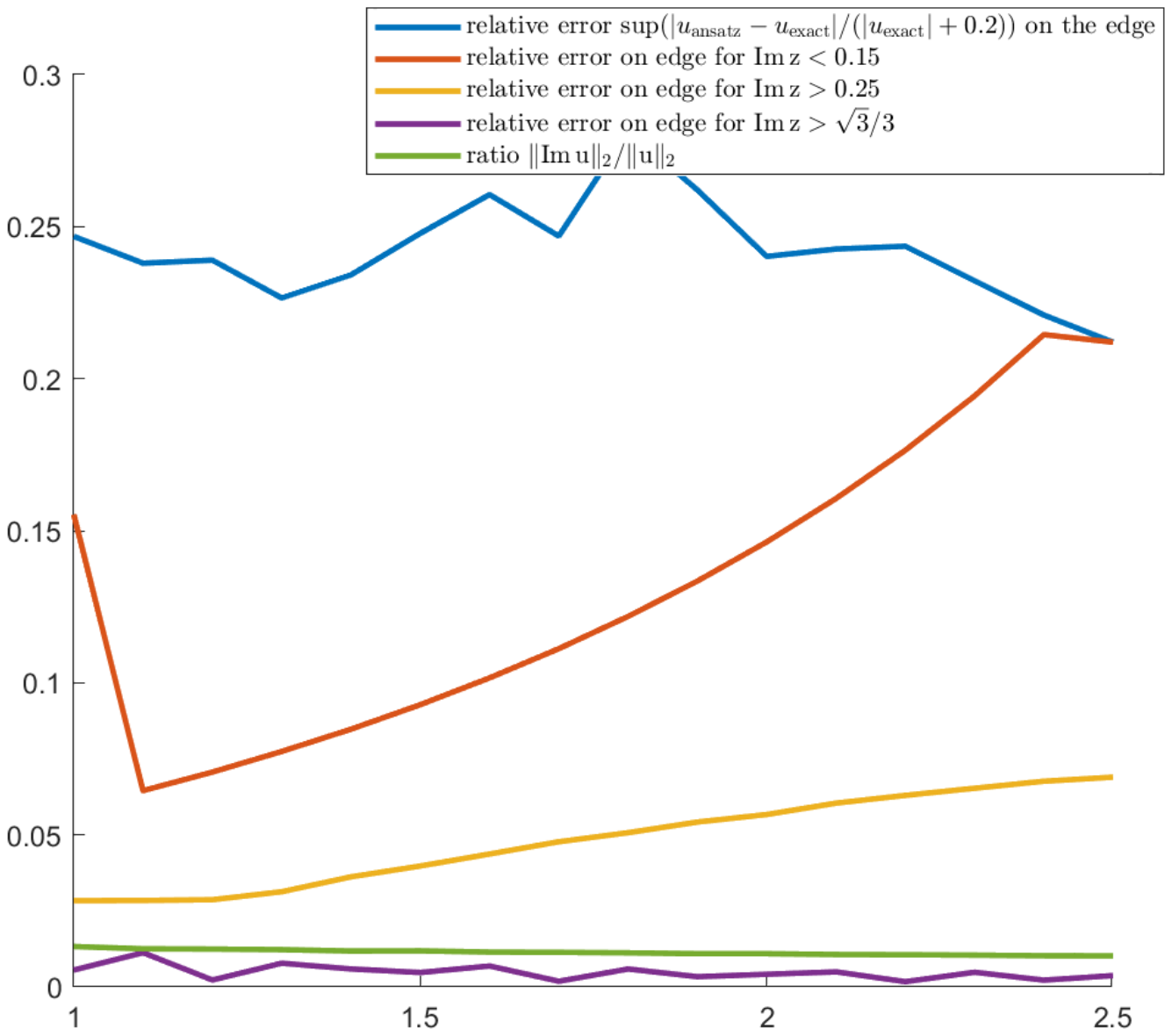}
\caption{\label{f:numbers} Numerical comparison of the rough WKB approximation and the exact solution.
The global relative error is poor due to the quality of the matching; the error for the Hankel function 
approximation ($ 0 < \Im z < 0.15 $) is also not that good as it is based on the Taylor expansion near $ 0 $;
on the other hand the approximation on the edge of the hexagon ($ \Im z > \sqrt 3/3 $) is excellent. Numerically we also see that the imaginary part is small.} 
\end{figure}

\subsubsection{Approximate solution near the edge}
From \eqref{eq:symu} we see that we seek an approximate solution which is real on  $\Sigma_0$ . 
The protected state is {\em not} real everywhere but numerically we see that the imaginary part decays with $ \alpha $.
We expect that this is related to the structure of the characteristic variety \eqref{eq:Hch}, \eqref{eq:0bra} but at the moment we do not have a rigorous argument for that. 

Leading terms in a WKB solutions to $ P ( \alpha, 0 ) u = 0 $ are given by $ W ( z )^{-\frac12} \exp ( \pm i \varphi ( z ) ) $. 
The approximation of the protected state should be real and bounded on $ \Sigma_0 $ (the edge of the hexagon). 
Since we can multiply the approximate solutions by (local) holomorphic functions this suggests the following ansatz in 
$ \Sigma $, near  $ \Sigma_0 \cap \{ \Re z = 0 \} $:
\begin{equation}
\label{eq:cos_ansatz} 
u_{\rm E} ( z ) = ( W_{\rm{hol}} ( z ) / W( z ) )^{\frac12} \cos (\alpha \varphi ( z ) ) , \ \ \
W_{\rm{hol}} \in \mathscr O ( \Sigma ) , \ \ \ W_{\rm{hol}} |_{\Sigma_0 } = W|_{\Sigma_0 } .
\end{equation}
In words, $ W_{\rm{hol}} $ is the holomorphic extension of the real analytic function $ W ( z ) $ from $ \Sigma_0 $ to $ \Sigma $:
\[  W_{\rm{hol}} ( z )  = 
- 2i\left(\frac{4\pi}{3}\right)^2 \left(\cos  \left( 4 \pi i   z / {\sqrt 3}\right) - \cos \left( 2  \pi i   z/  {\sqrt 3} \right)  \right) ,
\] 
    
       The ansatz \eqref{eq:cos_ansatz} is in remarkable agreement with the protected state away from the $ \pm \frac12 $ that is
   away from the center of the hexagon -- see Figure \ref{f:global}. That is particularly striking on the edge continued
   (in a different fundamental domain) to the center -- see Figure \ref{f:on_the_edge}.  We stress that a much finer analysis would be needed to justify this agreement. That has to involve matching the ansatz \eqref{eq:cos_ansatz} with the rougher
 ansatz near the center \eqref{eq:hankel_ansatz}.

\subsubsection{Approximate solution near the centre}

Near the center we will use a simplified version of $ \varphi ( z ) $ obtained by using the leading term in 
\eqref{eq:phi_at_0}:
\begin{equation}
\label{eq:defphi0}   \varphi_0 ( z ) :=   a i  (  z^3 - \bar z^3  )  , \ \ \ W_0 ( z ) := 2 \partial_{\bar z } \varphi_0 ( z ) = 
- 6 i a \bar z^2 , \ \ \
 a := 
\frac{ 32 \pi^4 } {27}  .
\end{equation}
(Eventually, when matching, we will use the lattice to move $ 0 $ to the $ \omega, - \omega^2 $ vertices of $ \Sigma $ and
vice versa.) To understand the structure of solutions to $ ( ( 2 D_{\bar z } )^2 - \alpha^2 W_0 ( z)^2 ) u ( z ) = 0 $,
we note that the leading term of a WKB solution is given by 
\begin{equation}
\label{eq:WKB2H}    u_\pm  ( \bar z ) :=  W_0(z)^{-\frac12} e^{ \pm i \varphi_0 ( z ) } = c_0 \bar z^{-1} e^{ \pm \alpha a \bar z^3 } 
e^{ \mp \alpha a z^3 } 
 . \end{equation}
Since we are free to multiply the approximate local solution by holomorphic functions, this suggests a substitution 
\begin{equation}
\label{eq:u2f}
u(\bar z)=\bar z^{-1}f( i \alpha a \bar z^3 ) 
\ \Longrightarrow \   w^2 f'' ( w ) + (w^2 + \tfrac1 9 ) f ( w ) = 0 , \ \ \ w := i \alpha a \bar z^3 . 
\end{equation}
The equation for $ f $  is essentially the Bessel equation: if $ f(w)=w^{1/2}g(w) $, then 
then
$$
w^2 g''(w)+wg'(w)+(w^2-\tfrac{1}{36}) g(w)=0. 
$$
This is the Bessel equation and in our case (in view of \eqref{eq:WKB2H}) 
is it natural to consider solutions as linear combinations of Hankel functions (see \cite[\S 9.1]{abs} for definitions and properties of these functions):
$ g ( w ) = H^{(j)}_{1/6}  ( w ) $, $ j = 1,2 $, 
\begin{equation}
\label{eq:propH}
\begin{gathered} 
H^{(j)}_{1/6} ( w ) = w^{1/6} F_1 ( w^2 ) - (-1)^j i w^{-1/6} F_2 ( w^2 ) ,  \ \ F_j ( 0 ) \neq 0, \ \  F_j ( w ) = \overline{ F_j ( \bar w ) }, 
\\ H^{(j)}_{1/6}  ( w ) \simeq (2/\pi)^{\frac12} w^{-\frac12} e^{ - (-1)^j i ( w - \frac \pi {3 } ) } , \ \ |w| \to \infty, \ \ 
-\pi < \arg w < 2 \pi . 
\end{gathered} 
\end{equation}
(We define the branch of $ w^{\frac16} $ for $ - \pi < \arg w < 2 \pi $ and the values on the negative axis correspond to 
$ \arg w = \pi $.) 
Returning to \eqref{eq:u2f}, we see that the two standard solutions to $ ( 2 D_{\bar z} )^2 - \alpha^2 W_0 ( z ) ) u = 0 $ are 
$$
u_j(\bar z)=\bar z^{1/2}H_{1/6}^{(j)} \left( i a \alpha \bar z^3 \right),\quad j=1,2, \ \ \ a = 
 32 \pi^4  /{27}  .
$$
We note that in view of \eqref{eq:propH}, $ \bar z \mapsto u_j ( \bar z ) $ are entire (once we choose the correct
branch of $ \bar z \mapsto \bar z^{1/2} $. We also note that in view of \eqref{eq:propH} the complex conjugates, 
$ \overline { u_j ( z ) } $ also solve the equation.

We use $ u_j ( \bar z ) $ as approximate solutions in 
\[   \Omega_\pm^\delta   :=  \pm \{ |z| <  \delta ,  | \arg z | \leq \pi / 6 \} , \ \ \ 
\Omega_+^\delta + \omega = \neigh_{ \Sigma} ( \omega ) , \ \ 
\Omega_-^\delta  - \omega^2 = \neigh_{ \Sigma } ( -\omega^2 )  ,\]
that is near the left and right corners of $ \Sigma $. That is the region where \eqref{eq:cos_ansatz} fails to provide
an approximation.

We need to take a combination of $ u_j ( \bar z ) $, $ \overline {u_j ( z ) } $ with holomorphic coefficients which
is smooth at $ 0 $ when extended from $ \Omega_+^\delta $ to a neighbourhood of $ 0$ using the symmetry 
$ z \mapsto e^{ i \pi / 3 } $ and matches the ansatz \eqref{eq:cos_ansatz}. To satisfy these conditions we arrived 
at the following solution to $ ( (2 D_{\bar z } )^2 - \alpha^2 W_0 ( z ) ) u = 0 $:
\begin{equation}
\label{eq:hankel_ansatz}
u_{\rm{C}} ( z ) = i  |z|  \left(   H^{(1)}_{1/6} (a i \alpha \bar z^3 )  H^{(2)}_{1/6} (i\alpha a z^3) - 
\overline {H^{(2)}_{1/6} (a i \alpha \bar z^3 ) } \,  \overline{H^{(2)}_{1/6} (i\alpha a z^3)} \right). 
\end{equation} 
We consider $ u_{\rm C} $ as an approximate solution to $ ( (2 D_{\bar z } )^2 - \alpha^2 W ( z ) ) u = 0$ near the centre.
A finer ansatz involving a more complicated equation than \eqref{eq:u2f}, depending on $ z $ as a parameter, would be needed
to produce a rigorous agreement. Nevertheless, the agreement shown in Figures \ref{f:on_the_edge} and \ref{f:global}
is remarkable.

\subsubsection{Quantisation rule for magic alphas} We expect real $ \alpha \in \mathcal A $  to have multiplicity two 
and the discussion in \S \ref{s:mult2} then shows that at such $ \alpha$'s, $ u ( \alpha ) $ and 
$ D_{\bar z } u ( \alpha ) $  of Proposition \ref{p:simple}
has simple zeros at the corners of the hexagon (in the sense of \eqref{eq:defM} and the discussion following it). 
We could use the ansatz \eqref{eq:cos_ansatz}  to find $ \alpha $'s for which $ u_{\rm C} (\alpha,  i/\sqrt 3) = 0 $
and that provides a good agreement (in view of the very good agreement with $ u $ on the edge of the hexagon -- 
see Figure \ref{f:on_the_edge}.

A more robust way, potentially applicable in other cases such as
the scalar model based on the Bistritzer--MacDonald potential, is given by counting zeros of $ u ( \alpha ) $
(the family of protected states in Theorem \ref{t:1}) on the edge of the hexagon.

\begin{proof}[``Proof" of Theorem \ref{t:2}] Every time a zero occurs at the corner
{{\em two} zeros are added to the zeros on the edge when $ \alpha $ increases} -- see the movies references in Figures \ref{f:WKB} (for the
Bistritzer--MacDonald potential) and \ref{f:on_the_edge} (for the potential \eqref{eq:Hpot}) -- {this follows from the symmetries of $ u $ in \eqref{eq:symu}.}  Hence of 
that number is given by $ N ( \alpha ) $ then the spacing between magic alpha's, $ \Delta \alpha $ is given by
solving $ N ( \alpha + \Delta \alpha ) - N ( \alpha ) = 2 $. The approximation of $ u ( \alpha ) $ given by 
\eqref{eq:cos_ansatz} is given by $ \cos ( \alpha \varphi ( z ) )$. On the edge we have
\[ \varphi ( z ( t ) ) = \frac{ 8 \pi } { 3 \sqrt 3 } \left ( \sin \left ( \frac{ 4 \pi } 3 t \right) -  2 \sin \left ( \frac{ 2 \pi } 3 t \right) \right), \ \ 
z ( t ) = i t/\sqrt 3, \ \ \  1 \leq t \leq 2 , \]
from which we see that $ t \mapsto \varphi ( z ( t ) ) $ is strictly increasing. Hence
$  N ( \alpha ) \simeq \alpha ( \varphi ( 2 ) - \varphi ( 1 ) ) $ and 
\[  \Delta \alpha \simeq \frac{ 2 } {  \varphi ( 2 ) - \varphi ( 1 ) } = \frac14 . \]
This provides a {\em heuristic} argument for \eqref{eq:deltal}. 
\end{proof}
 
\section{A toy model}

We now turn our attention to the toy model $V(x,y)=W(x)^2$ featured in Theorem~\ref{t:3}.
Here $W$ is a 1-periodic function and we take the lattice of periodicity $\Lambda=\mathbb Z^2$.

\subsection{Existence of protected states}

Before studying the toy model, we consider the question of existence of protected states for
operators of the form~\eqref{eq:scalar0} where $V$ is a general potential, periodic
with respect to some lattice $\Lambda\subset \mathbb C$, which no longer needs to satisfy the symmetries~\eqref{eq:defV}.

We recall that for any $ V \in L^\infty ( \mathbb C/\Lambda ; \mathbb C ) $, $ k \in \mathbb C \setminus \Lambda^* $,
$ \alpha \mapsto P( \alpha, k )^{-1} $ is a meromorphic family of operators. In fact, 
for $ k \notin \Lambda^* $, $ ( 2 D_{\bar z } + k )^{-2} : L^2 ( \mathbb C/\Lambda ) \to 
H^2 ( \mathbb C/\Lambda ) $ exists and hence
\[ P ( \alpha, k ) =  ( 2 D_{\bar z } + k )^2( I - \alpha^2  ( 2 D_{\bar z } + k )^{-2} V ( z )). \] 
Since $ \alpha \mapsto \alpha^2  ( 2 D_{\bar z } + k )^{-2} V ( z ) $ is a holomorphic family of compact operators,
meromorphic Fredholm theory (see \cite[Theorem C.8]{res}) gives the meromorphy of  $ \alpha \mapsto  ( I +  
\alpha^2 (   2 D_{\bar z} + k ) ^{-2} V   )^{-1}$  and hence of $ P ( \alpha, k )^{-1} $ (we have the required
invertibility at one point by putting $ \alpha = 0 $).

Under an additional assumption it is easy to that there are no protected states at $ k \in \Lambda^* $. 
That assumption, \eqref{eq:V0}, can never hold for $ V $ satisfying the symmetry assumptions~\eqref{eq:defV}. 
\begin{lemm}
\label{l:W0}
Suppose that $ V $ satisfies
\begin{equation}
\label{eq:V0}       \int_{ \mathbb C/\Lambda } V ( z ) dm ( z ) \neq 0 . 
\end{equation}
Then, in the notation of \eqref{eq:scalar0},  and for all $ k \in \mathbb C $, $ \alpha \mapsto P( \alpha, k )^{-1} $
is a meromorphic family of operators. In particular $ \ker P( \alpha, 0 ) = \{ 0 \} $ except for a discrete set of $ \alpha $'s. 
\end{lemm} 
\begin{proof} 
For $ k \in \Lambda^* $ we can restrict ourselves to the case of $ k = 0 $ (conjugation by $ e^{ i \langle k , z\rangle } $ reduces to this case). Let us assume for simplicity that $ | \mathbb C/ \Lambda| = 1$. We have the following Grushin problem (see \cite[\S C.1]{res}) for $ P ( 0 , 0 ) $:
\[ \begin{gathered}   \begin{pmatrix}  ( 2 D_{\bar z } )^2   & R_- \\
R_+ & 0 \end{pmatrix} : H^2 ( \mathbb C/\Lambda ; \mathbb C ) \oplus \mathbb C \to 
L^2 ( \mathbb C/\Lambda ; \mathbb C ) \oplus \mathbb C , \\ 
R_- u_- = u_- ,\ \ \  R_+ u  = \langle u , 1 \rangle . \end{gathered} \]
The inverse is given by 
\[   \begin{pmatrix} E & E_+ \\
E_- & 0 \end{pmatrix}, \ \ \  
E_+ = R_- , \ \  E_- = R_+ , \ \ 
Eu = ( 2 D_{\bar z } )^{-2} ( v - \langle v , 1 \rangle ) . 
\] 
We now consider $ P ( \alpha, 0 ) $ as a perturbation of $ P ( 0 , 0 ) $: for small values of $ \alpha $ 
the Grushin problem with the same $ R_\pm $ remains invertible with the inverse (see \cite[C.1.7]{res}):
\begin{equation}
\label{eq:Empz}   \begin{pmatrix} E ( \alpha ) & E_+ ( \alpha )  \\
E_- ( \alpha ) & E_{-+} ( \alpha)  \end{pmatrix}, \ \ 
E_{-+} ( \alpha ) = \sum_{k=1}^\infty \alpha^{2k}  E_- V ( E V )^{k-1}  E_+ . \end{equation} 
We recall from \cite[(C.1.1)]{res} that $ P ( \alpha, 0 ) $ is invertible if and only if $ E_{-+} ( \alpha ) \neq 0 $. 
The leading term in the expansion of $ E_{-+ } ( \alpha )  $ is given by
$ \alpha^2 \int_{\mathbb C /\Lambda } V ( z) dm ( z ) $. Hence \eqref{eq:V0} implies that 
$ P ( \alpha , 0 )^{-1} $ exists for $ \alpha \in D ( 0 , \varepsilon ) \setminus \{ 0 \} $. Since 
$ \alpha \mapsto P ( \alpha, 0 ) $ is a holomorphic family of Fredholm operators we can again apply 
analytic Fredholm theory to obtain global meromorphy of $ P ( \alpha , 0 ) $. 
\end{proof} 

\noindent
{\bf Remark.} Suppose that $ \Gamma \subset \Lambda^* $ such that $ 0 \notin \Gamma $ and 
for $ \gamma, \gamma' \in \Gamma $, $ \gamma + \gamma' \in \Gamma $. Then consider
\begin{equation}
\label{eq:defCG} V ( z ) := \sum_{ \gamma \in \Gamma } a_\gamma e^{ i \langle z , \gamma \rangle} \in C^\infty ( \Lambda/\Gamma ) , \ \ \ 
 a_{\gamma } = \mathcal O ( |\gamma|^{-\infty } ) . \end{equation}
 Then $ \ker P ( \alpha, 0 ) \neq \{ 0 \} $ for all $ \alpha $'s. (Of course the condition \eqref{eq:V0} is obviously violated). This follows from \eqref{eq:Empz}. In fact, if $ C^\infty_{\Gamma} $ denotes the set of $ V's $ of the form
 \eqref{eq:defCG} then 
 \[ E :  C^\infty_\Gamma \to C^\infty_\Gamma , \ \ \ V : C^\infty_\Gamma \to C^\infty_\Gamma , \ \ \ 
 E_- : C^\infty_\Gamma \to 0 . \]
 (Here as above $ V $ as an operator denotes multiplication by $ V $.) 
 
\subsection{Basic properties of the toy model}
\label{s:toy-basic}

We first assume that the potential takes the form $V(z)=W(z)^2$,
where $W$ is periodic with respect to some lattice $\Lambda\subset\mathbb C$.
 In that case 
the characteristic variety $ q^{-1} ( 0 ) $, $ q = ( 2 \bar \zeta )^2 - W (z)^2 $ is given by a union of two disjoint tori
$ \{ ( x, \xi ) : 2 \bar \zeta = \pm W ( z ) \}$, $ \zeta = \frac12 ( \xi_1 - i \xi_2 ) $. These tori are Lagrangian if and only 
of $ \Im \partial_z W (z ) \equiv 0 $ but we are interested in the general case of complex trigonometric polynomials $ W ( z ) $
satisfying the condition 
\begin{equation}
\label{eq:W0}      W_0 := | \mathbb C/\Lambda |^{-1} \int_{ \mathbb C/\Lambda } W ( z ) dm ( z ) \neq 0 . 
\end{equation}

\begin{figure} 
\centering
\includegraphics[width=7cm]{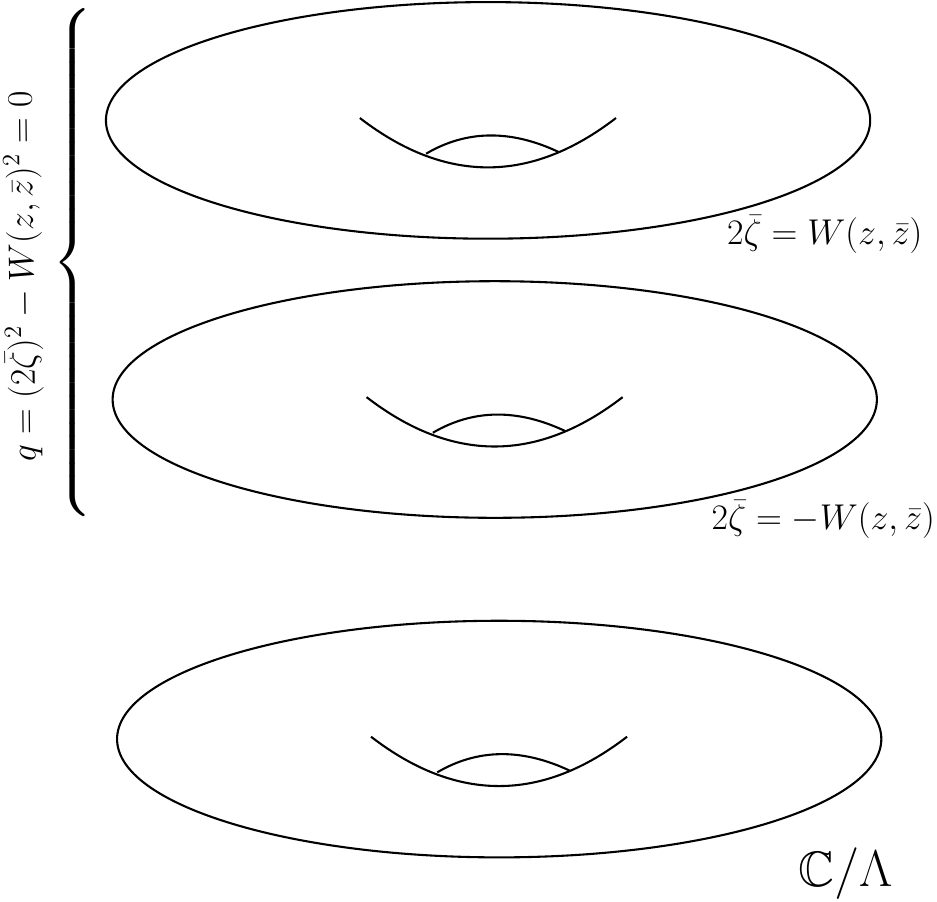}
\caption{\label{f:toy}  The characteristic variety $ q^{-1} ( 0 ) $ for 
the classical symbol of \eqref{eq:scalar0}, $ q ( x, \xi ) = ( 2 \bar \zeta )^2 - V ( z ) $ with $ V ( z ) = W ( z )^2  $
and $ W ( z ) \neq 0 $. It is a union of two {\em disjoint} tori projecting diffeomorphically to the base torus. 
The  tori are Lagrangian only if $ \Im \partial_z W ( z ) \equiv 0 $. Flat bands do not occur for this model but we 
can ask for the values of $ \alpha $ for which $ \ker_{ H^2 ( \mathbb C/\Lambda ) } P ( \alpha, k ) \neq \{ 0 \}$. An 
answer in a special case is given in Theorem \ref{t:3}.}
  \end{figure}

In view of Lemma~\ref{l:W0} it natural to ask about the set 
\begin{equation}
\label{eq:defAk}   \mathcal A ( k ) := \{ \alpha :  \ker P ( \alpha, k ) \neq \{ 0 \} \}. \end{equation}
with the multiplicity of element of $ \mathcal A ( k ) $, $ m ( \alpha , k )$, defined by 
\begin{equation}
\label{eq:defmk} m ( \alpha, k ) := 
\frac{1}{2 \pi i } \tr_{ L^2 ( \mathbb C/\Lambda)   } \oint_{ \alpha} P ( \beta, k )^{-1} \partial_\beta P ( \beta , k ) d \beta,
\end{equation}
where the integral is over a positively oriented circle including $ \alpha $ as the only possible pole
of $ \beta \mapsto P ( \beta, k )^{-1}  $. If $ k \in \Lambda^* $ this assumes meromorphy of $ 
\beta \mapsto P ( \beta, k )^{-1}   $. 

We henceforth specialize to the simplest case $ \Lambda  =  \mathbb Z^2 $ and $ W ( x ,y) = W (x ) $, 
$ ( x, y ) \in \mathbb R^2 $
where $W(x)=W(x+1)$ is real analytic and in fact extends to an entire function. That amounts to the study of the operator $P_\zeta(\alpha)$ defined in~\eqref{eq:Pzet}. We present a self-contained argument which indicates that even in the simplified setting (compared to that in \S \ref{s:WKB}) a rigorous derivation becomes quite 
involved.

As explained in~\eqref{e:realtor} below, we can reduce the study of the kernel of $P_\zeta(\alpha)$
on $1$-periodic functions on~$\mathbb R$ to the semiclassically rescaled operator
$P(\alpha,h):=(hD_x)^2-\alpha^2 W(x)^2$ on functions on~$\mathbb R$ satisfying $u(x+1)=e^{ik/h}u(x)$.

We now assume that the potential $W(x)$ extends to an entire function $W(z)$, which still has
to be 1-periodic. 
The operator $P(\alpha,h)$ naturally extends to an operator on entire functions.
This is the setting in which the analysis of this section is carried out. We emphasize
that the \emph{meaning of the letter $z$ is different here than in the other sections}:
it is the complexification of~$x$. Accordingly, we are looking for elements of the kernel
of $P(\alpha,h)$ which are \emph{holomorphic} in~$z$.

\subsection{Semiclassical formulation}
\label{s:sfo}

We consider the differential operator
$$
P(\alpha,h)=(hD_z)^2-\alpha^2W(z)^2,\quad
z=x+iy\in\mathbb C,\quad
D_z=\frac{1}{2i}(\partial_x-i\partial_y)
$$
where $\alpha\in\mathbb C$, $h>0$, and $W$ is an entire function. We will study
the limit $h\to 0$ and assume that
\begin{equation}
  \label{e:alpha-ass}
C_0^{-1}\leq |\alpha|\leq C_0
\end{equation}
for some constant $C_0$. We should mention that the problem scales and we could consider $ | \alpha | \gg 1 $
instead of  \eqref{e:alpha-ass} and $ h \to 0 $. It is however natural in the context of WKB to use
real $ h \to 0$ as the asymptotic parameter.

We study holomorphic solutions $u(z)$ to the equation
\begin{equation}
  \label{e:schr-1}
P(\alpha,h)u=0.
\end{equation}
We first review some basic properties of second order holomorphic ODEs.
For a holomorphic function $u$, define the vector-valued holomorphic function
\begin{equation}
  \label{e:vec-u-def}
\mathbf u(z)=(u(z),hD_zu(z))\ \in\ \mathbb C^2.
\end{equation}
The Cauchy problem for the equation~\eqref{e:schr-1}
is well-posed: for any $z_0\in\mathbb C$, the map $u\mapsto \mathbf u(z_0)$
is an isomorphism from the space of solutions to~\eqref{e:schr-1}
onto~$\mathbb C^2$ -- see for instance \cite[\S 7.2]{SjBook}. 

For $z_0,z_1\in\mathbb C$, define the transition matrix
$\mathcal M_{z_1\leftarrow z_0}:\mathbb C^2\to\mathbb C^2$ as follows:
for each solution $u$ to~\eqref{e:schr-1} we have
\begin{equation}
  \label{e:transition-def}
\mathbf u(z_1)=\mathcal M_{z_1 \leftarrow z_0}\mathbf u(z_0), \ \ \ \mathcal M_{z_1 \leftarrow z_0 } = 
\mathcal M_{z_1 \leftarrow z_0 } ( \alpha, h ) . 
\end{equation}
For two holomorphic functions $u,v$, define their Wronskian (it differs from the Wronskian used
in \S \ref{s:scal0} by a factor of $ h $):
\begin{equation}
  \label{e:Wr-def}
\Wr(u,v):=u(hD_zv)-(hD_zu)v.
\end{equation}
We will use the identity
\begin{equation}
  \label{e:Wr-identity}
hD_z \Wr(u,v)=uP(\alpha,h)v-vP(\alpha,h)u.
\end{equation}
If $u,v$ are both solutions to~\eqref{e:schr-1}, then~\eqref{e:Wr-identity}
shows that $\Wr(u,v)$ is constant. This implies that
\begin{equation}
  \label{e:transition-det}
\det \mathcal M_{z_1 \leftarrow z_0}=1.
\end{equation}
Finally, let $\Phi(z)$ be an antiderivative of~$W$:
\begin{equation}
  \label{e:Phi-def}
\partial_z\Phi=W.
\end{equation}

\subsection{WKB approximation} 
We will now study the solutions to \eqref{e:schr-1} using the global phase function $\Phi$.

\subsubsection{A growth bound}

From now on in this section, we fix a  smooth non-self-intersecting curve:
\begin{equation}
  \label{e:gamma-cond}
\gamma:[0,1]\to\mathbb C\setminus W^{-1}(0), \quad
\gamma'\neq 0, 
\end{equation}
We henceforth denote by $C_\gamma$ a constant depending only on $\gamma$, $W$, and the constant $C_0$ in~\eqref{e:alpha-ass}, whose precise meaning might change from place to place.

Define the real-valued function
\begin{equation}
  \label{e:psi-def}
\psi_{\gamma,\alpha}(t):=\Im\big(\alpha\Phi(\gamma(t))\big),\quad
t\in [0,1].
\end{equation}
The next lemma is a basic bound on the growth of $|\mathbf u|$ along $\gamma$,
where $u$ solves~\eqref{e:schr-1}: it can grow exponentially at most
like the integral of $|\psi_{\gamma,\alpha}'|$. This is classical and different formulations can be found for instance in 
\cite[\S 7.1]{SjBook} but we give a direct proof in our special case.
\begin{lemm}
  \label{l:energy-bd}
Assume that $u$ solves~\eqref{e:schr-1} and $ \mathbf u $ is given by \eqref{e:vec-u-def}. Then for $ \gamma $ satisfying \eqref{e:gamma-cond} 
we have for $t\in [0,1]$, 
\begin{equation} 
  \label{e:energy-bd}
|\mathbf u(\gamma(t))|\leq C_\gamma\exp\Big(h^{-1}\int_0^t|\psi'_{\gamma,\alpha}(s)|\,ds\Big)|\mathbf u(\gamma(0))|.
\end{equation}
\end{lemm}
\begin{proof}
For $u$ which solves~\eqref{e:schr-1}, define the energy
\begin{equation}
  \label{e:energy-def}
E(z):=\tfrac12\big(|h\partial_zu(z)|^2+|\alpha W(z)u(z)|^2\big),\quad
E_\gamma(t):=E(\gamma(t)).
\end{equation}
We may assume that $u$ is not identically~0 so that $E_\gamma(t)>0$ for all $t$.
We compute
$$
\begin{aligned}
\partial_t \log E_\gamma(t)&=2E_\gamma (t)^{-1}\Re\big(\partial_zE(\gamma(t))\gamma'(t)\big),
\\
\partial_z E(z)&=ih^{-1}\alpha W(z)\Im\big(\overline{\alpha W(z)u(z)}\cdot h\partial_z u(z)\big)+\tfrac12\overline W(z)\partial_z W(z)|\alpha u(z)|^2.
\end{aligned}
$$
For all $z=\gamma(t)$, $0\leq t\leq 1$, we have
$$
\begin{aligned}
\big|\Im\big(\overline{\alpha W(z)u(z)}\cdot h\partial_z u(z)\big)\big|&\leq E_\gamma(z),\\
\big|\tfrac12\overline W(z)\partial_z W(z)|\alpha u(z)|^2\big|&\leq C_\gamma E_\gamma (z).
\end{aligned}
$$
This implies
$$
|\partial_t \log E_\gamma(t)|\leq 2h^{-1}\big|\Im\big(\alpha W(\gamma(t))\gamma'(t)\big)\big|+ C_\gamma.
$$
Integrating and using the identity
$$
\psi'_{\gamma,\alpha}(t)=\Im\big(\alpha W(\gamma(t))\gamma'(t)\big)
$$
we obtain
$$
|\log E_\gamma(t)-\log E_\gamma(0)|\leq C_\gamma+2h^{-1}\int_0^t |\psi'_{\gamma,\alpha}(s)|\,ds\quad\text{for all }t\in [0,1].
$$
Since $C_\gamma^{-1}|\mathbf u(\gamma(t))|^2\leq E_\gamma(t)\leq C_\gamma|\mathbf u(\gamma(t))|^2$,
this shows~\eqref{e:energy-bd}. 
\end{proof}

\subsubsection{WKB solutions}
\label{s:WKB-sol}

We next construct approximate WKB solutions, $u^\pm_N$, to the equation~\eqref{e:schr-1}.
Let $\gamma$ be a curve satisfying~\eqref{e:gamma-cond}. Since $\gamma$ does not self intersect, we can fix a simply connected open set
set $\mathcal U_\gamma$ such that 
$$
\gamma\ \subset\ \mathcal U_\gamma\ \subset\ \mathbb C\setminus W^{-1}(0).
$$
Our solutions have the form (where $N\geq 0$ is an integer)
\[
\begin{gathered}
u^\pm_N(z) =\exp(i \Phi^\pm_N(z;\alpha,h)/h),\quad
z\in \mathcal U_\gamma,\\
\Phi^\pm_N(z;\alpha,h)=\pm \alpha\Phi(z)-ihb^\pm_N(z;h/\alpha), \end{gathered}
\]
where $\Phi$ is defined in~\eqref{e:Phi-def}. This means that
\[ u^\pm_N ( z ) = e^{ \pm i \alpha \Phi( z ) /h } A_\pm ( z , h/\alpha ) , \ \ \  A_\pm ( z , \tilde h ) = 
e^{ b^\pm_N ( z; \tilde h ) }  , \ \ \tilde h := h/\alpha . \]
The function whose exponential gives the amplitude of $ u_N^\pm $ is given by 
\begin{equation}
\label{eq:b2a}
b^\pm_N(z;\tilde h)=\sum_{j=0}^N \tilde h^j a^\pm_j(z),
\end{equation}
where the coefficients $a^\pm_j \in \mathscr O ( \mathcal U_\gamma) $ are determined by the equation
for $ b^\pm_N $: 
\begin{equation}
  \label{e:transport}
2W(z)\partial_z b^\pm_N(z;\tilde h)=-\partial_z W(z)\pm i\tilde h\big(\partial_z^2 b^\pm_N(z;\tilde h)+(\partial_z b^\pm_N(z;\tilde h))^2\big)+\mathcal O(\tilde h^{N+1}).
\end{equation}
Note that~\eqref{e:transport} determines the derivatives
$\partial_z a^\pm_j(z)$ uniquely. Since $ \partial_z a_0^\pm = - \frac12 \partial_z W/W $, we can take
\begin{equation}
  \label{e:a-0-def}
a^\pm_0(z)=-\tfrac12\log W(z)
\end{equation}
where we fix a branch of $\log W$ on~$\mathcal U_\gamma$.
\begin{lemm}
\label{l:WKB} In the notation of \eqref{e:psi-def}, 
we have for all $t\in [0,1]$ and $0<h<C_{\gamma,N}^{-1}$
\begin{gather}
  \label{e:WKB-size}
C_{\gamma,N}^{-1}e^{\mp\psi_{\gamma,\alpha}(t)/h}\leq |u^\pm_N(\gamma(t))|
\leq C_{\gamma,N} e^{\mp\psi_{\gamma,\alpha}(t)/h},
\\
  \label{e:WKB-rhs}
|P(\alpha,h)u^\pm_N(\gamma(t))|\leq C_{\gamma,N} e^{\mp \psi_{\gamma,\alpha}(t)/h}h^{N+2}.
\end{gather}
Here
$C_{\gamma,N}$ denotes a constant depending only on $N$, $\gamma$, $W$, and the constant $C_0$ in~\eqref{e:alpha-ass}.
\end{lemm}
\begin{proof}
The bound~\eqref{e:WKB-size} follows immediately from the definitions of $u^\pm_N$ and
$ \psi_{\gamma, \alpha } $.
To see~\eqref{e:WKB-rhs}, we compute
$$
\begin{aligned}
P(\alpha,h)u^\pm_N(z)=\,&u^\pm_N(z)\big((\partial_z\Phi^\pm_N(z;\alpha,h))^2
-ih\partial_z^2\Phi^\pm_N(z;\alpha,h)
-\alpha^2W(z)^2\big)\\
=\,&u^\pm_N(z)
\big(
\mp i\alpha h(2W(z)\partial_z b^\pm_N(z;h/\alpha)+\partial_zW(z))
\\
&-h^2(\partial_z^2b^\pm_N(z;h/\alpha)+(\partial_z b^\pm_N(z;h/\alpha))^2)
\big).
\end{aligned}
$$
Together with~\eqref{e:transport} this gives~\eqref{e:WKB-rhs}.
\end{proof}
We also have the bound
\begin{equation}
  \label{e:WKB-wron}
\partial_z\big(e^{b^+_N(z;\tilde h)+b^-_N(z;\tilde h)}\big(2W(z)-i\tilde h\partial_z(b^+_N(z;\tilde h)-b^-_N(z;\tilde h))\big)\big)=\mathcal O(\tilde h^{N+1}), 
\end{equation}
where $ \tilde h = h/\alpha $. In fact, we have, 
\[  \begin{split} \Wr ( u_N^+ , u_N^- ) & = e^{ b_N^+  + b_N^-  } 
\left(   
 \partial_z ( -\alpha \Phi   - i h b_N^-  )
 -\partial_z ( \alpha \Phi  - i h b_N^+ ) \right) \\
& = 
\alpha  e^{ b_N^+ + b_N^- }  ( -2 W + i \tilde h \partial_z ( b_N^+ - b_N^- ) ) . \end{split} \]
On the other hand, \eqref{e:Wr-identity}, \eqref{e:WKB-size} and \eqref{e:WKB-rhs} give 
$  \partial_z\Wr ( u_N^+ , u_N^- ) = \mathcal O ( \tilde h^{N+1} ) $.

Next, using the WKB solutions, for a solution $u$ to~\eqref{e:schr-1}
we define the vector of modified Wronskians
\begin{equation}
  \label{e:Wr-mod-def}
{\rm{\bf{Wr}}}_N(z)
=\begin{pmatrix}
u^+_N(z)^{-1}\Wr(u,u^+_N)(z)\\
u^-_N(z)^{-1}\Wr(u,u^-_N)(z)
\end{pmatrix},\quad
z\in\mathcal U_\gamma.
\end{equation}
Then, for all $t\in [0,1]$,
\begin{equation}
  \label{e:Wr-mod-compute}
{\rm{\bf{Wr}}}_N(z)=
B_N(z;\alpha,h)\mathbf u(z)
\end{equation}
where
\begin{equation}
  \label{e:B-N-def}
B_N(z;\alpha,h)=
\begin{pmatrix}
\partial_z\Phi^+_N(z;\alpha,h) & -1 \\
\partial_z\Phi^-_N(z;\alpha,h) & -1
\end{pmatrix} = \begin{pmatrix} \  \ \alpha W ( z ) + \mathcal O ( h ) & -1 \\
- \alpha W ( z )  + \mathcal O ( h ) & -1 \end{pmatrix} .
\end{equation}
Since $ W \neq 0 $,  both $B_N$ and its inverse are bounded in norm uniformly in $h$
for small $h$.

\subsubsection{Approximating the transition matrix}

We now assume that $\gamma$ is a curve as in~\eqref{e:gamma-cond}
and $\gamma,\alpha$ satisfy the additional condition
\begin{equation}
  \label{e:stokes-condition}
\psi_{\gamma,\alpha}'(t)\geq 0\quad\text{for all }t\in [0,1],
\end{equation}
where the function $\psi_{\gamma,\alpha}$ was defined in~\eqref{e:psi-def}.
Define the diagonal matrix
\begin{equation}
  \label{e:S-gamma-def}
S_{\gamma,N}(\alpha,h):=\begin{pmatrix} e^{i(\Phi^+_N(\gamma(0);\alpha,h)-\Phi^+_N(\gamma(1);\alpha,h))/h} & 0 \\
0 & e^{i(\Phi^-_N(\gamma(0);\alpha,h)-\Phi^-_N(\gamma(1);\alpha,h))/h}\end{pmatrix}.
\end{equation}
The transition matrix $\mathcal M_{\gamma(1)\leftarrow\gamma(0)}$ defined in  
\eqref{e:transition-def} is now approximated as follows:
\begin{lemm}
\label{l:transfer-approx}
Under the condition~\eqref{e:stokes-condition}, and with $B_N$
defined in~\eqref{e:B-N-def}, we have
\begin{equation}
  \label{e:transfer-approx}
\begin{aligned}
\mathcal M_{\gamma(1)\leftarrow \gamma(0)}=&\,B_N(\gamma(1);\alpha,h)^{-1}S_{\gamma,N}(\alpha,h)B_N(\gamma(0);\alpha,h)\\
&+
\mathcal O_{\gamma,N}(h^{N+1}e^{(\psi_{\gamma,\alpha}(1)-\psi_{\gamma,\alpha}(0))/h}).
\end{aligned}
\end{equation}
\end{lemm}
\begin{proof}
1. Let $u$ be a solution to~\eqref{e:schr-1}.
By Lemma~\ref{l:energy-bd} and~\eqref{e:stokes-condition}, we have for all $t\in [0,1]$
\begin{equation}
  \label{e:transfer-energy}
|\mathbf u(\gamma(t))|\leq C_{\gamma}e^{(\psi_{\gamma,\alpha}(t)-\psi_{\gamma,\alpha}(0))/h}|\mathbf u(\gamma(0))|.
\end{equation}
Define the Wronskians
$$
\Wr_\pm(t):=\Wr(u,u^\pm_N)(\gamma(t)),\quad
{\rm{\bf{Wr}}}_N(\gamma(t)):=\begin{pmatrix}
e^{-i\Phi^+_N(\gamma(t);\alpha,h)/h}\Wr_+(t)\\
e^{-i\Phi^-_N(\gamma(t);\alpha,h)/h}\Wr_-(t)
\end{pmatrix}.
$$
By~\eqref{e:Wr-identity} and since $P(\alpha,h)u=0$, we have for all $t\in [0,1]$
$$
\begin{aligned}
|\partial_t \Wr_\pm(t)|&\leq C_\gamma h^{-1}|u(\gamma(t))|\cdot |P(\alpha,h)u^\pm_N(\gamma(t)|
\\
&\leq C_{\gamma,N}h^{N+1}e^{(\mp \psi_{\gamma,\alpha}(t)
+\psi_{\gamma,\alpha}(t)-\psi_{\gamma,\alpha}(0))/h}|\mathbf u(\gamma(0))|
\end{aligned}
$$
where the last inequality follows from~\eqref{e:WKB-rhs} and~\eqref{e:transfer-energy}. This implies that
$$
\begin{aligned}
|\partial_t \Wr_+(t)|&\leq  C_{\gamma,N}h^{N+1}e^{-\psi_{\gamma,\alpha}(0)/h}|\mathbf u(\gamma(0))|,\\
|\partial_t \Wr_-(t)|&\leq  C_{\gamma,N}h^{N+1}e^{(2\psi_{\gamma,\alpha}(1)
-\psi_{\gamma,\alpha}(0))/h}|\mathbf u(\gamma(0))|.
\end{aligned}
$$
Integrating in $t\in [0,1]$, we see that
\begin{equation}
  \label{e:transapprox-1}
\begin{aligned}
\Wr_+(1)&=\Wr_+(0)+\mathcal O_{\gamma,N}(h^{N+1}e^{-\psi_{\gamma,\alpha}(0)/h}|\mathbf u(\gamma(0))|),\\
\Wr_-(1)&=\Wr_-(0)+\mathcal O_{\gamma,N}(h^{N+1}e^{(2\psi_{\gamma,\alpha}(1)-\psi_{\gamma,\alpha}(0))/h}|\mathbf u(\gamma(0))|).
\end{aligned}
\end{equation}

\noindent 2. Multiplying both sides of~\eqref{e:transapprox-1}
by~$e^{-i\Phi^\pm_N(\gamma(1);\alpha,h)/h}$, we get
\begin{equation*}
\begin{aligned}
e^{ - i \Phi^+_N ( \gamma ( 1 ) ; \alpha, h )/h } \Wr_+(1)&=
e^{ ( - i \Phi^+_N  ( \gamma ( 1 ) ;  \alpha, h ) + i \Phi^+_N ( \gamma ( 0 ) ; \alpha, h ))/h ) }
e^{ - i \Phi^+_N ( \gamma ( 0 ) ; \alpha , h ) } \Wr_+ ( t ) \\
& \ \ \ \ \ \ +\mathcal O_{\gamma,N}(h^{N+1}e^{
\psi_{\gamma, \alpha } ( 1 ) /h -\psi_{\gamma,\alpha}(0)/h} |\mathbf u(\gamma(0))|),\\
e^{ - i \Phi^-_N ( \gamma ( 1 ) ; \alpha, h )/h } \Wr_-(1)&=
e^{ ( - i \Phi^-_N ( \gamma ( 1 ) ;  \alpha, h ) + i \Phi^-_N ( \gamma ( 0 ) ; \alpha, h ))/h ) }
e^{ - i \Phi^-_N ( \gamma ( 0 ) ; \alpha , h ) } \Wr_+ ( t ) \\
& \ \ \ \ \ \ +\mathcal O_{\gamma,N}(h^{N+1}e^{
\psi_{\gamma, \alpha } ( 1 ) /h -\psi_{\gamma,\alpha}(0)/h} |\mathbf u(\gamma(0))|),
\end{aligned}
\end{equation*}
(we used here the facts that $ \Re ( - i \Phi^\pm_N ( \gamma ( 1 ) ; \alpha, h )) = \pm \psi_{\gamma, \alpha} ( 1 ) +  \mathcal O ( h ) $).
That is, 
$$
{\rm{\bf{Wr}}}_N(\gamma(1))=S_{\gamma,N}(\alpha,h){\rm{\bf{Wr}}}_N(\gamma(0))
+\mathcal O_{\gamma,N}(h^{N+1}e^{(\psi_{\gamma,\alpha}(1)-\psi_{\gamma,\alpha}(0))/h}|\mathbf u(\gamma(0))|).
$$
By~\eqref{e:Wr-mod-compute}
this implies
$$
\begin{aligned}
B_N(\gamma(1);\alpha,h)\mathbf u(\gamma(1))=\,&S_{\gamma,N}(\alpha,h)B_N(\gamma(0);\alpha,h)
\mathbf u(\gamma(0))\\
&+\mathcal O_{\gamma,N}(h^{N+1}e^{(\psi_{\gamma,\alpha}(1)-\psi_{\gamma,\alpha}(0))/h}|\mathbf u(\gamma(0))|).
\end{aligned}
$$
Therefore
$$
\begin{aligned}
\mathbf u(\gamma(1))=\,&B_N(\gamma(1);\alpha,h)^{-1}S_{\gamma,N}(\alpha,h)B_N(\gamma(0);\alpha,h)
\mathbf u(\gamma(0))\\
&+\mathcal O_{\gamma,N}(h^{N+1}e^{(\psi_{\gamma,\alpha}(1)-\psi_{\gamma,\alpha}(0))/h}|\mathbf u(\gamma(0))|).
\end{aligned}
$$
Since this holds for each solution~$u$ to~\eqref{e:schr-1},
we get~\eqref{e:transfer-approx}.
\end{proof}

\subsection{WKB analysis in the periodic case}
\label{s:WKBp} 

From now on, we assume that the potential $W$ is 1-periodic:
$$ W (  z + 1 ) = W ( z ) .$$ 
We also assume that $\gamma$ is a non-self-intersecting curve as in~\eqref{e:gamma-cond} such that
\begin{equation}
\label{eq:gamma_per}
\gamma(1)=\gamma(0)+1,
\end{equation}
and $W$ has winding number~0 along~$\gamma$:
\begin{equation}
  \label{e:W-winding}
\int_\gamma{W'(z)\over W(z)}\,dz=0.
\end{equation}

\subsubsection{Periodicity of WKB solutions}

Let $u^\pm_N(z)=\exp(i\Phi^\pm_N(z;\alpha,h)/h)$, $z\in\mathcal U_\gamma$ (a simply connected neighbourhood of $ \gamma $ which satisfies \eqref{e:gamma-cond} and \eqref{eq:gamma_per})  be the approximate WKB solutions to~\eqref{e:schr-1} constructed in~\S\ref{s:WKB-sol}. Shrinking $\mathcal U_\gamma$
if necessary, we assume that the set $\mathcal U_\gamma\cap (\mathcal U_\gamma-1)$
is connected.

\setcounter{prop}{13}

The solutions $u^\pm_N$ have the following periodicity property:
\begin{lemm}
\label{l:WKB-periodic}
There exist constants $c_0,\dots,c_{N+1}$ such that for all $z\in \mathcal U_\gamma\cap (\mathcal U_\gamma-1)$
\begin{equation}
  \label{e:WKB-periodic}
\begin{gathered} 
\Phi^\pm_N(z+1;\alpha,h)-\Phi^\pm_N(z;\alpha,h)
=\pm \alpha c(h/\alpha),\\
c(\tilde h)=\sum_{j=0}^{N+1} c_j\tilde h^j,\quad
c_0=\int_0^1 W(z)\,dz, \quad c_1 = 0 .
\end{gathered}
\end{equation}
\end{lemm}
\begin{proof}
Since $\partial_z\Phi=W$ and $W$ is 1-periodic, we have $\partial_z(\Phi(z+1)-\Phi(z))=0$. Therefore, $\Phi(z+1)-\Phi(z)$ is constant:
$$
\Phi(z+1)-\Phi(z)= \int_0^1 W ( z ) dz =: c_0 . 
$$
Next, the functions $\partial_z a^\pm_j(z+1;\tilde h)$ and $\partial_z a^\pm_j(z;\tilde h)$ (see 
\eqref{eq:b2a})
satisfy the same equation derived from~\eqref{e:transport}. Since that equation
has a unique solution, we have
\begin{equation}
  \label{e:WKB-per-0}
\partial_z a^\pm_j(z+1)=\partial_z a^\pm_j(z),\quad
z\in \mathcal U_\gamma\cap (\mathcal U_\gamma-1).
\end{equation}
It follows that
\begin{equation}
  \label{e:WKB-per-1}
a^\pm_j(z+1)-a^\pm_j(z)=c^\pm_j,\quad
z\in \mathcal U_\gamma\cap (\mathcal U_\gamma-1)
\end{equation}
for some constants $c^\pm_j$.

Next, by~\eqref{e:WKB-wron} we have for all $z\in \mathcal U_\gamma\cap (\mathcal U_\gamma-1)$
$$
\begin{gathered}
e^{b^+_N(z+1;\tilde h)+b^-_N(z+1;\tilde h)}\big(2W(z+1)-i\tilde h\partial_z(b^+_N(z+1;\tilde h)-b^-_N(z+1;\tilde h))\big)\\
=
e^{b^+_N(z;\tilde h)+b^-_N(z;\tilde h)}\big(2W(z)-i\tilde h\partial_z(b^+_N(z;\tilde h)-b^-_N(z;\tilde h))\big)+
\mathcal O(\tilde h^{N+1}).
\end{gathered}
$$
By~\eqref{e:WKB-per-0} the terms in the parentheses on both sides are the same,
thus
$$
e^{b^+_N(z+1;\tilde h)+b^-_N(z+1;\tilde h)}
=
e^{b^+_N(z;\tilde h)+b^-_N(z;\tilde h)}
+\mathcal O(\tilde h^{N+1}).
$$
In other words,
\begin{equation}
\label{eq:hexp}
\exp\Big(\sum_{j=0}^N \tilde h^j (c_j^++c_j^-)\Big)=1 + \mathcal O(\tilde h^{N+1}).
\end{equation}
By~\eqref{e:a-0-def} 
$ a^\pm_0 ( z)  = \frac12 \log W ( z)  $ and hence $ c_0^\pm = \frac12 \int_\gamma W'/W dz $. 
The winding condition~\eqref{e:W-winding} then shows that
$c^+_0+c^-_0=0$. Expanding the exponential in \eqref{eq:hexp} then gives
$$
c^+_j+c^-_j=0.
$$
Recalling that $ \Phi_N^\pm = \pm \alpha \Phi - i \sum_{j=0}^N \tilde h^j a_j^\pm $, 
gives~\eqref{e:WKB-periodic} with $c_j:=-ic^+_{j-1}$, $ j \geq 1 $. In view of \eqref{e:a-0-def},
$ c_1 =  \frac i 2 \int_\gamma W' ( z ) / W ( z ) dz = 0 $, 
by assumption \eqref{e:W-winding}. \end{proof}

\subsubsection{Stokes loops}
\label{s:stokes-loop}

We now introduce the notion of a Stokes loop.
If the potential $W$ admits a Stokes loop then we are able
to prove a quantization condition for the corresponding operator~$P$.
\begin{defi}
\label{d:stokes-loop}
Assume that $W$ is a 1-periodic entire function on~$\mathbb C$.
Let $\gamma: [0,1]  \to\mathbb C$ be a smooth non-self-intersecting curve. 
We say that $\gamma$ is a \textbf{Stokes loop} for $W$, if { $ \gamma $ extends to a smooth non-self-intersecting curve 
$ \widetilde \gamma : \mathbb R \to \mathbb C $} and  the following conditions hold
for all $t\in\mathbb R$:
\begin{gather}
  \label{e:stokes-loop-1}
\widetilde \gamma(t+1)= \widetilde \gamma(t)+1,\\
  \label{e:stokes-loop-2}
W(\gamma(t))\gamma'(t) > 0  .
\end{gather}
\end{defi}

Another way to formulate the condition on the curve $ \gamma $ is to consider it as a closed non-self-intersecting curve  $ \gamma : \mathbb R / \mathbb Z \mapsto \mathbb C/\mathbb Z $, 
where the action is $ z \mapsto z + 1 $. Indeed, $ \widetilde \gamma $ defines such a curve. That curve is non-self-intersecting: $ \gamma ( [s ] ) = \gamma ([ t ] ) $ ($ t \mapsto [t ] $ is the map $ \mathbb R \to \mathbb R/\mathbb Z $) means that for some $ m $, $ \widetilde \gamma ( s ) = \widetilde \gamma ( t ) + 
m = \widetilde \gamma ( t + m ) $. As $ \widetilde \gamma $ is non-self-intersecting, this implies that $ s = t + m $, that is $ [s] = [t] $.

It is clear that a Stokes loop satisfies~\eqref{e:gamma-cond}. We also have 
\eqref{e:W-winding}: 
\begin{lemm}
\label{l:wind}
Suppose that $ \gamma $ is a Stokes loop in the sense of Definition \ref{d:stokes-loop}. Then 
the winding number condition~\eqref{e:W-winding} holds.
\end{lemm}
\begin{proof}
It is convenient to change the parametrisation of $ \gamma $ to $ \rho : [0, L] \to \mathbb C $, 
$ |\rho' (s ) | = 1$, so that now $ \gamma ( L ) = \gamma ( 0 ) +1 $. Similarly it is useful to define
$ w ( s ) := W ( \rho ( s) ) / | W ( \rho ( s ) ) | $ and to note that the winding numbers around 0  of $ 
s \mapsto W ( \rho ( s) ) $ and of $ s \mapsto w ( s) $ are the same. The condition \eqref{e:stokes-loop-2} becomes
$ w ( s ) \rho' ( s) = 1 $, that is 
\begin{equation}
\label{eq:w2gp} 
w ( s ) = \overline {\rho' ( s ) } . \end{equation}
Condition \eqref{e:stokes-loop-1} shows that 
$ s \mapsto  \rho' ( s ) $ is a smooth closed curve. 
Also, \eqref{eq:w2gp} shows that the winding number around zero of the curve
$ w $ is the negative of the same number for the curve $ \rho' $. 

Hence it suffices to show that the winding number of $ [0, L ] \ni s \mapsto \rho' ( s ) $ is $ 0 $. As remarked after Definition \ref{d:stokes-loop} we can 
consider $ \rho $ as a closed curve on $ \mathbb C/\mathbb Z $. Let $ [ 0 , L ]  \ni s \mapsto \theta ( s ) $ be continuously chosen angle between $ \rho' ( s ) $ and 
$ 1  $ (thought of as vectors in $ \mathbb R^2 \simeq \mathbb C$). The curvature of $ \rho $ is given by $ \kappa_\rho ( s ) = \theta' ( s ) $.  (We recall that to see this
we write $ \rho' ( s ) = ( \cos \theta ( s ) , \sin \theta ( s ) ) $ so that $ \kappa_\rho  ( s ) = 
\rho '' ( s ) \cdot \rho' ( s )^\perp = \theta' ( s ) $). Let $ \rho_1 : [0,1 ] \to \mathbb C/\mathbb Z  $ be given by $ \rho_1 ( s ) = - s + i M \mod 1 $ 
with $ M$ chosen large enough so that $ \rho_1 $ and $ \rho $ are disjoint. Then the curves $ \rho_1 $ and $ \rho $ form a boundary 
of a bounded subset of $ \mathbb C/\mathbb Z $ homeomorphic to $ ( \mathbb R/\mathbb Z)  \times [0,1 ] $ (this follows from the Jordan--Sch\"onflies 
theorem -- see \cite[Theorem 6.7.1]{Riem} --  applied to the cylinder compactified at one of the infinities)  and the Gauss--Bonnet theorem for a flat surface (the cylinder $ \mathbb C /\mathbb Z $;
see for instance \cite[\S 4-5]{doca} for the general case) shows that 
\[  \int_\rho \kappa_\rho ( s ) ds + \int_{\rho_1 } \kappa_{\rho_1}  (s ) ds = 0 . \]
Since $ \kappa_{\rho_1 } \equiv 0 $, it follows that $ \int_{\rho} \kappa_{\rho} (s) ds = 0 $. 
 The calculation of $ \kappa_\rho ( s) $ above also shows that
$ \kappa_\rho ( s) = - i \overline{ \rho' ( s) } \rho'' ( s ) = - i \rho''  ( s ) / \rho' ( s ) $. 
The winding number of $ [ 0 , L ] \ni s \mapsto \rho' ( s ) $ is 
\[ \frac{1}{ 2 \pi i } \int_{0}^L \rho'' ( s ) / \rho' ( s ) ds = \frac{1 } { 2 \pi } \int_\rho \kappa_\rho ( s ) ds = 0 , \]
which proves the claim.
\end{proof}

Moreover, any Stokes loop $\gamma$ satisfies the condition~\eqref{e:stokes-condition}
for $\alpha$ in the upper half-plane:
\begin{equation}
  \label{e:stokes-verified}
\psi_{\gamma,\alpha}'(t)=\Im\big(\alpha W(\gamma(t))\gamma'(t)\big)\geq 0\quad\text{for all }t\in [0,1],\ \Im\alpha\geq 0.
\end{equation}
We also get that the coefficient $c_0$ defined in~\eqref{e:WKB-periodic} is positive:
\begin{equation}
\label{eq:defc0} 
c_0:=\int_\gamma W(z)\,dz=\int_0^1 W(\gamma(t))\gamma'(t)\,dt>0.
\end{equation}
If $W=c_0>0$ is a constant potential, then it has a Stokes loop $\gamma(t)=t$.
The next lemma shows that perturbations of this potential have Stokes loops as well:
\begin{prop}
  \label{p:stokes-perturb}
Assume that $W$ is a 1-periodic entire function and
\begin{gather}
  \label{e:stokes-perturb-1}
\int_0^1 W(z)\,dz=c_0>0,\\
  \label{e:stokes-perturb-2}
W(z)\neq 0\quad\text{and}\quad |\Im W(z)|\leq \Re W(z)\quad\text{when }|\Im z|\leq 1.
\end{gather}
Then $W$ has a Stokes loop.
\end{prop}
\begin{proof}
We are looking for $ \gamma ( t ) $ of the form 
\begin{equation}
\label{eq:ga2f} 
\gamma(t):=t+if(t), \ \  f \in C^\infty ( \mathbb R ; \mathbb R) ,\quad
f(t+1)=f(t).\end{equation}
The condition \eqref{e:stokes-loop-2} then means that 
\[  \begin{split}
&  \Re W ( t + i f ( t) ) - f' ( t ) \Im W ( t + i f ( t ) )  > 0 , \\ &   f' (t ) \Re W ( t + i f ( t ) ) + \Im W ( t + i f ( t ) ) = 0. 
\end{split} \]
This means that $ f $ has to solve
$$
f'(t)=-{\Im W(t+if(t))\over \Re W(t+if(t))},
$$
and we choose $ f ( 0 ) = 0 $ as the initial condition. 
By~\eqref{e:stokes-perturb-2} this Cauchy problem has a solution for $ 0 \leq t \leq 1 $ and 
and moreover $|f ( t ) |\leq 1$ there. In the notation of \eqref{eq:ga2f} we have
$$
W(\gamma(t))\gamma'(t)={|W(t+if(t))|^2\over \Re W(t+if(t))}>0,  \ \ \ t \in [ 0, 1 ] . 
$$
In particular,
$$
\int_{0}^{1+if(1)}W(z)\,dz=\int_0^1 W(\gamma(t))\gamma'(t)\,dt>0.
$$
Since $\int_0^1 W(z)\,dz>0$, we see that
$$
0=\Im\int_1^{1+if(1)}W(z)\,dz=\int_0^{f(1)} \Re W(1+is)\,ds.
$$
Since $\Re W(1+is)>0$, we see that $f(1)=0$. But then periodicity of $ W $ shows that
$ f_1 ( t ) := f ( t + 1 ) $ solves the same equation as $ f ( t )$ with the same initial data. 
Hence $ f ( t ) = f ( t + 1 ) $ and   $\gamma(t+1)=\gamma(t)+1$.
Thus $\gamma$ is a Stokes loop.
\end{proof}

\subsection{Mutliplicity in terms of the transition matrix}
\label{s:multipler}

For $k\in\mathbb C$, let $\mathcal V_{k,h}$ be the space of entire functions
satisfying the periodicity condition
\begin{equation}
  \label{e:V-k-h}
u(z+1)=e^{ik/h} u(z),\quad
z\in\mathbb C.
\end{equation}
We study the values of $(\alpha,k)$ for which the operator $P(\alpha,h)$ 
has a kernel on $\mathcal V_{k,h}$. More precisely, for $ \alpha $ satisfying 
\eqref{e:alpha-ass} and a fixed $ k\in\mathbb C$, 
we define the multiplicity in the way similar to
\eqref{eq:mult} (but with the roles of $ \alpha $ and $ k $ reversed).  For that we consider 
\begin{equation}
\label{eq:defPk} 
\begin{gathered} P_k ( \alpha, h ) =P(\alpha,h): H^2_k  \to L^2_k , \ \ \ 
H^s_k  := \{ w \in H^s_{\rm{loc}} ( \mathbb R ) : 
w(x+1)=e^{ik/h} w(x) \},  \end{gathered}  \end{equation}
($ H_k^0 =  L^2_k \simeq L^2 ( \mathbb R / \mathbb Z) $) and then put
\begin{equation}
  \label{e:m-h-def}
m_h(\alpha,k)= \frac{1}{2 \pi i } \tr_{ L^2_k } \oint_{ \alpha} P_k ( \beta, h )^{-1} \partial_\beta P_k ( \beta , h ) d \beta,
\end{equation}
where the integral is over a positively oriented circle including $ \alpha $ as the only possible pole
of $ \beta \mapsto P_k ( \beta, h ) $.

We note that, in the notation of~\eqref{eq:Pzet},
\begin{equation}
\label{e:realtor} \begin{gathered} P_k( \alpha, h ) = h^2 \tau ( k ) P_{k/h} ( \alpha/h ) \tau(k)^{-1} , \\
\tau (k ) w ( x ) := e^{ i k x /h  } w ( x ), \ \ \  \tau ( k ) :  H^s ( \mathbb R/\mathbb Z ) \to H^s_k . 
\end{gathered}
\end{equation}
Hence, the meromorphy of $ \beta \mapsto P_k ( \beta, h ) $ is guaranteed by Lemma \ref{l:W0}, except 
for  $ k/h \in 2 \pi \mathbb Z $ when we can have $ \ker P_k ( \alpha,  h ) \neq \{ 0 \} $ for all $ \alpha $ (see the remark 
after Lemma \ref{l:W0}). In that case, we put $ m_h ( \alpha, k ) = \infty $. 

Since we consider $V(x,y)=W(x)^2$, the Fourier mode decomposition in~$y$ gives,
in the notation of~\eqref{eq:scalar0},
$$
P(\beta,k)=\bigoplus_{m\in\mathbb Z}P_{2\pi mi+k}(\beta).
$$
From~\eqref{e:realtor} we then have
\begin{equation}
\label{eq:mult2}   m ( \alpha/ h , k/h ) = \sum_{m\in\mathbb Z} m_h ( \alpha,2\pi i mh+ k )  , 
\end{equation}
where the right hand side was defined in \eqref{eq:defmk}.

We will now relate the multiplicity defined by \eqref{e:m-h-def} to the multiplicity of zeros of an eigenvalue equation for $ \mathcal M_{ z_0 +1  \leftarrow z_0 } $. We note that for fixed $ h $, the holomorphic dependence on 
the parameter $ \alpha $ shows that 
\[  \mathbb C \ni \alpha \mapsto \mathcal M_{ z_0 +1 \leftarrow z_0  }  ( \alpha, h ) \ \text{ is holomorphic.} \]

\begin{prop}
\label{p:multi}
Fix $ k \in \mathbb C $ and $ h > 0 $. For arbitrary $z_0\in\mathbb C$ 
\[   \alpha \mapsto P_k ( \alpha, h )^{-1}  \text{ is meromorphic }  \Longleftrightarrow \ 
\det ( \mathcal M_{z_0 +1  \leftarrow z_0 } ( \alpha, h ) - e^{ i k /h } ) \not \equiv 0 . \]
If the determinant on the right hand side is not identically $ 0 $, then 
\begin{equation}
\label{eq:m2M}   m_h ( \alpha, k ) = \frac{1}{ 2 \pi i } \tr_{ \mathbb C^2 } \oint_\alpha 
( \mathcal M_{z_0 +1 \leftarrow z_0} ( \beta , h ) -e^{ik/h} )^{-1} \partial_\beta \mathcal M_{z_0 + 1 \leftarrow z_0} ( \beta , h ) d\beta ,
\end{equation}
where the integral is as in \eqref{e:m-h-def}. In other words, the multiplicity $ m_h ( \alpha, k ) $ is
the same as the multiplicity of the zero of the eigenvalue equation
\[  \beta \mapsto  \det (  \mathcal M_{z_0 + 1\leftarrow  z_0 } ( \beta , h ) -e^{ik/h} ) \]  at $ \beta = \alpha $. 
\end{prop}
\begin{proof}
1. We start by observing that since $ W ( z + 1 ) = W ( z ) $, $ \mathcal M_{ z_0  \leftarrow z } = 
\mathcal M_{ z_0 + 1\leftarrow z + 1 } $. Hence 
\begin{equation}
\label{eq:z02z} 
\mathcal M_{ z+ 1 \leftarrow z  } = \mathcal M_{ z + 1  \leftarrow z_0 +1  } \mathcal M_{ z_0 +1 \leftarrow  z_0  } \mathcal M_{ z_0 \leftarrow  z } =  \mathcal M_{ z  \leftarrow z_0   } \mathcal M_{ z_0 +1 \leftarrow  z_0  } \mathcal M_{ z \leftarrow  z_0 }^{-1}
\end{equation} 
and the statements for a fixed $ z_0 $ are equivalent to statements for all $ z $. 

The obstruction to the meromorphy of $ P_k ( \alpha, h )^{-1} $ is having a non-trivial solution $ w ( \alpha ) \in H_k^2 $
to $ P_k ( \alpha, h ) w ( \alpha ) = 0 $ for {\em all} $ \alpha $. As discussed in \S \ref{s:toy-basic},  $x \mapsto w ( \alpha, x ) $ extends to an entire function which then satisfies \eqref{e:V-k-h} and 
$ P ( \alpha, h ) w = 0 $. Definition of 
$ \mathcal M_{ z + 1 \leftarrow z  } ( \alpha, h ) $ in \S \ref{s:sfo}, then shows that $ e^{ ik/h  } $ is an eigenvalue
of $ \mathcal M_{ z + 1 \leftarrow z  } $ for all $ z $. That gives $ \Leftarrow $ in the first statement. To see
$ \Rightarrow $ we note that meromorphy of $ P_k ( \alpha, h )^{-1} $ implies existence of $ \alpha_0 $
for which $ \ker P_k ( \alpha_0 , h ) = \{ 0 \} $. If $ e^{ i k/h } $ were an eigenvalue of $ \mathcal M_{ z_0 + 1 \leftarrow  z_0  } ( \alpha_0 , h ) $ then, by \eqref{eq:z02z} it would be an eigenvalue of $ \mathcal M_{ z +1  \leftarrow  z } ( \alpha_0, h ) $ for all $ z $ 
and we would have have an entire solution to $ P ( \alpha_0 , h ) u = 0 $ satisfying \eqref{e:V-k-h}. But that 
means that $ \ker P_k ( \alpha, h ) $ is nontrivial.

The formula \eqref{eq:m2M} makes the above argument quantitative and it is easier to obtain 
in the setting of operators acting on periodic function:
\[  \widetilde P_k (\alpha, h ) := e^{ - i x k/h } P_k ( \alpha, h ) e^{ i x k /h } : 
H^2 ( \mathbb R /  \mathbb Z ) \to L^2 ( \mathbb R/\mathbb Z) .\]
The related operator acting on holomorphic functions is given by 
$ \widetilde P_k ( \alpha, h ) = ( h D_z + k  )^2 - \alpha^2 W^2 $. The corresponding transition matrix,  $  \mathcal M^k_{ z_1 \leftarrow z_0 } $, is defined as in \eqref{e:transition-def} but with \eqref{e:vec-u-def} replaced by
$ \mathbf u(z)=(u(z),( hD_z +  k ) u(z)) $. It 
is related to the original transition matrix
as follows:
\begin{equation}
\label{eq:M2Mt}   \mathcal M^k_{z_1 \leftarrow z_0 } ( \alpha, h ) = e^{  i k( z_0 - z_1)/h  }  \mathcal M_{ z_1 \leftarrow  z_0 } 
( \alpha, h ) . \end{equation}
The final reduction is to a system:
\begin{equation}
\label{eq:defDka}  D_k ( \alpha , h ) := \begin{pmatrix} h D_z + k & - 1 \\
- \alpha^2 W(z)^2 & h D_z +  k \end{pmatrix}, \ \ \ 
D_k ( \alpha , h ) \mathcal M^k_{ z \leftarrow  z_0  } ( \alpha, h ) = 0 ,  
\end{equation}
and $ \mathcal M^k_{z_0 \leftarrow z_0 } ( \alpha, h )  = I_{\mathbb C^2}$. 

Proceeding as in \eqref{eq:P2D2} we see that 
in the definition of $ m_h ( \alpha , k ) $ we can replace $ P_k ( \alpha, h ) $ by 
$D_k ( \alpha, h ) $ (now acting on periodic $ \mathbb C^2 $-valued functions) and in view of
\eqref{eq:z02z}  and \eqref{eq:M2Mt}, the equality \eqref{eq:m2M}
is equivalent to 
\begin{equation}
\label{eq:eqm}
\begin{split}
&  \frac{1}{2 \pi i } \tr_{ L^2 ( \mathbb R/\mathbb Z ; \mathbb C^2 ) } \oint_{ \alpha} D_k ( \beta, h )^{-1} \partial_\beta D_k ( \beta , h ) d \beta =  \\
& \ \ \ \ \ \frac{1}{ 2 \pi i } \tr_{ \mathbb C^2 } \oint_\alpha 
( \mathcal M^k_{1\leftarrow 0} ( \beta , h ) - I )^{-1} \partial_\beta \mathcal M^k_{1\leftarrow 0} ( \beta , h ) d\beta .
\end{split}
\end{equation}

\noindent 2. A systematic derivation of the formula~\eqref{eq:eqm} follows for setting up a Grushin problem (see
\cite[\S C.1]{res}) 
for $ D_k ( \alpha, h ) $ following \cite[(3.12)]{ela}. 
To pose it we introduce 
\[ \chi \in C^\infty( \mathbb R/\mathbb Z ; [ 0 , 1 ] ) , \ \ \
 \mathbb R/\mathbb Z \simeq ( - \tfrac12, \tfrac12 ] , 
 \ \ \ \supp \chi = [ - \tfrac14, \tfrac12 ] , \ \ \
\chi |  _{ [ -\frac18, \frac14 ] } = 1. \]
We write $ \chi' = \chi'_+ + \chi_-' $ where $ \chi_+' $ is supported near $ - \frac14 $ and 
$ \chi_-' $ is supported near $ \pm \frac12 $. 
The following operators are also well defined:
\begin{equation}
\label{eq:defRpm} 
\begin{gathered}   R_- : \mathbb C^2 \to L^2 ( \mathbb R/\mathbb Z;\mathbb C^2 ) , \ \ \ ( R_- u_- ) ( x ) := \chi'_- \mathcal M^k_{ x \leftarrow 0  }u_- , \\
R_+ : H^1 ( \mathbb R/\mathbb Z ;\mathbb C^2)  \to \mathbb C^2, \ \ \  R_+ \mathbf u := \mathbf u ( 0 ) . 
\end{gathered}
\end{equation}
We then claim that 
\begin{equation}
\label{eq:GrDk}   \mathscr P := \begin{pmatrix} (i/h ) D_k ( \alpha, h )  & R_- \\
R_+ & 0 \end{pmatrix} : H^1 ( \mathbb R/\mathbb Z ; \mathbb C^2 ) \oplus \mathbb C^2 \to 
L^2 ( \mathbb R /\mathbb Z ; \mathbb C^2 ) \oplus \mathbb C^2 , \end{equation}
is invertible with the inverse given by 
\begin{equation}
\label{eq:defEP} 
\mathscr P^{-1} = \mathscr E := \begin{pmatrix} E & E_- \\
E_+ & E_{-+} \end{pmatrix} , \ \ \  E_{-+} =  \mathcal M^k_{ 0 \leftarrow 1}-I_{\mathbb C^2 }  ( \alpha, h ) . 
\end{equation}
The conclusion \eqref{eq:eqm}, and hence \eqref{eq:m2M}, then follow from 
\cite[Proposition 4.1]{ela}.

\noindent 3. To show invertibility of \eqref{eq:defRpm} it is enough to construct an inverse. The operator
$ D_k ( \alpha, h ) $ is a Fredholm operator and as 
\[  \mathscr P^t_k ( \alpha, h )   := \begin{pmatrix} (i/h ) D_k ( \alpha, h )  & (1-t )R_- \\
(1-t ) R_+ & t I_{\mathbb C^2}  \end{pmatrix} : H^1 ( \mathbb R/\mathbb Z ; \mathbb C^2 ) \oplus \mathbb C^2 \to 
L^2 ( \mathbb R /\mathbb Z ; \mathbb C^2 ) \oplus \mathbb C^2 , \]
are finite rank perturbations of a Fredholm operator (with $ D_k ( \alpha, h ) $ as the only non-zero entry), they are also Fredholm operators (see for instance 
\cite[\S 2.5]{HNotes} for basic properties of those operators). Since $ D_k ( 0, k ) $ is invertible for 
$ k/h \notin 2 \pi  \mathbb Z $, it follows that $ \mathscr P^k_1 $ is invertible. Hence continuity of the index
shows that the index of $ \mathscr P $ is $ 0 $ and the right inverse is automatically the left inverse.

We construct the right inverse explicitly as follows. We first consider the problem
\begin{equation} 
\label{eq:Gr1}   ( i/h) D_k  \mathbf u ( x)  + \chi'_- ( x)  \mathcal M^k _{ x\leftarrow 0 }  u_- = 0, \ \ \  \mathbf u ( 0 ) = v_+ \in \mathbb C^2 . \end{equation} 
We can define multivalued matrix valued functions on $ \mathbb R / \mathbb Z $ as follows:
\begin{equation} 
\label{eq:defIpm}   I_\pm  ( [x ] ) =  \mathcal M^k_{ x \leftarrow 0  } , \ \  \pm x \geq - \tfrac14 , \end{equation}
where $ x \mapsto [ x ] $ is the map
$ \mathbb R \to \mathbb R/\mathbb Z $. We can think of $ I_+ ( x )  $ as solving the equation forward from  $ 0 $ and
$ I_- ( x ) $  backward from $ 0$. The function  
$ \chi ( x ) I_+ ( x )  v_+ , ( 1 - \chi ( x ) )  I_- ( x ) v_+ $ are now well defined as functions on the circle 
(we drop $ [ x ] $ notation when there is no ambiguity). We can then put
\[ \mathbf u ( x ) = \chi ( x ) I_+ ( x )  v_+  + ( 1 - \chi ( x ) )  I_- ( x ) v_+ \in H^1 ( \mathbb R/\mathbb Z; \mathbb C^2 ) , \]
so that
\[  \begin{split} (i /h ) D_k \mathbf u & = (i/h) [ D_k , \chi ] I_+v_+ - (i/h ) [ D_k, \chi ] I_-  v = 
\chi' I_+  v_+ - \chi' I_-  v_+ \\
& = \chi'_- ( x )  ( I_+ ( x ) v_+  - I_- ( x ) v_+ ), \end{split} \]
where we used the fact that $ I_+ ( x ) v_+ = I_- ( x ) v_+ $ on the support of $ \chi_+' $ (near $ 0 $). On the other
hand for $ [ x ] $ on the support of $ \chi_-' $, say for $ x $ near $ [ \frac12 ] $, we have 
\[  \begin{split} \chi'_- ( [ x ] )  I_+ ( [ x ]  )  & = \chi'_- ( x ) \mathcal M^k_{ x \leftarrow 0 } = 
\chi_-' ( x) \mathcal M^k_{ x - 1 \leftarrow -1 }  = \chi_- ' ( x ) \mathcal M^k_{ x - 1 \leftarrow 0 } \mathcal 
M^k_{ 0 \leftarrow -1 } \\ & = \chi_- ' ( x ) I_- ( [ x ] ) \mathcal M^k_{ 1 \leftarrow 0 } .
\end{split}  \]
 Hence \eqref{eq:Gr1} is solved by putting
\begin{equation}
\label{eq:Ep}
\mathbf u = E_+ v_+ := ( \chi I_+ + ( 1- \chi ) I_- ) v_+, \ \ \ u_- = E_{-+} v_+ := ( \mathcal M^k_{ 0 \leftarrow 1}-I_{\mathbb C^2 }  ) v_+ .
\end{equation}
It remains to solve 
\[  ( i/h) D_k  \mathbf u ( x)  + \chi'_- ( x)  \mathcal M^k _{ x \leftarrow 0 }  u_- = v \in L^2 ( \mathbb 
R/ \mathbb Z; \mathbb C^2 ) , \ \ \  \mathbf u ( 0 ) = 0 , \]
and that is done similarly using the forward and backward solutions to the non-homogeneous problem. Since all that matters for us is the formula for $ E_{-+} $ in \eqref{eq:Ep}, we leave the details to the reader.  
This completes the proof of \eqref{eq:defEP} and as explained there of \eqref{eq:m2M}. 
\end{proof}


\subsection{Quantisation condition}
\label{s:qc}

Now, we assume that $\gamma$ is a Stokes loop and $ \Im \alpha \geq 0 $. We note that this assumption can be relaxed to $ \Im \alpha \geq - h^M $ for some (large) $ M $.

We will then
put  $z_0:=\gamma(0)$ in Proposition \ref{p:multi}, so that $\mathcal M_{z_0 + 1 \leftarrow z_0}=\mathcal M_{\gamma(1)\leftarrow  \gamma(0)}$.
Putting
\[
Z_h(\alpha): = \tfrac12 \tr\mathcal M_{\gamma(1)\leftarrow  \gamma(0)}( \alpha, h ) ,
\]
and using~\eqref{e:transition-det}, the eigenvalue equation can be written as 
\begin{equation}
  \label{e:eigeq}
0 = \tfrac12 e^{ - i k/h } \det (  \mathcal M_{z_0 + 1 \leftarrow z_0} ( \alpha , h ) -e^{ik/h} )  = \cos ( k/h )  -  Z_h(\alpha),
\end{equation}
and we want to find $ \alpha$'s for which  the right hand side vanishes.

Recalling~\eqref{e:psi-def}, we note that
\[
\psi_{\gamma,\alpha}(1)-\psi_{\gamma,\alpha}(0)=c_0\Im\alpha\quad\text{where }
c_0:=\int_0^1 W(z)\,dz>0.
\]
Now, by Lemma~\ref{l:transfer-approx} (recalling~\eqref{e:stokes-verified})
\[
\mathcal M_{\gamma(1)\leftarrow  \gamma(0)}=B_N(\gamma(1);\alpha,h)^{-1}
S_{\gamma,N}(\alpha,h)B_N(\gamma(0);\alpha,h)+\mathcal O_{\gamma,N}(h^{N+1}e^{c_0\Im\alpha/h}).
\]
By Lemma~\ref{l:WKB-periodic} we have
\[
B_N(\gamma(1);\alpha,h)=B_N(\gamma(0);\alpha,h),\quad
S_{\gamma,N}(\alpha,h)
=\begin{pmatrix}
e^{-i\alpha h^{-1} c(h/\alpha)} & 0 \\
0 & e^{i\alpha h^{-1}c(h/\alpha)}
\end{pmatrix}.
\]
Therefore, 
\begin{equation}
\label{eq:Zh2cos}
Z_h(\alpha)=\cos(\alpha h^{-1}c(h/\alpha))+R_h(\alpha),\quad
R_h(\alpha)=\mathcal O_{\gamma,N}(h^{N+1}e^{c_0\Im\alpha/h}).
\end{equation}
One observation is that for the eigenvalue equation to be solvable, we need
\begin{equation}
\label{eq:imagk}
|\Im k|\leq c_0|\Im\alpha| +C h.
\end{equation}
We can also assume that $ |\Re k | \leq \pi h $ since the spaces
$\mathcal V_{k+2\pi h,h}$ and $\mathcal V_{k,h}$ defined in~\eqref{e:V-k-h}
are the same.

We first consider the simplified equation obtained from \eqref{e:eigeq} and \eqref{eq:Zh2cos}:
\begin{equation}
\label{eq:cos2cos}  F_0 ( \alpha, k , h ) :=   \cos ( c ( h/\alpha )  \alpha / h) - \cos ( k/h ) = 0  \end{equation}
which holds when  
\[  c (h /\alpha )  \alpha = 2 \pi n h \pm k  , \ \ \ \  n \in \mathbb Z , \]
 noting that the roots are double
when $ k/h \in \pi  \mathbb Z $. The solutions 
satisfying \eqref{e:alpha-ass} 
have asymptotic expansions (where we use the fact that
$ c_1 = 0 $ in \eqref{e:WKB-periodic}) 
\begin{equation}
\label{eq:alpm}  \begin{split}  \alpha_\pm ( n , k , h ) & = c_0^{-1} ( 2 \pi n h \pm k ) + d_2 ( 2 \pi n h \pm k )^{-1} h^2  \\ 
& \ \ \ \ + d_3 
( 2 \pi n h \pm k )^{-2 } 
h^3 + \cdots  , \ \ \ d_2 = - c_2, \cdots  \end{split} 
\end{equation}

We can now compare these solutions to the solutions of the equation \eqref{e:eigeq},
for  $ \Im \alpha \geq - h^M $, using \eqref{eq:Zh2cos} and Rouch\'e's theorem
 applied to a contour consisting of $h^N$-sized disks centered at $\alpha_\pm(n,k,h)$. The symmetry $ \alpha \to 
- \alpha $ then gives the result for $ \Im \alpha \leq 0 $ as well. 
To be consistent with this section, in comparison with Theorem \ref{t:3} we assume that $ \int_\gamma W ( z ) dz > 0$, that is $ \arg W_0 = 0 $.
Rescaling, $ \alpha $ this can be removed.

\begin{theo}
\label{t:4}
Under the assumptions of Theorem \ref{t:3} and when $ \arg W_0 = 0 $,
there exist a contant $ C_0 $ and maps
\[  \widetilde \alpha_\pm  :  \mathbb Z \times \mathbb C    \to \mathbb C , \]
such that with $ \alpha_\pm ( n, k ) := \alpha_\pm ( n, k ,1 ) $ (see \eqref{eq:alpm}),  
\[   \widetilde \alpha_\pm  ( n , k ) - \alpha_\pm ( n , k  ) = \mathcal O (  | 2 \pi n \pm k |^{-\infty } ) , \]
and for $ | \alpha |  \geq C_0 $, in the notation of
\eqref{eq:defmk}, 
\[ m ( \alpha, k ) = | \{ (n,  \varepsilon  ) : \widetilde \alpha_{ \varepsilon } ( n, k ) =  \alpha , \ \varepsilon = \pm 1  \} |. \] 
\end{theo}

\medskip\noindent {\sc Acknowledgements.}
We would like to thank Zhongkai Tao for help with the material in~\S \ref{s:rank2}, John Lott for the reference~\cite{daly}, and 
Zhenhao Li and Mengxuan Yang for comments on the first version of this paper. We are also grateful to the anonymous referee for many constructive suggestions. 

The authors were partially supported by the 
Simons Targeted Grant Award No. 896630 ``Moir\'e Materials Magic", while 
 the first named author was partially supported by NSF grant DMS-2400090 and a Simons Fellowship.
 Part of this work was completed while the first named author was in residence at the Mathematical Sciences Research Institute in Berkeley, California in Fall 2024, supported by NSF grant  DMS-1928930.

\end{document}